\documentclass[pra,twocolumn,superscriptaddress,nofootinbib]{revtex4}
\usepackage[letterpaper, left=0.8in, right=0.8in, top=0.9in, bottom=0.9in]{geometry}
\usepackage{titlesec}

\usepackage{cmap} 
\usepackage[utf8]{inputenc}
\usepackage[english]{babel}
\usepackage[T1]{fontenc}
\usepackage{amsmath}
\usepackage{amsfonts}
\usepackage{bm}
\usepackage{xcolor}
\definecolor{blueviolet}{rgb}{0.2, 0.2, 0.6}
\definecolor{webgreen}{rgb}{0,.5,0}
\definecolor{webbrown}{rgb}{.6,0,0}
\usepackage[pdftex,
	bookmarks=false,
	colorlinks=true, 
	urlcolor=webbrown, 
	linkcolor=blueviolet, 
	citecolor=webgreen,
	pdfstartpage=1,
	pdfstartview={FitH},  
	bookmarksopen=false
	]{hyperref}
\usepackage{tikz}
\usepackage{natbib}
\usepackage{natbib}
\usepackage{bibunits}

\defaultbibliographystyle{apsrev4-1_with_title}
\defaultbibliography{references}
\usepackage{mathtools}
\usepackage{graphicx}
\usepackage{csquotes}
\usepackage{braket}
\usepackage{dsfont}
\usepackage{listings}
\usepackage{comment}
\allowdisplaybreaks

\usepackage{multirow}
\usepackage{amssymb}

\usepackage{amsthm}

\usepackage{color}
\usepackage{nicefrac}

\mathtoolsset{showonlyrefs}     

\DeclareFixedFont{\ttb}{T1}{txtt}{bx}{n}{9} 
\DeclareFixedFont{\ttm}{T1}{txtt}{m}{n}{9}  

\usepackage{color}
\definecolor{deepblue}{rgb}{0,0,0.5}
\definecolor{deepred}{rgb}{0.6,0,0}
\definecolor{deepgreen}{rgb}{0,0.5,0}

\newcommand\pythonstyle{\lstset{
language=Python,
basicstyle=\ttm,
morekeywords={self},              
keywordstyle=\ttb\color{deepblue},
emph={MyClass,__init__},          
emphstyle=\ttb\color{deepred},    
stringstyle=\color{deepgreen},
frame=tb,                         
showstringspaces=false
}}

\lstnewenvironment{python}[1][]
{
\pythonstyle
\lstset{#1}
}
{}

\newcommand\pythoninline[1]{{\pythonstyle\lstinline!#1!}}

\definecolor{orange}{RGB}{255,127,0}

\def\bra#1{\ensuremath{\mathinner{\langle{#1}|}}}
\def\ket#1{\ensuremath{\mathinner{|{#1}\rangle}}}

\newcommand{\ketbra}[2]{\lvert #1 \rangle \! \langle #2 \rvert}

\newcommand{\disp}{{\hat{D}}}
\newcommand{\xh}{{\hat{x}}}
\newcommand{\ph}{{\hat{p}}}

\newcommand{\R}{\mathbb{R}}
\newcommand{\C}{\mathbb{C}}
\newcommand{\Lip}{\mathrm{Lip}}

\newcommand{\indicator}{\mathds{1}}

\newtheorem{example}{Example}
\newtheorem{proposition}{Proposition}
\DeclareMathOperator{\tr}{tr}
\DeclareMathOperator*{\E}{{\mathbb{E}}}
\DeclareMathOperator*{\Var}{\mathrm{Var}}

\newtheorem{theorem}{Theorem}
\newtheorem{corollary}{Corollary}
\newtheorem{definition}{Definition}

\newtheorem{lemma}{Lemma}
\newtheorem{fact}{Fact}

\newtheorem{remark}{Remark}
\newtheorem{algorithmsimp}{Algorithm}

\usepackage[noend]{algpseudocode}
\usepackage{algorithm,algorithmicx}

\algrenewcommand\alglinenumber[1]{\sf\scriptsize\color{blue}{#1}}
\algrenewcommand\algorithmicrequire{\textbf{Input:}}
\algrenewcommand\algorithmicensure{\textbf{Output:}}

\begin{document}

\title{Quantum Advantage for Sensing Properties of Classical Fields}

\author{Jordan Cotler}
\affiliation{Department of Physics, Harvard University, Cambridge, MA 02138 USA}
\affiliation{Harvard Quantum Initiative, Harvard University, Cambridge, MA 02138 USA}

\author{Daine L. Danielson}
\affiliation{Department of Physics, Harvard University, Cambridge, MA 02138 USA}
\affiliation{Black Hole Initiative, Harvard University, Cambridge, MA 02138 USA}
\affiliation{Center for Theoretical Physics, Massachusetts Institute of Technology, Cambridge, MA 02139 USA}

\author{Ishaan Kannan}
\affiliation{Harvard Quantum Initiative, Harvard University, Cambridge, MA 02138 USA}

\date{\today}

\begin{abstract}
Modern precision experiments often probe unknown classical fields with bosonic sensors in quantum-noise–limited regimes where vacuum fluctuations limit conventional readout. We introduce Quantum Signal Learning (QSL), a sensing framework that extends metrology to a broader property-learning setting, and propose a quantum-enhanced protocol that simultaneously estimates many properties of a classical signal with shot noise suppressed below the vacuum level. Our scheme requires only two-mode squeezing, passive optics, and static homodyne measurements, and enables post-hoc classical estimation of many properties from the same experimental dataset. We prove that our protocol enables a  quantum speedup for common classical sensing tasks, including measuring electromagnetic correlations, real-time feedback control of interferometric cavities, and Fourier-domain matched filtering. To establish these separations, we introduce an optimal-transport conditioning method, and show both worst-case exponential separations from all entanglement-free strategies and practical speedups over homodyne and heterodyne baselines. We further show that when squeezing is treated as a resource, a protocol with squeezed light can sense a structured classical background exponentially faster than any coherent classical probe.

\end{abstract}

\maketitle

\section{Introduction}
\vspace{-2mm}
 
From gravitational-wave interferometers and microwave cavities to optical and circuit-QED transducers, many leading precision platforms now operate in regimes limited not by instrumental noise but by quantum fluctuations of the electromagnetic field \cite{Aasi2015AdvancedLIGO,Tse2019QuantumEnhancedLIGO,Brubaker2017HAYSTAC,Backes2021QuantumEnhancedAxion,Macklin2015JTWPA,Purdy2013RPShotNoise, kannan2026exponentialquantumadvantagelearning}. In such experiments the signal is often a classical field, for example ambient radiation, a fluctuating force, or a time-dependent control imperfection, coupled linearly to a bosonic sensor. The resulting inference tasks are rarely captured by the textbook single-parameter model $U_\theta(T) = e^{-i\theta G T}$ with known generator $G$. Instead one typically cares about structured properties of an uncertain environment, such as quadrature correlations in an electromagnetic transducer \cite{Malnou2019SqueezedVacuumAxion,Gould2021EntangledWitnessQNC,Qiu2023BroadbandSqueezedMicrowaves,Vaartjes2024StrongMicrowaveSqueezing}, phase-space mixing in an interferometric cavity \cite{Oelker2016AudioBandFDS,McCuller2020FDSAdvLIGO,Ganapathy2023BroadbandFDSSqueezing,Jia2024SQLScience}, or the ability to score many candidate waveforms using post-hoc template-bank queries \cite{Allen2005ChiSquare,Allen2012FINDCHIRP,Babak2013IHOPE,Usman2016PyCBC,Messick2017GstLAL,
DalCantonHarry2017O2TemplateBank,Roulet2019GeometricTemplateBank,Mukherjee2021GstLALTemplateBank}. These are distributional questions, since uncertainty in the field and in prior information induces a nontrivial distribution over sensor responses, and the object of interest is a functional of that distribution rather than a single parameter.

\begin{figure*}[t!]
\setlength{\abovecaptionskip}{2pt}
\setlength{\belowcaptionskip}{-4pt}

  \includegraphics[width=\textwidth]{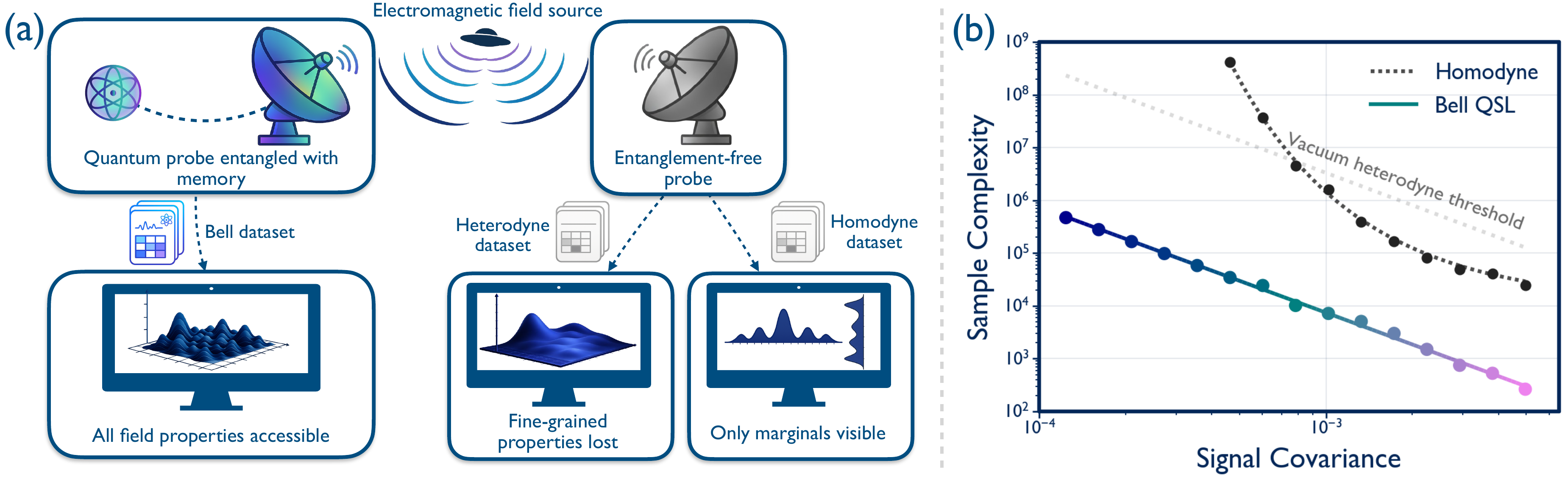}
 \caption{(a) \textit{Quantum-enhanced signal learning vs.~conventional electromagnetic (EM) sensing.} Bell QSL leverages probes entangled with quantum memory to collect a classical record that is used for quantum-limited, informationally-complete estimation of the EM field distribution. Heterodyne measurements suffer a vacuum noise penalty, blurring fine-grained features, while homodyne measurements only resolve marginals that may not contain higher-dimensional properties. (b) \textit{Quantum advantage for estimating correlations in an electromagnetic field.} Minimum sample complexity required to estimate a Gaussian EM field covariance $c$ in the shot-noise-limited regime, as in \eqref{eq:Sigma_c}. At $r=2$ and $|\sigma_x^2 - \sigma_p^2| \ll c$, in the presence of white noise with amplitude $0.2 c$, Bell QSL scales polynomially and orders of magnitude below the vacuum noise limit, achieving quantum advantage over both homodyne and heterodyne tomography.}
\label{fig:em_field}
\vspace{-3pt}
\end{figure*}

A basic obstacle is that the standard measurement primitives do not simultaneously deliver quantum-limited precision and broad learnability. Heterodyne detection, and more generally entanglement-free joint quadrature readout, is informationally complete and supports post-hoc analysis, but it carries an irreducible vacuum noise penalty that cannot be circumvented with probe squeezing \cite{Serafini2017QuantumContinuousVariables}. Homodyne measurement can resolve sub-vacuum fluctuations along one quadrature and is therefore well matched to carefully aligned tasks, but it only reveals one-dimensional marginals of the underlying phase-space response. In practice, rapidly and reliably preparing squeezed probes aligned with many homodyne quadrature angles is technically costly, so the available baseline is often a small fixed set of angles. This creates a dilemma in quantum-noise-limited sensing. One can either collect flexible data at the vacuum noise floor, or sharper, quantum-limited data that cannot access intricate phase-space structure.

We show that entanglement resolves this dilemma and yields a quantum advantage for sensing classical fields. Our protocol prepares a two-mode squeezed vacuum (TMSV) between a sensing mode and an idler, applies the unknown field-induced channel to the sensing mode, and performs a standard continuous-variable (CV) Bell measurement of commuting EPR quadratures \cite{BraunsteinKimble1998,Furusawa1998, BraunsteinKimble2000DenseCodingCV, Serafini2017QuantumContinuousVariables, jiang2024bellfix}. Operationally, the classical field induces an unknown phase-space displacement $\alpha$ on the sensing mode, and the Bell measurement produces a complex outcome $\zeta$ that is a blurred snapshot of this displacement. $\zeta$ satisfies an additive-noise model $\zeta = \alpha + Z$, where $Z$ is an independent complex Gaussian with law $\mathcal{N}_{\mathbb{C}}(0,e^{-2r})$ set by the finite EPR correlations of the two-mode squeezed state. Thus each channel use yields simultaneous estimates of both quadratures, and the noise in both estimates is suppressed below the vacuum level by a factor $e^{-2r}$. This should be contrasted with heterodyne, which yields the same two-quadrature structure but with an irreducible vacuum variance. The resulting dataset is informationally complete and supports post-hoc learning, since many properties of the underlying signal can be estimated from the same measurement record without rerunning the experiment.

We formalize these tasks via a learning framework we call Quantum Signal Learning (QSL), which models classical-field sensing as estimating properties of the phase-space response induced by a linear bosonic coupling. Within this framework we give estimator families for a wide class of physical properties and prove sample-complexity guarantees that improve with squeezing. We then establish quantum advantage in two complementary senses. First, in worst-case regimes relevant to template-bank scoring, our Bell QSL protocol enables exponential improvements for Fourier-domain matched-filter queries over any entanglement-free data collection strategy. Second, motivated by realistic laboratory constraints, we introduce an optimal-transport-based conditioning framework that quantifies when restricted-angle homodyne becomes ill-conditioned for learning genuinely two-dimensional phase-space features, while Bell QSL remains efficient. We illustrate these advantages in experimentally commonplace sensing problems motivated by correlated classical electromagnetic fields and phase rotation noise in interferometric cavities.

These separations follow the bosonic-metrology premise that probe energy governs sensing precision, placing squeezing and displacement on equal footing \cite{Safranek_2016} and treating entanglement as the advantage-generating resource.  While we argue for the importance of this interpretation, we also demonstrate that unentangled, squeezed-probe strategies can sense a fluctuating classical field exponentially faster than any coherent-state classical probe. In complementary work, \cite{Robert_unreleased} studies the unique power of squeezed light for sensing time-varying signals, achieving exponential gains over any squeezing-free classical protocol. Squeezing and entanglement thus emerge as distinct sources of quantum advantage in classical-field sensing.

\vspace{-1em}
\section{Quantum Signal Learning}
\vspace{-1em}

An $n$-mode bosonic signal Hamiltonian consists of linear terms in the quadratures. This Hamiltonian circumscribes most classical-field sensing tasks, where quantum-classical interactions usually enter as c-number couplings to the canonical quadratures,
\begin{equation}
H(t) = \sum_{j=1}^n f_{x, j}(t) \xh_j + f_{p, j}(t) \ph_j.
\end{equation}
 QSL begins with the observation that the time-evolution channel induced by evolving $n$ bosonic probes under $H(t)$ for total time $T$ can be described by the following quantum channel, where $\disp(\alpha)$ is the displacement operator:
\begin{equation}
\mathcal{E}_H^T(\rho) = \int d^{2n}\alpha \ P_H^{(T)}(\alpha)\disp(\alpha)\rho\disp(\alpha)^\dagger\ .
\end{equation} 
The linear structure of $H$ is realized by the channel consisting of only displacement terms. With this representation, any property $\Psi$ of the signal is encoded in the classical density $P_H^{(T)}$, which can be extracted by selecting a function $\psi(\alpha)$, called the \textit{property kernel}, and integrating it against $P_H^{(T)}$:
\begin{equation}
\Psi\!\left(P_H^{(T)}\right) = \int d^{2n}\alpha \  P_H^{(T)}(\alpha)\psi(\alpha) . 
\end{equation}
In this manner, classical field properties are linear functionals of the displacement distribution induced by the signal. The QSL problem follows from this observation.
\begin{definition}[Quantum Signal Learning] Given a linear Hamiltonian $H$ and an evolution time $T$, the Quantum Signal Learning problem \textnormal{\textsc{QSL}}$(H, T, \Psi, \epsilon, \delta)$ is to estimate the selected linear functional $\Psi(P_H^{(T)})$ to within absolute error $\epsilon$, with success probability at least $1-\delta$, given query access to $\mathcal{E}_{H}^{(T)}$. 
\end{definition}
\begin{figure}[t]
    \centering
    \includegraphics[width=\linewidth]{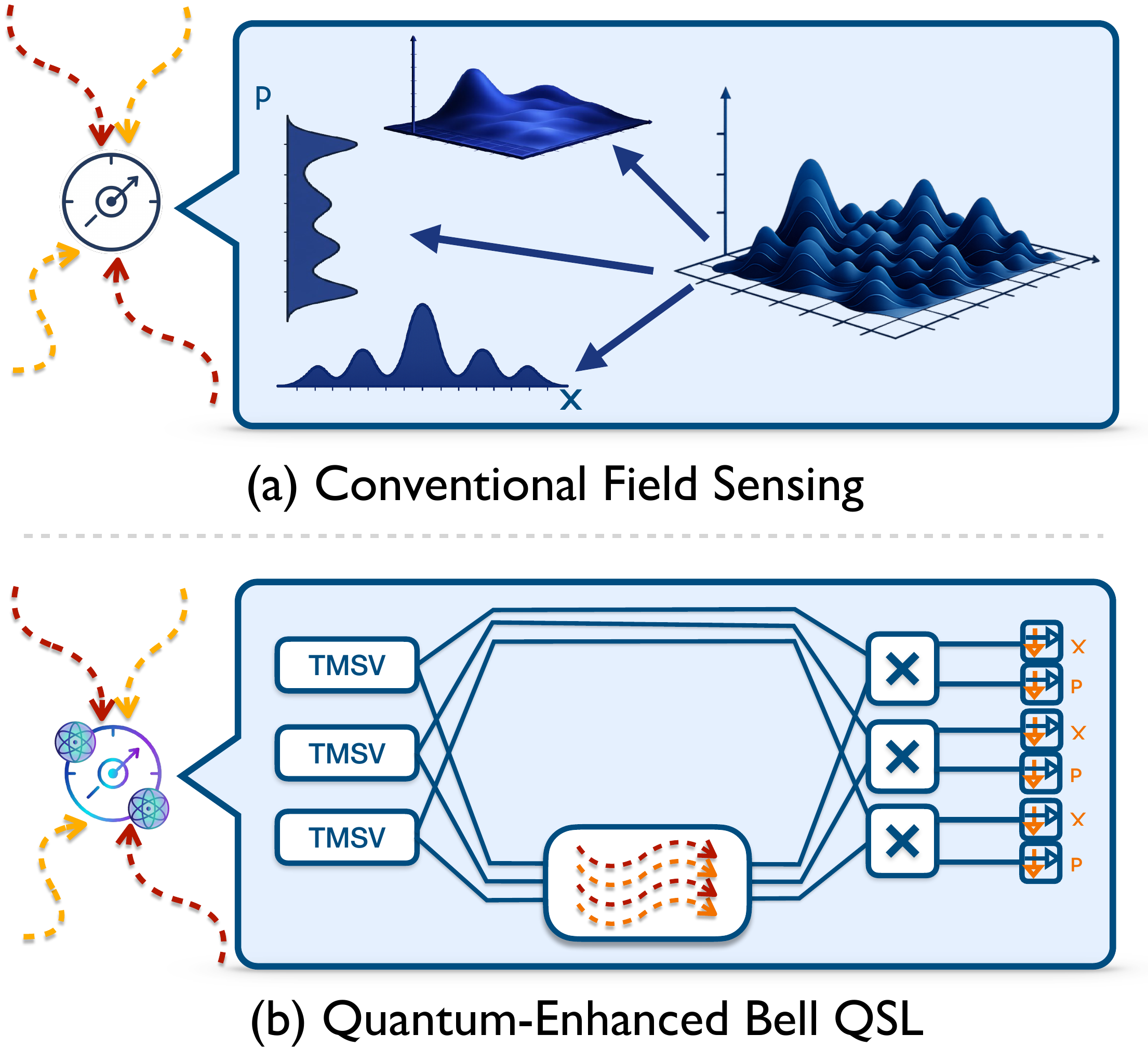}
    \caption{(a)\textit{ Conventional bosonic sensing of classical fields}. Homodyne measurements project high-dimensional field distributions onto chosen marginal axes, losing information, while heterodyne measurements blur sharp features with vacuum noise.  (b) \textit{Bell QSL}. Entangled two-mode squeezed vacuum probes are prepared, with half the modes acting as probes and the other half stored in quantum memory. Beamsplitters and homodyne measurements constitute CV Bell-basis measurement, enabling informationally complete, sharp property recovery.}
    \label{fig:algorithm}
\end{figure}
Several natural properties of classical fields can be viewed neatly in this framework. For instance, for a single mode, the kernel $\psi_{\mathrm{cov}} = xp = \Re(\alpha)\Im(\alpha)$ specifies a property $\Psi_{\mathrm{cov}}$ which is exactly the covariance of a (mean-zero) classical field. Alternatively, $\psi_{k}(\alpha) = e^{ik\alpha}$ induces that $\Psi_{k}$ is exactly the signal's Fourier coefficient of frequency $k\in\mathbb{C}^n$. QSL thus subsumes the canonical single-parameter sensing setting, where distributions are delta functions, and encompasses a wide array of multiparameter fields, stochastic signals and noise sources. Moreover, beyond linearity, which holds for all relevant classical sources, QSL does not assume prior knowledge of the signal's structure as is otherwise necessary in the canonical metrology setting. As we will see, the learning-theoretic QSL reformulation of quantum-enhanced sensing opens up the possibility for quantum speedups not present in the usual metrological picture.

\vspace{-1em}
\subsubsection*{The Bell QSL Algorithm}
\vspace{-1em}

Now we provide an efficient, entanglement-enabled algorithm for the QSL problem. Our algorithm, called Bell Quantum Signal Learning (Bell QSL), has several practically useful properties. The algorithm collects a classical dataset that can be reused for arbitrary post-hoc property estimation as it is informationally complete. While our estimators provide rigorous guarantees for property estimation, the dataset is more feature-rich than entanglement-free counterparts, and can be used as input to machine learning pipelines (see Fig.~\ref{fig:unsupervised} later on). Experimentally, the algorithm only requires constant-depth quantum information processing in the form of two-mode entangled state preparation and passive beamsplitters. Moreover, it only utilizes measurements in a static homodyne basis, requiring no real-time control or re-locking of squeezing axes, and achieves optimal uncertainties in all quadrature readouts. 

Let $T_r$ denote the operator that maps any $f:\mathbb{C}^n\rightarrow \mathbb{C}$ to its squeezed-vacuum-blurred measurement distribution by convolution against a mean-zero Gaussian density with variance $e^{-2r}$. The algorithm, depicted in Figure~\ref{fig:algorithm}, proceeds in the following steps.

\begin{algorithmsimp}[Bell QSL] Given access to an $n$-mode linear Hamiltonian $H$ for time $T$, perform the following steps.
\begin{enumerate}
    \item Prepare $n$ two-mode squeezed vacuum states, and allow one bipartition to evolve under $H$ for time $T$ while storing the other half in quantum memory. Then, interfere each pair on a 50:50 beamsplitter, and for each two-mode pair $j$, measure quadratures $\xh_1^{(j)}, \ph_2^{(j)}$.

    \item Collect outcome $\zeta = \{\zeta_1, ...,\zeta_n\}$, with $\zeta^{(j)} = \xh_1^{(j)}+i \ph_2^{(j)}$. Repeat $N$ times.
    
    \item Given property kernel $\psi$, return the estimator
    \begin{equation} \label{eq:general_estimator_main_text}
        \hat{\Psi} = \frac{1}{N}\sum_{i=1}^N g_\psi(\zeta^{(i)}) \ ,
    \end{equation}
    where $g_\psi$ is the inverse of $\psi$ under $T_r$.
\end{enumerate}
\end{algorithmsimp}
In practice, $\hat{\Psi}$ can always be computed using the Fourier transform of $\psi$ when it is well defined, or by a particular polynomial transformation of measurement shots; these estimators are discussed in detail, with fully general sample complexity analyses, in Appendices \ref{sec:Fourier_alg} and \ref{sec:distribution_alg}.

The data-collection strategy of Bell QSL is a well-studied primitive across CV quantum optics and information. Originally established for CV teleportation and dense coding \cite{BennettWiesner1992, Bennett1993Teleportation}, recent works have used Bell-basis measurements to establish quantum speedups in certain sensing and learning settings \cite{jiang2024bellfix, Prabhu_2026}, with an experimental demonstration of multimode displacement channel learning in \cite{liu_brunel_2025}. Our central contributions lie in applying this  experimental primitive within the QSL framework to extract utility for a wide array of practical sensing tasks, and in developing novel theoretical methods to establish quantum speedups in these regimes. In the following sections, we analyze the performance of Bell QSL in multiple regimes, establishing that our algorithm can achieve practical quantum advantages over conventional benchmarks.

\vspace{-1em}
\section{Applications of Bell QSL} \label{sec:applications}
\vspace{-1em}
Here we demonstrate the utility of Bell QSL for shot-noise-limited field sensing, including characterizing correlated electromagnetic fields and real-time feedback control in interferometers. 

\vspace{-1em}
\subsection{Electromagnetic Field Sensing}
\label{sec:em_correlation}
\vspace{-1em}
A broad class of electromagnetic sensors, from microwave and circuit-QED resonators to narrowband optical cavities, can be modeled as a single bosonic mode driven by an incident classical field.  An $LC$ mode with canonical charge and flux operators $[\hat\Phi,\hat Q]=i\hbar$ has the free Hamiltonian $H_0=\hat Q^{\,2}/(2C)+\hat\Phi^{\,2}/(2L)$. An incident classical electromagnetic field induces an effective electromotive force and flux bias determined by the sensor geometry,
\begin{equation}
V_{\rm eff}(t)\equiv \oint_{\mathcal C}\mathbf E(\mathbf r,t)\cdot d\boldsymbol\ell, \quad \Phi_{\rm ext}(t)\equiv \int_{\mathcal S}\mathbf B(\mathbf r,t)\cdot d\mathbf A,
\end{equation}
so that (up to an irrelevant c-number) the driven Hamiltonian can be written as
\begin{equation}
H(t) = H_0 - V_{\rm eff}(t)\,\hat Q - I_{\rm eff}(t)\,\hat\Phi, \quad I_{\rm eff}(t):=\Phi_{\rm ext}(t)/L.
\label{eq:em_transducer_H}
\end{equation}
An unknown EM environment enters through two linear couplings to conjugate quadratures, one ``electric-like'' ($V_{\rm eff}\hat Q$) and one ``magnetic-like'' ($I_{\rm eff}\hat\Phi$). Treating $(V_{\rm eff}(t), I_{\rm eff}(t))$ as stochastic processes, natural when modeling uncertain signals, clutter, or fluctuating backgrounds, bounded-time evolution induces a displacement channel with endpoint distribution $P_H^{(T)}(\alpha)$. (The reduction to this channel, including the rotating-frame drive and explicit displacement amplitude, is given in App.~\ref{sec:prac_Hams}.)

In many common operating regimes (e.g.~narrowband sensing, aggregation over many microscopic field contributions, or approximately Gaussian priors), the induced endpoint displacement $\alpha$ is well modeled as a centered Gaussian
on phase space,
\begin{equation}
\alpha \sim \mathcal N_\mathbb{C}\big(0,\Sigma(c)\big),\quad \Sigma(c)=
\begin{pmatrix}
\sigma_x^2 & c\\
c & \sigma_p^2
\end{pmatrix},
\label{eq:Sigma_c}
\end{equation}
where $c$ encodes a \emph{joint} quadrature covariance of the incident field as seen by the transducer \cite{Menzel2010DualPath,RevModPhys.82.1155,Tsang2011,Mallet2011ItinerantSqueezed}. In general, electromagnetic dynamics and boundary conditions tie ``electric-like'' and ``magnetic-like'' degrees of freedom, and after transduction they appear precisely as correlations between conjugate quadratures.

In the shot-noise-limited regime, we are interested in estimating values of $c \ll 1$, below the vacuum-limited heterodyne variance. Moreover, squeezed homodyne tomography can only be executed along fixed measurement axes, because of the challenge of re-locking squeezing axes during a short sensing window. Even in this simple Gaussian-covariance estimation problem, one can see that e.g.~a pair of hypothesis covariances $P_{\pm}=\mathcal N\big(0,\Sigma(\pm c)\big)$ share identical $x$ and $p$ marginals, differing in the sign of the correlation. Common homodyne measurements at angles $\Theta=\{0,\pi/2\}$, which image these marginals up to squeezed-Gaussian smoothing, cannot distinguish these hypotheses.

This ``hard instance'' can be extended in a practically meaningful way by moving slightly away from exact isotropy.
Concretely, consider a near-isotropic family in which the marginal variances are almost equal
(e.g.\
$\sigma_x^2=\sigma^2+\Delta/2$ and $\sigma_p^2=\sigma^2-\Delta/2$ with $|\Delta|\ll \sigma^2$)
, so that the $x$- and
$p$-marginals are nearly insensitive to the correlation parameter even though the physically relevant information
resides in the joint quadrature correlation $c$.
The key observation is that while instances which are strictly \textit{impossible} for homodyne learners are
knife-edge phenomena, common learning tasks \textit{near} this isotropic worst-case will be very hard for
entanglement-free homodyne protocols, which access $c$ only through small symmetry-breaking effects, resulting in an
ill-conditioned inference problem. Meanwhile, Bell QSL will extract $c$ with optimal estimator variance, regardless of the instance, simply by using the empirical covariance $\hat c=\frac{1}{N}\sum_{i=1}^N \zeta_x^{(i)}\zeta_p^{(i)}$.

\begin{theorem}[Practical quantum advantage in learning an EM correlation, informal] 
\label{thm:EM_advantage_informal}
Consider the near-isotropic EM family characterized by \eqref{eq:Sigma_c} with $\sigma_x^2,\sigma_p^2, c \ll O(1)$ and $\Delta = |\sigma_x^2 - \sigma_p^2| \ll |c|$. Then Bell QSL distinguishes the pair using $O(e^{-4r}/c^2)$ shots, where $r$ is the squeezing parameter, while any adaptive $\{x, p\}$ homodyne strategy requires at least $\Omega(e^{-4r}/\Delta^2)$ shots to do so with high probability.
\end{theorem}
Theorem \ref{thm:EM_advantage_informal} implies a separation on the sample complexity of learning $c$, because homodyne strategies can often fail by a sign error. The performance gap between Bell QSL and entanglement-free benchmarks, which improves with squeezing, is demonstrated in Fig. \ref{fig:em_field}. In  Sec.~\ref{sec:OT_main_text} (and App.~\ref{sec:transport_QA}) we formalize this notion of ill-conditioning and a ``distance to marginal blindness'' using an optimal-transport stability framework.

In many precision experiments where signals are fleeting and optical probes operate at low noise temperatures, Bell QSL only requires storage of entangled idler modes for short times. However, in lower-frequency applications such as radar, longer integration times require the prepared entanglement to be stored coherently for longer times, which can be practically challenging. In such settings, photon loss is often the dominant error source. Assuming a loss rate of $0<p<1$ during the storage process, one finds that the $e^{-4r}/\epsilon^2$ idealized sample complexity is replaced by $(1-p)e^{-4r}/\epsilon^2 + O(p)$ \cite{jiang2024bellfix}. Operationally, under a finite number of samples, the loss rate creates a \textit{precision floor}, such that for values of $\epsilon$ above this floor, the entangled protocol retains its full $e^{-4r}$ speedup whereas below the floor, the provable speedup is lost if one compares against a noiseless benchmark. This effect is identified in Ref.~\cite{cotler2026noisyquantumlearningtheory} as an unavoidable consequence of local noise models which entirely decohere the given state.

\vspace{-1em}
\subsection{Feedback Control in Interferometric Cavities}
\label{sec:cavity_noise}
\vspace{-2mm}

Beyond direct sensing of external fields, a natural application is learning internal cavity or interferometer noise parameters that appear as Gaussian distortions of optical quadratures. In high-sensitivity interferometers, including gravitational-wave detectors such as LIGO, many relevant technical and quantum noise processes are well approximated as Gaussian, so their net effect on a probe field is captured by a small number of covariance parameters: loss, squeezing strength, and quadrature rotation~\cite{CavesSchumaker1985TwoPhotonI, BuonannoChen2001SRQuantumNoise, DanilishinKhalili2012LRR}. However, these parameters can be difficult to track in real-time sensing from homodyne marginals alone.

A common effective model for phase or control imperfections is an unknown quadrature rotation,
\begin{equation}
\begin{pmatrix}\hat x\\ \hat p\end{pmatrix}
\longmapsto
R(\theta)\begin{pmatrix}\hat x\\ \hat p\end{pmatrix},
\quad
R(\theta)=
\begin{pmatrix}
\cos\theta & -\sin\theta\\
\sin\theta & \cos\theta
\end{pmatrix},
\label{eq:rotation_channel}
\end{equation}
where $\theta$ may represent squeezing-angle drift, detuning-induced quadrature mixing, or more general phase-space misalignment. Bell QSL provides a natural real-time monitoring strategy that can be used to learn $\theta$ with a practical speedup. 

Concretely, one can apply a known displacement $\beta$ onto the probe arm of a TMSV, inject it into a sideband of the interferometer near the primary sensing band, and run the remainder of Bell QSL as usual on the returned probe and stored idler. As shown in App.~\ref{sec:D4-interferometer-noise}, the resulting data $\zeta^{(i)}$ then has mean $\E[\zeta] = e^{i\theta}\beta$ while its per-quadrature measurement noise is $\nu_r = \tfrac12 e^{-2r}$ around standard, small $\theta$. Thus the sample mean $\bar\zeta=\frac{1}{N}\sum_{i=1}^N \zeta^{(i)}$ yields the signed phase estimate $\hat\theta=\arg\!\left(e^{-i\arg\beta}\bar\zeta/|\beta|\right)$, with sample complexity $N=O(\nu_r/(|\beta|^2\epsilon^2))$ for $\epsilon$-accurate estimation, in direct analogy with the moment estimators used for Gaussian covariance learning (Sec.~\ref{sec:ex_gaussian_covar}). Both quadrature components are recoverable from every measurement. 

Meanwhile, heterodyne measurements incur the standard multiplicative overhead of $e^{2r}$. Because $\theta$ lies along an unknown axis, homodyne measurements must be performed adaptively to identify both its magnitude and sign, and each measurement shot may carry little information about one of the two quadratures. In real-time measurements, this adaptive scanning may be experimentally intractable in conjunction with axis-aligned squeezing, preventing entanglement-free error-mitigation beyond the vacuum-noise limit.

Operationally, the two-mode-squeezed light may be injected into optical sidebands at each analysis frequency, preventing interaction with the primary signal, while the idler can be kept anywhere outside the main readout band. This enables non-invasive quantum-enhanced noise monitoring, and data can be fed back to stabilize operating parameters such as squeezing angle or detuning \cite{McCuller2020FDSAdvLIGO,Yap2020EPRFDS, Gould2021EntangledWitnessQNC}. In this way, Bell QSL interfaces naturally with the control layer that often determines whether shot-noise-limited interferometry can be sustained over long integration times.

\vspace{-1em}
\section{Quantum Advantage for QSL} \label{sec:QA_main_text}
\vspace{-2mm}

In Section~\ref{sec:applications}, we proposed applications of Bell QSL to learning electromagnetic signals and noise, and demonstrated a simple example of practical quantum advantage in learning field correlations. Here, we establish much more general results on quantum advantage in QSL.

\vspace{-1em}
\subsection{Matched Filtering}
\vspace{-2mm}
The quantum advantage in Theorem~\ref{thm:EM_advantage_informal} was established against a practically motivated restricted-angle homodyne baseline. Bell QSL also achieves a worst-case advantage
against \emph{any} entanglement-free data-collection strategy, including adaptive protocols with arbitrarily complex measurement bases.  This separation arises in matched filtering and beamforming, ubiquitous sensing primitives in which one scores a large bank of candidate templates against a signal measurement record and selects the hypotheses most consistent with the data. Variants of this procedure appear throughout signal processing, from communications to radar, and are central to modern high-precision searches such as template-bank pipelines in gravitational-wave detectors, where the number of candidate waveforms can be enormous~\cite{Allen2005ChiSquare,Allen2012FINDCHIRP,Babak2013IHOPE,Usman2016PyCBC,Messick2017GstLAL,
DalCantonHarry2017O2TemplateBank,Roulet2019GeometricTemplateBank,Mukherjee2021GstLALTemplateBank}.

Concretely, consider an $n$-mode waveform
\begin{equation}
F(t)=\sum_{k=1}^n \sqrt{\omega_k}\big(A_k\cos(\omega_k t)+B_k\sin(\omega_k t)\big),
\end{equation}
where $\sqrt{\omega_k}$ normalizes the phase-space response for different frequency modes. The frequencies $\omega_k$ are each taken to be resonant with the known narrowband response of one sensor mode. The matched filtering problem is as follows. Interrogating the signal using $n$ modes induces displacements $\alpha \in \mathbb{C}^n$ with law $P(\alpha)$ determined by the coefficients $A_k, B_k$. A template is a vector $w \in \mathbb{C}^n$, corresponding to a list of the 2$n$ values $\{A_k, B_k\}$ defining one hypothesis for the waveform. The matched-filter output $Y_w = \Re(\alpha^\dagger w)$ then quantifies how well the template $w$ describes the signal inducing displacement $\alpha$. For Fourier-domain scoring we use the bounded feature $\varphi_{Y_w}(t) = \E[e^{i t Y_w}]$, which is a numerically stable surrogate for maximum-likelihood-based scores and is especially natural when the signal modes are independent in the Fourier basis. Crucially, a query $(w,t)$ is equivalent to a characteristic-function evaluation via the relation $\varphi_{Y_w}(t) = \chi_P(itw/2)$ (App.~\ref{sec:worst_case_QA}). Matched filtering with a bank of $M$ templates thus reduces to estimating $\chi_P$ at $M$ points. With this observation, we find that Bell QSL achieves an exponential worst-case advantage in estimating sinusoidal waveforms via matched filtering.

\begin{theorem}[Worst-case quantum advantage for matched filtering, informal]
\label{thm:matched_filtering_informal}
There exist priors over the sinusoidal coefficients $\{A_k,B_k\}$ for which any entanglement-free protocol (including those with squeezed probe states) requires $N \ge \Omega(3^n)$ channel uses to score any single template to constant accuracy. In contrast, Bell QSL with two-mode squeezing $r = \Theta(\log n)$ can estimate all $M$ scores to accuracy $\epsilon$ using only $N = O(\epsilon^{-2}\log M)$ shots, independent of $n$.
\end{theorem}

While the sinusoidal-waveform sensing task considered in Theorem~\ref{thm:matched_filtering_informal} is commonplace, this worst-case advantage occupies only a narrow region of the full problem space, since it relies on a highly structured prior over the coefficients $\{A_k,B_k\}$ that hides the relevant information from entanglement-free readout. Next, we return to experimentally motivated baselines and develop a framework that quantifies the tradeoff between the magnitude of quantum advantage and the fraction of physically relevant instances for which that advantage persists.

\vspace{-1em}
\subsection{Optimal Transport and Practical Advantage}\label{sec:OT_main_text}
\vspace{-2mm}

We now introduce a general conditioning methodology that turns measurement constraints into sample-complexity lower bounds. Any single-shot measurement strategy can be described as a map $\mathcal{M}$ that sends the underlying displacement law $P$ to a one-round outcome distribution $\mathcal{M}(P)$. For instance, homodyne measurements made along a fixed set of axes in one round consist of projection to finitely many quadrature marginals followed by known Gaussian blurring, and $\mathcal{M}$ is simply the resulting product distribution. Let $\mathcal{C}$ denote the hypothesis class of possible phase-space distributions (e.g.~the family of centered Gaussians earlier). Now consider the following optimal-transport (OT) stability modulus, 
\begin{align}
&\omega_{\mathcal{M},\mathcal{C}}(\eta)
:=  \\
&\quad \sup\{W_1(P,Q): P,Q\in\mathcal{C},\,
\mathrm{TV}(\mathcal{M}(P),\mathcal{M}(Q))\!\le\! \eta\},
\end{align}
where $W_1$ is the Wasserstein-1 (earthmover) distance. Intuitively, $\omega_{\mathcal{M},\mathcal{C}}$ quantifies how far apart two distributions in $\mathcal{C}$ can be while remaining nearly indistinguishable after measurement. Equivalently, it measures how much phase-space structure is hidden from $\mathcal{M}$, and how ill-conditioned the inverse problem is. This intuition can be directly tied to sample lower bounds.

\begin{theorem}[OT ambiguity implies minimax lower bounds, informal]
\label{thm:OT_informal}
Fix an experiment $\mathcal{M}$ and a class $\mathcal{C}$ with finite first moment. Then there exists a $1$-Lipschitz function $f$ such that any estimator based on $N$ i.i.d.~samples from $\mathcal{M}(P)$ has worst-case error at least $\Omega(\omega_{\mathcal{M},\mathcal{C}}(1/N))$ in estimating $\E_{P}[f]$.
\end{theorem}
\noindent Note that when $\omega_{\mathcal{M}, \mathcal{C}}(0) > 0$, no number of measurements can drive the error below a constant, meaning some property is completely non-identifiable. This optimal-transport viewpoint also supports an inverse approach to finding a quantum advantage: given a sensing constraint, one identifies nearly indistinguishable instance pairs by maximizing $\omega$, and then relaxes around them to quantify how much entanglement improves sample complexity.

\begin{figure}
    \centering
    \includegraphics[width=\linewidth]{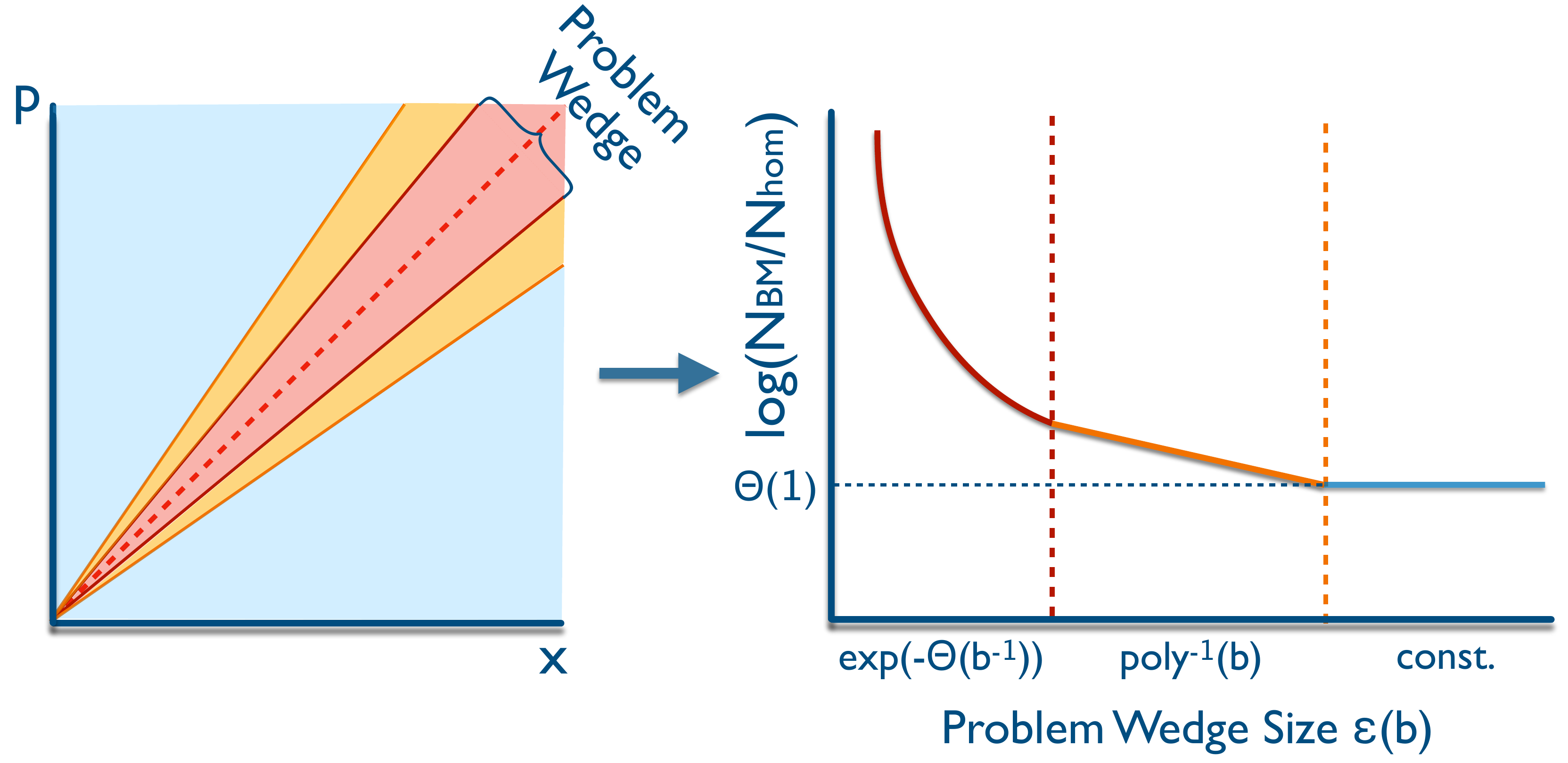}
    \caption{The \textit{problem wedge} of quantum advantage, illustrated for the delta-mixture toy example. As the symmetry-breaking parameter $\varepsilon \rightarrow 0$, homodyne shots become non-identifiable and sample complexity diverges. In the red region, $\varepsilon$ is exponentially small relative to $b$, yielding exponential advantage, and in the orange region, it is polynomially small. The same scaling controls the advantage in Theorem \ref{thm:EM_advantage_informal}, with $\varepsilon\rightarrow\Delta$.}
    \label{fig:deltas_prob_wedge}
    \vspace{-3mm}
\end{figure}

The following toy example illustrates a geometric intuition for Theorem \ref{thm:OT_informal}: quantum advantage depends on how much of phase space the problem class explores. Consider the distributions
\begin{equation}
P = \frac{\delta_{(a,b)}+\delta_{(-a,-b)}}{2}, \ \ 
Q = \frac{\delta_{(a,-b-\varepsilon)}+\delta_{(-a, b+\varepsilon)}}{2}\,.
\label{eq:deltas_eq_main_text}
\end{equation}
When $\varepsilon = 0$, $P$ and $Q$ have identical $\{x,p\}$ marginals despite living in disjoint quadrants of phase space (depicted in App.~\ref{sec:ex_sum_of_deltas}), so $\{x,p\}$-limited homodyne can never distinguish them. This example is used to demonstrate the utility of Bell measurement in collecting a robust dataset for classical learning algorithms beyond our estimators in Fig.~\ref{fig:unsupervised}. We therefore interpret $\varepsilon$ as a symmetry-breaking parameter that controls how strongly the hidden two-dimensional structure becomes visible to restricted homodyne. We call the corresponding restricted instance family a \emph{problem wedge}: a subset of the full problem space whose scale is determined by $\varepsilon$, viewed as a 
function of the intrinsic scale $b$.

For $\varepsilon>0$, the homodyne marginals differ only at $O(\varepsilon^2)$ (in e.g.~KL divergence), so $\{x,p\}$-limited homodyne incurs sample complexity $\Theta(e^{-2r}/\varepsilon^2)$, while Bell QSL distinguishes using only $O(e^{-2r}/b^2)$ shots (App.~\ref{sec:ex_sum_of_deltas}). The advantage factor is thus controlled by $(b/\varepsilon)^2$: in the red wedge of Fig.~\ref{fig:deltas_prob_wedge}, where $\varepsilon$ is exponentially small in $b$, the advantage is exponential, while for polynomially small $\varepsilon$ it is polynomial. Although the delta-mixture makes this geometry particularly transparent, the same $\varepsilon^{-2}$ conditioning governs the advantage found in Theorem \ref{thm:EM_advantage_informal} for EM field sensing, where the role of $\varepsilon$ is played by the symmetry-breaking $\Delta$ that leaks $c$ into the accessible homodyne statistics. In our optimal transport language, the $W_1$ separation is $\Theta(b)$ while $\{x,p\}$-restricted homodyne only sees the symmetry breaking through $O(\varepsilon^2)$ leakage, yielding the corresponding $1/\varepsilon^2$ lower bound. Using Theorem \ref{thm:OT_informal}, these simple, illustrative examples of problem-wedge tradeoff can be generalized to less structured, multimodal sensing tasks, and against homodyne learners with different measurement settings.

Prior exponential quantum learning advantages often arise only under worst-case benchmarks against all entanglement-free learners \cite{chen2021exponentialseparationslearningquantum, aharonov2022quantum, Huang_2022, jiang2024bellfix, chen2024optimaltradeoffsestimatingpauli}, a class broad enough to include exponentially deep, computationally intractable measurement strategies. Quantum advantage is then forced into vanishingly small corners of the instance space that do not coincide with the set of physically interesting instances. By instead benchmarking against weaker, but experimentally feasible conventional strategies, our perspective reveals sizable regimes where quantum information processing yields genuine, usable advantages. Going forward, benchmarking against experimentally feasible strategies will be essential to avoid artificially narrowing the search for quantum-enhanced experiments in ways that obscure such useful advantages.

\begin{figure}
    \centering
    \setlength{\belowcaptionskip}{-3pt}

    \includegraphics[width=\linewidth]{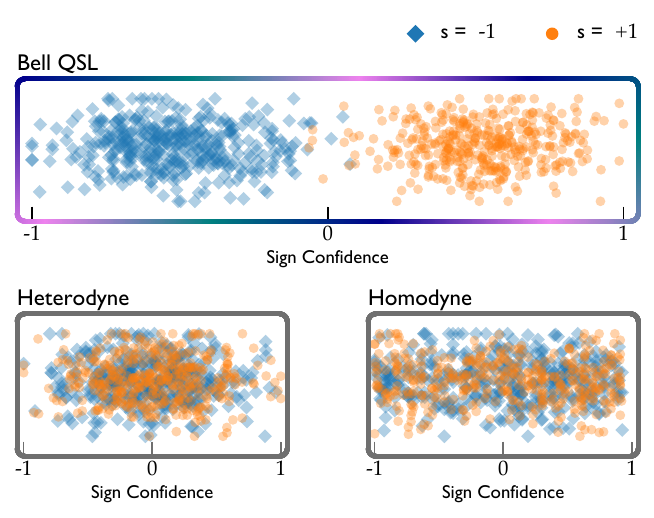}
    \caption{Classical unsupervised learning employing Bell QSL data can learn signal properties beyond the scope of heterodyne and homodyne tomography. 
    Each measurement method collects many draws from the underlying distributions $P, Q$ from \eqref{eq:deltas_eq_main_text}, with squeezing $r=1.2$, $a=0.35$, $b = 0.06$, and $\varepsilon = 0$. Each dataset is fed into unsupervised ML (Gaussian kernel PCA) to learn whether the data were drawn from $P$ (sign +1) or $Q$ (sign -1). Even at modest $b<1$ and squeezing, the classical Bell dataset constitutes a far more feature-rich training dataset for ML algorithms than conventional strategies.}
    \label{fig:unsupervised}
    \vspace{-3mm}
\end{figure}

\vspace{-1em}
\subsection{Advantage without Entanglement}
\vspace{-2mm}

We have shown that Bell QSL, leveraging entanglement, can operate at equal probe-state energy to homodyne and heterodyne tomography while resolving features of classical fields inaccessible to entanglement-free protocols. Here we show that it is also possible to demonstrate exponential separations that arise from squeezing alone.

We consider the following problem. Place a sensor in a background where it experiences many random kicks of fixed amplitude. In each sensing timestep, the sensor is thus evolved by the displacement operator $\disp(X)$, with $X = \sum_{i=1}^n z_i$, and $z_i = \pm 1$; however, we impose that the $z_i$'s are either random or satisfy a parity constraint. This is a toy instantiation of a more practical physical problem: estimating \textit{moments} of classical signals.

\begin{theorem}[Exponential quantum advantage with only squeezed probes] \label{thm:exp_qa_squeezing_only}
Any strategy with arbitrary energy and arbitrary POVMs, but which uses classical light to probe the signal, requires $\exp(\Omega(n))$ measurements to distinguish random and parity-constrained classical backgrounds, while there is a strategy using a single squeezed probe of $r = \Theta(\log n)$ that requires $O(1)$ shots.
\end{theorem}
\noindent By embedding the effective instance size $n$ into a single-mode task, Theorem \ref{thm:exp_qa_squeezing_only} (proven in App.~\ref{sec:entanglement_req}) shows that exponential-in-$n$ advantages are possible purely from using squeezed light vs.~classical (coherent) light. 

We remark that this result falls under a different operational definition of quantum advantage than is assumed in our previous separations. Note that classical coherent-state probes can differ in energy based on mean photon number (i.e.~phase-space displacement), while squeezed light is generated from coherent states by non-energy-conserving symplectic unitaries. Because both displacements and squeezing can enhance probe sensitivity, the canonical bosonic Heisenberg limit treats the two transformations on equal footing and pertains only to probe state preparation with a bounded energy budget~\cite{Safranek_2016}. However, in this Fisher-information picture, optimizing the tradeoff between displacement and squeezing energy yields at most polynomial improvements. Theorem \ref{thm:exp_qa_squeezing_only} shows that a learning-theoretic QSL analysis can be much more illustrative: for certain problems, squeezed probe states can be exponentially more informative than \textit{any} displaced classical probe, independent of energy. This establishes a \textit{non-classical} energy budget as a unique sensing resource whose effect cannot be replicated by any amount of classical resources.

\vspace{-1em}
\section{Discussion}
\vspace{-2mm}

Our work establishes that quantum advantage for sensing classical fields becomes both sharper and more useful once we formulate the task in terms of the practical limitations of real-world experiments. Having introduced the Quantum Signal Learning framework to capture these features, we established that quantum sensors enhanced by quadrature squeezing and entangled quantum memory can extract properties of classical signals far more efficiently than their conventional counterparts. Our work makes such quantum learning advantage statements precise, demonstrating utility for ubiquitous physical tasks such as electromagnetic field sensing and interferometric feedback control. Moreover, our algorithm requires only a single layer of entanglement-generating operations and conventional homodyne measurements, while enabling post-hoc estimation of many properties from the collected classical record. We introduce an optimal transport-based method to prove quantum learning speedups over practical, experimentally feasible conventional benchmarks, going beyond previous worst-case advantage statements.

Our work extends and generalizes the literature on bosonic learning theory and metrology. While Ref.~\cite{jiang2024bellfix} exhaustively studies the task of learning a complete description of a displacement channel with post-hoc queries (a theoretically-motivated task akin to full process tomography), QSL takes a more general and experimentally-motivated perspective: properties can be any desired functional of the signal. Requiring post-hoc learning of the full channel can make the task artificially hard as it is equivalent to learning exponentially many predetermined observables, so QSL only requires learning the few relevant signal properties efficiently. This allows for nontrivial and experimentally useful separations in less specialized settings than worst-case multimode channel learning, including with even a single sensor mode. Moreover, the query-complexity reformulation of sensing efficiency makes new kinds of quantum speedups possible that are not present from the canonical Fisher information perspective. QSL, being simple yet extremely expressive, allows for a unification of the learning-theoretic and metrological perspectives in an experimentally-minded fashion.

We have demonstrated preliminary applications of Bell QSL to tasks of electromagnetic field sensing and noise characterization. Expanding on these concepts may be fruitful for radar sensing and detectors in high-precision fundamental physics experiments and astronomy \cite{SilvaFeaver2016DMRadioPathfinder,Malnou2019SqueezedVacuumAxion,Tse2019QuantumEnhancedLIGO,Chang2019QuantumNoiseRadar,Ouellet2019ABRACADABRA10cm,Backes2021QuantumEnhancedAxion,Romanenko2023DarkSRF}. Entanglement-enabled real-time noise monitoring could enable real-time error mitigation in interferometers for gravitational wave astronomy, and meaningfully improve large-hypothesis matched filter testing in e.g.~extreme-high-bandwidth multiplexing. 

Theoretically, our work establishes a preliminary framework for proving \textit{practical} quantum advantage in quantum sensing by associating a limited measurement model with an optimal-transport stability modulus. Interfacing these concepts with the learning tree formalism \cite{chen2021exponentialseparationslearningquantum}, which is already equipped with sophisticated lower-bound proof tools, may allow the quantum advantage statements in this work to generalize to other quantum-enhanced experimental settings.

Searching for meaningful quantum-enhanced experiments requires keeping experimental constraints firmly in view. Taking this approach, our work reveals quantum speedups across a wide range of important sensing tasks, establishing that practical quantum advantage extends naturally to the sensing of classical fields.

\vspace{-1em}
\subsection*{Acknowledgments}
\vspace{-2mm}

JC is supported by an Alfred P. Sloan Fellowship.  DLD acknowledges support as a Black Hole Initiative Fellow. His contribution to this publication is funded in part by the Gordon and Betty Moore Foundation (Grant \#13526). It was also made possible through the support of a grant from the John Templeton Foundation (Grant \#63445). The opinions expressed in this publication are those of the author(s) and do not necessarily reflect the views of these Foundations.  IK is supported in part by the Nobile Research Initiative.

\bibliographystyle{apsrev4-1_with_title}
\bibliography{references}

\appendix
\onecolumngrid

$$$$

\newpage

\setcounter{page}{1}
\renewcommand{\thepage}{S\arabic{page}}

\begin{center}
\textbf{\Large Supplementary Information}
\end{center}

\renewcommand{\appendixname}{APPENDIX}
\renewcommand{\thesubsection}{\Alph{section}.\arabic{subsection}}
\renewcommand{\thesubsubsection}{\Alph{section}.\arabic{subsection}.\alph{subsubsection}}
\makeatletter
\renewcommand{\p@subsection}{}
\renewcommand{\p@subsubsection}{}
\makeatother

\renewcommand{\figurename}{Supplementary Figure}
\setcounter{figure}{0}
\setcounter{secnumdepth}{3}

\bigskip

\noindent \textbf{\ref{sec:related}.~~\hyperref[sec:related]{Related Work}} \dotfill\textbf{\pageref{sec:related}}
\medskip

\noindent \textbf{\ref{sec:prelims}.~~\hyperref[sec:prelims]{Preliminaries}} \dotfill\textbf{\pageref{sec:prelims}}
\medskip

\noindent \qquad \begin{minipage}{\dimexpr\textwidth-1.3cm}
 \hyperref[sec:cv_overview]{Continuous-variable quantum mechanics} $\bullet$ 
\hyperref[sec:cv_measurements]{Homodyne, Heterodyne, and Bell measurements} $\bullet$ 
 \hyperref[sec:OT_overview]{Review of optimal transport theory}
\end{minipage}
\medskip

\noindent \textbf{\ref{sec:QSL_protocol}.~~\hyperref[sec:QSL_protocol]{Quantum Signal Learning through Bell Sensing}} \dotfill\textbf{\pageref{sec:QSL_protocol}}
\medskip

\noindent \qquad \begin{minipage}{\dimexpr\textwidth-1.3cm}
\hyperref[sec:QSL_problem]{The Quantum Signal Learning problem} $\bullet$
 \hyperref[sec:alg]{General form of Bell QSL algorithm} $\bullet$
 \hyperref[sec:Fourier_alg]{Bell QSL for Fourier-defined kernels} $\bullet$
 \hyperref[sec:distribution_alg]{Bell QSL for polynomial moment kernels}
\end{minipage}
\medskip

\noindent \textbf{\ref{sec:qa_theory}.~~\hyperref[sec:qa_theory]{Quantum Advantage in Learning Noisy Classical Signals}} \dotfill\textbf{\pageref{sec:qa_theory}}
\medskip

\noindent \qquad \begin{minipage}{\dimexpr\textwidth-1.3cm}
 \hyperref[sec:worst_case_QA]{Advantage against all entanglement-free strategies in matched filtering} $\bullet$ 
 \hyperref[sec:transport_QA]{Theory of practical advantage against restricted learners}
\end{minipage}
\medskip

\noindent \textbf{\ref{sec:prac_Gaussian_QA}.~~\hyperref[sec:prac_Gaussian_QA]{Examples of Practical Quantum Advantage}} \dotfill\textbf{\pageref{sec:prac_Gaussian_QA}}
\medskip

\noindent \qquad \begin{minipage}{\dimexpr\textwidth-1.3cm}
  \hyperref[sec:ex_sum_of_deltas]{A classical mixture of peaked displacements} $\bullet$
\hyperref[sec:ex_gaussian_covar]{Learning Gaussian correlations} $\bullet$
  \hyperref[sec:prac_Hams]{Sensing correlated EM fields and cavity phase noise} $\bullet$
 \hyperref[sec:entanglement_req]{Superpolynomial quantum advantage without entanglement}
\end{minipage}
\medskip

\section{Related Work}\label{sec:related}
\noindent \textbf{Quantum radar.} At the level of terminology, quantum radar \cite{Karsa_2024} may seem closely related to QSL, particularly in its application to electromagnetic signals. This line of work, however, is operationally very distinct, its only technical similarity being the initial preparation of entangled EPR pairs. The concept of quantum radar was inspired by \cite{lloyd2008quantumillumination}, which made the observation that objects surrounded by thermal noise could be imaged with beyond-classical signal-to-noise ratio (SNR) by first preparing maximally-entangled pairs of photons, then shining half towards the object while the other half remain in coherent quantum memory. Despite noise sources causing photon loss and breaking the generated entanglement, the received light has boosted SNR because of the entanglement imprint that signal photons carry. This quantum illumination (QI) protocol was concretely formulated in \cite{Tan_2008} for Gaussian states. Limitations were soon realized, showing e.g. that QI loses its utility in the single-photon detection regime \cite{Shapiro_2009}, and that the advantage is degraded when the returning photons do not have known amplitudes and phases (as is often the case in practice) \cite{Zhuang_2017}. Regardless, QI, or quantum radar, is a useful early demonstration that entanglement can be traded for SNR, and was practically realized in experiments such as \cite{Lopaeva_2013}.

Conceptually, QI is a substantial departure from Bell QSL: we prepare entangled bosonic sensors, and keep them coherent in the presence of an external field, rather than sending out photons to receive later. Moreover, the quantum advantage in our setting relies on the preservation of entanglement, and the performance gap that our algorithm achieves over unentangled benchmarks grows with the complexity of the ambient noise, unlike prior quantum sensing and illumination applications which are generally degraded to classical benchmarks when noise is taken into account.

\medskip

\noindent\textbf{Continuous-variable quantum learning.} A recent line of work in quantum learning theory has approached  canonical tomography problems, for which qubit-based learning guarantees (e.g.~\cite{Bisio_2010, Zhao_2024, mele2026optimallearningquantumchannels, huang2023learningpredictarbitraryquantum, schuster2023learning}) have been generalized to the bosonic continuous-variable setting. These include, for instance, algorithms to learn observables of continuous-variable quantum states \cite{gandhari2023precisionboundscontinuousvariablestate, Wu_2024}, or full tomography of Gaussian states \cite{Weedbrook2012}, and Gaussian unitaries \cite{fanizza2025efficientlearningbosonicgaussian}. These results can be viewed as the continuous-variable counterparts of canonical quantum learning theory, and are concerned with worst-case bounds on learning arbitrary observables, or complete classical descriptions of bosonic quantum systems. While related in quantum-mechanical formalism, our work departs from the usual worst-case quantum learning setting over arbitrarily complex quantum systems, instead learning specific, physically-motivated properties of classical fields realized as linear bosonic Hamiltonians.

More closely aligned to our work is \cite{jiang2024bellfix}, which uses the Bell data-collection strategy in Bell QSL to achieve a worst-case quantum advantage for learning highly non-Gaussian displacement properties. A modification of this worst-case instance was realized experimentally in \cite{liu_brunel_2025}. We apply their results to demonstrate worst-case quantum advantage in the practical setting of matched filtering to learn sinusoidal signals, and the estimators we develop allow for the efficient post-hoc evaluation of physically-motivated properties such as moments of the signal beyond the characteristic function. Their analysis of algorithmic noise also ports to our setting, suggesting that Bell QSL retains beyond-classical capabilities even in the presence of constant local error in the implementation. Leveraging this robustness, our work builds on a body of work towards achieving quantum advantages that remain practically meaningful in the presence of noise \cite{cotler2026noisyquantumlearningtheory}.

In the qubit setting, \cite{Prabhu_2026} leverages Bell measurements and other entangled probes (e.g. GHZ states) to achieve exponential (noise-free) advantages in sensing stochastic fields. This work observes that quantum advantages are controlled by total variation between characteristic functions of the classical stochasticity, supplementing the observation from \cite{jiang2024bellfix} that information-theoretic hardness in the absence of entanglement usually requires a sensing objective encoded in highly nonlocal (high Fourier frequency) components that cannot be estimated efficiently with local measurements. Importantly, we identify that this condition can be concretely studied and instantiated through the marginals of the state's high-dimensional distribution, and we use this observation to introduce a probabilistic framework for practical quantum advantages. This framework, based in optimal transport theory, provides a systematic method to identify sensing and learning tasks which may admit entanglement-enabled advantage under a given set of experimental constraints, enabling us to identify simple, single-mode sensing tasks which admit a speedup rather than tailoring high-frequency, parity-constrained problems as in \cite{Prabhu_2026}.

\medskip
\noindent \textbf{Quantum metrology.} Quantum Signal Learning departs from quantum metrology in the modern quantum Fisher information (QFI) sense, which typically studies estimation of a low-dimensional parameter vector $\theta$ encoded in a known parametric family of states or channels, and optimizes probe preparation, interaction time, and measurement to minimize a chosen risk for that \emph{prespecified} model \cite{Helstrom1976,Paris2009,Giovannetti2011}. Multi-parameter generalizations still adopt this fixed-structure viewpoint: one commits to a particular parameterization and then confronts additional incompatibility constraints (e.g.\ the noncommutativity of optimal measurements across parameters), so optimal strategies remain tailored to the chosen parameter set rather than yielding a broadly reusable measurement record \cite{Szczykulska2016,Ragy2016}. Even when the parameter is promoted to a time-dependent classical waveform, the objective is typically to reconstruct that waveform under a specified prior and cost functional, with performance characterized by waveform-level quantum Cram\'er--Rao bounds \cite{Tsang2011}. Closely related system-identification and Hamiltonian-learning approaches likewise target structured generator models (often with coherent control and adaptive design), aiming to infer a compact Hamiltonian description rather than to answer many post-hoc queries about an induced response distribution \cite{Wiebe2014}.

By contrast, QSL is designed for regimes common in quantum-noise-limited platforms where the signal and ambient fluctuations induce a nontrivial \emph{distribution} over phase-space responses, and where one wants to query many properties of that distribution (matched-filter scores, correlations, higher moments, etc.) from the same experimental dataset. This is precisely the setting in which QFI-based optimality can be misaligned with practical learnability: a strategy that maximizes QFI for one assumed signature is often brittle to model mismatch and typically provides little information about other noncommuting or out-of-model features. Our framework instead emphasizes post-hoc learnability, leveraging a universal Bell-sensing primitive that returns an informationally complete two-quadrature record with sub-vacuum additive noise, and providing sample-complexity guarantees whose dependence is governed by the complexity of the chosen property rather than by committing to a particular parametric signal model in advance.

Finally, a central lesson of quantum metrology is that entanglement-enabled Heisenberg scaling is generically fragile under Markovian loss and decoherence, leading to at most constant-factor improvements unless additional structure is available \cite{Escher2011,Demkowicz2012,Demkowicz2015ProgOpt, depradenne2026restrictionsnoncliffordfaulttolerance}. A major thread of recent work therefore develops error-corrected metrological protocols that embed the signal generator into a code space to restore improved scalings in the presence of noise \cite{Kessler2014,Arrad2014,Dur2014}. Our advantage statements do not rely on such long-coherence or quantum error-correction assumptions: QSL treats the net effect of signal and ambient uncertainty as the displacement distribution to be learned, and the Bell-sensing estimator family retains robustness to constant implementation noise while still improving the effective two-quadrature noise floor with squeezing \cite{jiang2024bellfix}.

\medskip
\noindent \textbf{Experimental predecessors.} While its formal presentation in the language of sample complexity and quantum advantage is relatively contemporary, joint measurement of noncommuting quadratures using EPR pairs has been a standard primitive in quantum optics for decades. Historically, Bell pairs were introduced to quantum information theory as a core component of superdense coding and quantum teleportation \cite{BennettWiesner1992,Bennett1993Teleportation, BraunsteinKimble1998, BraunsteinKimble2000DenseCodingCV}, and prepared experimentally for early realizations of the EPR paradox \cite{Ou1992}. Early experimental applications of Bell-basis measurements appeared in \cite{PhysRevLett.76.4656, Furusawa1998} to realize these theoretical primitives. These experiments established the core benefits of realizing entanglement-enabled quantum protocols on continuous variable systems, as the Bell measurement strategy is compatible with static-axis, offline squeezing, linear optics, and complex classical postprocessing.

As the utility of Bell-basis measurement began to be realized, closely related joint-quadrature and conditional-measurement ideas were pursued as routes to surpassing vacuum-limited readout in precision interferometry. Seminal work in gravitational-wave detection \cite{Kimble2001} emphasized that optomechanical dynamics imprint frequency-dependent correlations between amplitude and phase quadratures, motivating readout strategies that effectively access appropriate correlated quadrature combinations rather than a single fixed homodyne axis. Along similar lines, quantum-dense metrology demonstrated in a laser interferometer that EPR resources and dual homodyne readout can provide simultaneous information about two noncommuting observables with sub-vacuum uncertainty, enabling a sub-shot-noise error-detection channel that operates outside the primary science band and monitors whether disturbances populate an undesired quadrature \cite{Steinlechner2013}. More recently, entangled signal-idler schemes have been realized as practical alternatives to long filter cavities in gravitational-wave architectures \cite{Ma2017,Sudbeck2020}. Our contribution is to lift this mature experimental primitive into a distributional learning framework. We treat Bell measurement as an informationally complete, low-noise data-collection module for classical-field-induced displacement distributions, and we give explicit estimators with end-to-end sample-complexity guarantees for broad families of post-hoc queries, thereby quantifying advantage regimes that are only implicit in much of the historical literature.

\medskip
\noindent \textbf{Quantum advantage in bosonic sensing.} Our work establishes quantum advantages enabled jointly by squeezing and entanglement, and pertains to a wide class of practical sensing tasks due to the optimality of Bell-basis readout. Meanwhile, \cite{Robert_unreleased} establishes exponential quantum advantages for bosonic sensing in a different context: for specific choices of time-varying classical fields, they demonstrate that squeezed probes, absent entanglement, enable exponential advantages over classical ones. This is manifestly distinct from our approach, and aligns more closely with the setting of our Theorem \ref{thm:exp_qa_squeezing_only}.

Recently, the QSL formalism introduced in this work has been utilized by Ref.~\cite{kannan2026exponentialquantumadvantagelearning} to prove exponential quantum advantages in bosonic sensing of single-mode fields using single-qubit control. This suggests that the framework introduced here may lend itself to new formulations of quantum metrology which elucidate problems where quantum information processing capabilites yield beyond-quadratic or superpolynomial advantages over conventional approaches. 

\section{Preliminaries}
\label{sec:prelims}

\subsection{Overview of Gaussian continuous-variable quantum mechanics}
\label{sec:cv_overview}
\vspace{-1em}

Here we collect standard facts in continuous-variable quantum mechanics used throughout this work.
\vspace{-1em}
\subsubsection{Phase space as a representation of the bosonic Hilbert space.}\vspace{-1em}
An $n$-mode bosonic quantum state lives in the Hilbert space $L^2(\mathcal{\mathbb{R}}^{n})$. All such states have convenient representations as functions on the \textit{phase space} $\mathbb{R}^{2n}$. To elucidate this, we require basic notational conventions. A bosonic mode is labeled in phase space by canonical quadrature operators $\xh, \ph$. Generalizing to $n$ modes, the quadrature vector 
\begin{equation}
    \hat{R} = (\xh_1, \xh_2, ..., \xh_n, \ph_1, \ph_2, ..., \ph_n)
\end{equation}
labels phase space coordinates. Under the convention $\hbar = 1$, the quadrature operators satisfy the canonical symplectic commutation relations
\begin{equation}
    [\hat{R}_j, \hat{R_k}] = i\Omega_{jk}\ , \quad \Omega = \begin{pmatrix} 0 & \mathds{1}_n \\
    -\mathds{1}_n & 0
    \end{pmatrix} \ .
\end{equation}
Physically, the operators $\xh, \ph$ can be obtained from the creation and annihilation operators $\hat{a}, \hat{a}^\dagger$ according to
\begin{equation}
    \xh = \frac{\hat{a}^\dagger + \hat{a}}{\sqrt{2}}, \quad \ph = \frac{i(\hat{a}^\dagger - \hat{a})}{\sqrt{2}}
\end{equation}
The canonical commutation relations fix the vacuum noise convention, such that vacuum states, which we review shortly, satisfy
\begin{equation}
    \Delta x^2 = \Delta p^2 = \frac{1}{2} \ .
\end{equation}
The representation of the bosonic Hilbert space through phase space is obtained via understanding quadrature vectors as the Lie algebra of the Weyl-Heisenberg group. In particular, given a particular coordinate vector $\alpha \in \mathbb{R}^{2n}$, we can define the displacement (or Weyl) operator 
\begin{equation}
    D(r) = \exp(-ir^T \Omega \hat{R})
\end{equation}
which generates the coherent state $\ket{r} = D(r)\ket{0}$. In the basis of creation and annihilation operators, where coordinates are described in $\mathbb{C}^n$ rather than $\mathbb{R}^{2n}$, the displacement operator may be written
\begin{equation}
    D(\alpha) = \exp(\alpha^\dagger \hat{A} -\hat{A}^\dagger \alpha) \ ,
\end{equation}
where $\hat{A} = (a_1, a_2, ..., a_n)^T$. Simply put, the kernel $\alpha^\dagger \hat{A} - \hat{A}^\dagger\alpha$, using complex numbers, encodes the same algebra as the symplectic form, allowing us to trivially interchange  complex and real numbers within the symplectic algebra. We will often alternate between describing bosonic systems in terms of canonical quadrature operators and displacements in the creation/annihilation basis, where coherent states will be written as $\ket{\alpha}$. 

A well-known result in continuous-variable quantum mechanics is that the basis of coherent states is overcomplete, which allows us to represent any operator $\rho\in L^2(\mathbb{R}^n)$ with a phase-space distribution that is a function on $\mathbb{R}^{2n}$ (equipped with the symplectic inner product). The most common phase-space distribution is the Wigner function,
\begin{equation}
\label{eq:Wigner_func}
    W_\rho(\beta) = \frac{1}{\pi^n} \int d^{2n}\alpha \ \exp(\beta^\dagger \alpha - \alpha^\dagger \beta)\  \tr(\rho D(\alpha)) \ .
\end{equation}
Crucially, the map $\rho\rightarrow W_\rho$ is invertible, so phase space gives a faithful representation of any bosonic operator on the full Hilbert space. 

For much of this work, it is convenient to move between phase-space representations using characteristic functions. Given a state $\rho$, we define its (Weyl) characteristic function as
\begin{equation}
    \chi_\rho(\alpha) = \tr(\rho D(\alpha)),
\end{equation}
where $D(\alpha)$ is the displacement operator in the creation-annihilation basis. The characteristic function is the symplectic Fourier transform of the Wigner function, which is immediate from the definition in Equation \ref{eq:Wigner_func}. Once again, the pair $(W_\rho,\chi_\rho)$ uniquely determine one another. For Gaussian states, which we review next, $\chi_\rho$ is a Gaussian function of $\zeta$ and is completely specified by $(\bm m,\bm\Sigma)$.

Two transformation rules will be used repeatedly. First, displacements act on characteristic functions by a phase factor:
\begin{equation}
    \chi_{D(r)\rho D(r)^\dagger}(\zeta) = \exp(\zeta^\dagger \alpha_r - \alpha_r^\dagger \zeta)\chi_\rho(\zeta),
\end{equation}
where $\alpha_r \in \mathbb{C}^n$ is the complex displacement corresponding to $r \in \mathbb{R}^{2n}$. Second, symplectic Gaussian unitaries act by a linear change of variables in Fourier space:
\begin{equation}
    \chi_{U_S \rho U_S^\dagger}(\zeta) = \chi_\rho(\zeta'),
\end{equation}
where $\zeta'$ denotes the transformed phase-space coordinate induced by $S$. These identities make explicit why characteristic functions are natural for analyzing Gaussian processes and Gaussian measurements in the following sections.

\vspace{-1em}
\subsubsection{Gaussian states and squeezing}\label{subsubsec:Gaussian_states}
\vspace{-1em}

A central class of continuous-variable states are \textit{Gaussian states}, defined as those whose phase-space quasiprobability representations are Gaussian functions on $\mathbb{R}^{2n}$. To state this precisely, we define the first and second moments of a bosonic state $\rho\in L^2(\mathbb{R}^n)$ in the quadrature representation. The mean vector is
\begin{equation}
    \bm m = (m_1,\dots,m_{2n})^T, \quad m_j = \tr(\rho \hat{R}_j)
\end{equation}
and the covariance matrix is
\begin{equation}
    \bm\Sigma_{jk} = \frac{1}{2}\tr\!\left(\rho\left\{\hat{R}_j - m_j, \hat{R}_k - m_k\right\}\right),
\end{equation}
where $\{A,B\} = AB + BA$ denotes the anticommutator. When $\rho$ is a Gaussian state, its Wigner function is a Gaussian density on phase space:
\begin{equation}
    W_\rho(r) = \mathcal{N}(r;\bm m,\bm\Sigma),
\end{equation}
where, for $r \in \mathbb{R}^{2n}$,
\begin{equation}
    \mathcal{N}(r;\bm m,\bm\Sigma) = \frac{\exp\!\left(-\frac{1}{2}(r-\bm m)^T\bm\Sigma^{-1}(r-\bm m)\right)}{(2\pi)^n \sqrt{\det(\bm\Sigma)}}.
\end{equation}
Thus, a Gaussian state is fully specified by its mean vector and covariance matrix. When utilizing complex coordinates, we will describe Gaussian states with Wigner functions which are \textit{complex} Gaussian distributions, e.g. for $\nu>0$, we let $\mathcal{N}_\mathbb{C}(0, \nu)$ denote the distribution with complex density 
\begin{equation}
    z\rightarrow \frac{1}{\pi\nu}\exp\left(-\frac{|z|^2}{\nu}\right) \ .
\end{equation}
Equivalently, given a random variable $Z\sim \mathcal{N}_\mathbb{C}(0, \nu)$, $\Re(Z), \Im(Z)$ are independent Gaussians distributed as $\mathcal{N}(0, \nu/2)$. The only operational difference between the conventions is the factor of $2$ in the variance parameter. 

Gaussian states are preserved under Gaussian dynamics, including Gaussian unitaries and Gaussian channels. A Gaussian unitary is generated by a Hamiltonian at most quadratic in the quadratures (e.g. a single-mode $H$ which consists only of terms $\xh^2, \ph^2, \xh\ph, \ph\xh, \xh, \ph$), and its action on phase space is an affine symplectic transformation. Concretely, there exist a symplectic matrix $S \in \mathrm{Sp}(2n,\mathbb{R})$ and a displacement vector $d \in \mathbb{R}^{2n}$ such that
\begin{equation}
    U^\dagger \hat{R} U = S\hat{R} + d.
\end{equation}
At the level of moments,
\begin{equation}
    \bm m \mapsto S\bm m + d, \quad \bm\Sigma \mapsto S\bm\Sigma S^T.
\end{equation}
This separates the transformation into a displacement portion, which affects only first moments, and a symplectic portion, which reshapes the covariance ellipsoid in phase space.

More generally, a Gaussian channel is a completely positive trace-preserving map that takes Gaussian states to Gaussian states. Any such channel acts affinely on moments: there exist real matrices $X \in \mathbb{R}^{2n \times 2n}$ and $Y \in \mathbb{R}^{2n \times 2n}$ with $Y \succeq 0$, and a displacement vector $d \in \mathbb{R}^{2n}$, such that
\begin{equation}
    \bm m \mapsto X\bm m + d, \quad \bm\Sigma \mapsto X\bm\Sigma X^T + Y.
\end{equation}
The matrix $X$ captures the coherent, symplectic-like part of the evolution (including loss, amplification, and mode-mixing), while $Y$ captures injected Gaussian noise. Complete positivity imposes the constraint
\begin{equation}
    Y + \frac{i}{2}\Omega - \frac{i}{2}X\Omega X^T \succeq 0,
\end{equation}
which reduces to $X \in \mathrm{Sp}(2n,\mathbb{R})$ and $Y=0$ in the unitary case. In this sense, both Gaussian unitaries and Gaussian channels admit a natural decomposition into displacement and a linear phase-space map, with channels adding an additional incoherent term.

The symplectic portion of a Gaussian unitary admits further structure. By the Bloch--Messiah (symplectic singular value) decomposition, any $S \in \mathrm{Sp}(2n,\mathbb{R})$ can be written as
\begin{equation}
    S = O_1 \begin{pmatrix}
        Z & 0 \\
        0 & Z^{-1}
    \end{pmatrix} O_2,
\end{equation}
where $O_1,O_2 \in \mathrm{Sp}(2n,\mathbb{R}) \cap O(2n)$ are passive interferometric transformations and $Z=\mathrm{diag}(z_1,\dots,z_n)$ with $z_j \ge 1$ encodes single-mode squeezing.
Operationally, this means that the nontrivial resource in Gaussian unitaries is the ability to squeeze, while passive optics implement orthogonal symplectic transformations that preserve energy. 

For a single mode, the action of a symplectic transformation is easiest to visualize by how it maps the unit covariance circle into an ellipse. In $\mathrm{Sp}(2,\mathbb{R})$, one can classify symplectic transformations by their induced trajectories in phase space, noting that a quadratic Hamiltonian can only generate three kinds of phase-space trajectories because the Heisenberg equations of motion reduce to a second-order differential equation with either real, zero, or complex eigenvalues. Passive rotations correspond to elliptic motion, generated by quadratic Hamiltonians of the form $\xh^2 + \ph^2$, and act by
\begin{equation}
    S_{\mathrm{rot}}(\theta) = \begin{pmatrix}
        \cos\theta & \sin\theta \\
        -\sin\theta & \cos\theta
    \end{pmatrix},
\end{equation}
which preserves the circular symmetry of the vacuum covariance. Shearing transformations correspond to parabolic motion, generated by Hamiltonians proportional to $\xh^2$ or $\ph^2$, and act by
\begin{equation}
    S_{\mathrm{shear}}(s) = \begin{pmatrix}
        1 & 0 \\
        s & 1
    \end{pmatrix}
    \quad \textnormal{or} \quad
    S_{\mathrm{shear}}(s) = \begin{pmatrix}
        1 & s \\
        0 & 1
    \end{pmatrix},
\end{equation}
which preserve phase-space area while tilting the covariance ellipse. Finally, squeezing corresponds to hyperbolic motion, generated by Hamiltonians proportional to $\xh\ph + \ph\xh$, and acts by opposite rescalings of conjugate quadratures:
\begin{equation}
    S_{\mathrm{sq}}(r) = \begin{pmatrix}
        e^{-r} & 0 \\
        0 & e^{r}
    \end{pmatrix}.
\end{equation}
In this hyperbolic case, the symplectic map squeezes the basis itself, contracting one quadrature axis while expanding the conjugate axis, and hence transforms the vacuum covariance as
\begin{equation}
    \bm\Sigma_{\mathrm{vac}} \mapsto S_{\mathrm{sq}}(r)\bm\Sigma_{\mathrm{vac}} S_{\mathrm{sq}}(r)^T
    = \frac{1}{2}\begin{pmatrix}
        e^{-2r} & 0 \\
        0 & e^{2r}
    \end{pmatrix}.
\end{equation}
Equivalently, the squeezed vacuum $S_{\mathrm{sq}}(r)\ket{0}$ satisfies
\begin{equation}
    \Delta x^2 = \frac{1}{2}e^{-2r}, \quad \Delta p^2 = \frac{1}{2}e^{2r}.
\end{equation}

\vspace{-1em}
\subsection{Homodyne, Heterodyne, and Bell measurements}
\label{sec:cv_measurements}
\vspace{-1em}

The two most ubiquitous measurement strategies in quantum optics are \textit{homodyne} and \textit{heterodyne} measurements. In this section, we make precise the concept that homodyne measurements enable shot-noise-limited measurements of a single quadrature axis, while heterodyne measurements simultaneously capture orthogonal quadratures with an additional noise cost. These measurements represent either extreme of the broader family of Gaussian, entanglement-free bosonic measurements, also known as generaldyne measurements, which interpolate between sharp estimates of single quadratures and noisy, simultaneous estimates of noncommuting quadratures. For the restricted-measurement separations discussed in this work, it will suffice to consider only the homodyne and heterodyne cases, as access to generaldyne measurements does not change the relevant asymptotics (namely, because measurement of noncommuting quadratures stays vacuum-noisy limited). Then, we review the concept of continuous-variable Bell measurement, the key tool in our algorithm.

\vspace{-1em}
\subsubsection{Homodyne measurements} 
\vspace{-1em}

Theoretically, homodyne measurement corresponds to projective measurement in the eigenbasis of a single quadrature. For a single mode and an angle $\theta$, define the rotated quadrature operator
\begin{equation}
    \hat{x}_\theta = \xh \cos\theta + \ph \sin\theta \ .
\end{equation}
Let $\ket{x_\theta}$ denote the generalized eigenstate satisfying $\hat{x}_\theta\ket{x_\theta} = x\ket{x_\theta}$. The homodyne POVM is
\begin{equation}
    \Pi_{\textnormal{hom}}^{\theta}(x) = \ketbra{x_\theta}{x_\theta}, \quad x\in \mathbb{R} \ ,
\end{equation}
with outcome distribution $P_\rho^\theta(x) = \tr(\rho \Pi_{\textnormal{hom}}^\theta(x))$. Homodyne measurements are thus the projection of the bosonic state onto a known phase-space axis. Identifying the state with its Wigner function, learning properties of a bosonic state equates to reconstructing properties of its Wigner function from the function's marginal.
\begin{fact}
\label{fact:homodyne_marginal}
    Let $\rho$ be any single-mode continuous-variable state with Wigner function $W_\rho(x,p)$. Then the outcome distribution of homodyne measurement at angle $\theta$ is the Wigner marginal along the axis conjugate to $\hat{x}_\theta$,
    \begin{equation}
        P_\rho^\theta(x) = \int_{\mathbb{R}} dp \ W_\rho(x\cos\theta - p\sin\theta, x\sin\theta + p\cos\theta) \ .
    \end{equation}
    In particular, if $\rho$ is Gaussian with mean $\bm m$ and covariance matrix $\bm\Sigma$, then $P_\rho^\theta(x)$ is a one-dimensional Gaussian with mean $\bm u_\theta^T \bm m$ and variance $\bm u_\theta^T \bm\Sigma \bm u_\theta$, where $\bm u_\theta = (\cos\theta,\sin\theta)^T$.
\end{fact}
\begin{proof}
    By the defining properties of the Wigner function, the probability density for measuring $\hat{x}$ is obtained by marginalizing over the conjugate phase-space coordinate,
    \begin{equation}
        P_\rho^0(x) = \bra{x}\rho\ket{x} = \int_{\mathbb{R}} dp \ W_\rho(x,p) \ .
    \end{equation}
    For general $\theta$, let $U_\theta$ denote the Gaussian unitary implementing the symplectic rotation $S_{\mathrm{rot}}(\theta)$, so that $U_\theta^\dagger \xh U_\theta = \hat{x}_\theta$. Measuring $\hat{x}_\theta$ on $\rho$ is equivalent to measuring $\hat{x}$ on $U_\theta \rho U_\theta^\dagger$, hence
    \begin{equation}
        P_\rho^\theta(x) = \int_{\mathbb{R}} dp \ W_{U_\theta \rho U_\theta^\dagger}(x,p) \ .
    \end{equation}
    Under a symplectic rotation, the Wigner function transforms by a linear change of variables, so $W_{U_\theta \rho U_\theta^\dagger}(x,p) = W_\rho(x\cos\theta - p\sin\theta, x\sin\theta + p\cos\theta)$, which gives the stated marginal formula. If $\rho$ is Gaussian then $W_\rho$ is Gaussian, and its marginal is again Gaussian with mean and variance given by the projected first and second moments, yielding the stated parameters.
\end{proof}
With these definitions, we can make precise the statement that homodyne measurements, while only capable of observing marginals of the full phase-space distribution, can do so with an uncertainty limited only by the intrinsic variance of the measured state.
\begin{fact}
\label{fact:homodyne_charfunc}
    Let $\rho$ be a Gaussian state with mean $\bm m$ and covariance matrix $\bm\Sigma$, and let $\bm u_\theta = (\cos\theta,\sin\theta)^T$. Then
    \begin{equation}
        P_\rho^\theta(x) = \mathcal{N}(x \ ; \bm u_\theta^T \bm m, \bm u_\theta^T \bm\Sigma \bm u_\theta),
    \end{equation}
    so the homodyne outcome variance is exactly the intrinsic quadrature variance of $\rho$ along $\bm u_\theta$, with no additional measurement noise penalty.
\end{fact}
\begin{proof}
    Consider the classical characteristic function of the homodyne outcomes,
    \begin{equation}
        \Phi_\theta(t) = \int_{\mathbb{R}} dx \ P_\rho^\theta(x)e^{itx} = \tr(\rho e^{it\hat{x}_\theta}) \ .
    \end{equation}
    Using $\hat{x}_\theta = \bm u_\theta^T \hat R$ and the definition $D(r)=\exp(-ir^T\Omega \hat R)$, note that for $r=t\Omega^T \bm u_\theta$ we have $r^T\Omega \hat R = -t\bm u_\theta^T \hat R = -t\hat{x}_\theta$, and hence $e^{it\hat{x}_\theta} = D(r)$. Therefore $\Phi_\theta(t)$ is equal to $\chi_\rho$ evaluated at the corresponding displacement parameter, which yields $\Phi_\theta(t)=\chi_\rho(t\bm u_\theta)$ under the stated identification of real and complex coordinates. For Gaussian $\rho$, $\chi_\rho$ is a Gaussian function, and restricting it to a one-dimensional line produces a one-dimensional Gaussian characteristic function with mean $\bm u_\theta^T \bm m$ and variance $\bm u_\theta^T \bm\Sigma \bm u_\theta$, yielding the stated form for $P_\rho^\theta(x)$.
\end{proof}
Physically, balanced homodyne detection is implemented by interfering the signal mode with a strong coherent local oscillator on a 50:50 beamsplitter and subtracting the photocurrents at the two output photodiodes. In the limit of a large local oscillator amplitude, the difference current is proportional to a single quadrature of the signal field, with the measured quadrature selected by the relative optical phase between the signal and the local oscillator \cite{LvovskyRaymer2009}. This setup is standard in optical experiments and has close analogues in microwave platforms, where an IQ mixer and phase-sensitive amplification play the role of the optical local oscillator and photocurrent subtraction \cite{Wallraff2004,Schuster2005,PhysRevLett.60.764,PhysRevA.39.2519,CastellanosBeltran2008,PhysRevB.83.134501,RevModPhys.82.1155}.

Because homodyne measurement is sharp on a single quadrature, it can directly resolve squeezed fluctuations when the squeezing axis is aligned with the measured quadrature. Concretely, if we can prepare a single-mode squeezed probe whose squeezed quadrature coincides with $\hat{x}_\theta$, then the homodyne outcome variance satisfies $\Var(\hat{x}_\theta) = (1/2)e^{-2r}$, yielding the familiar $e^{-2r}$ reduction in shot variance relative to vacuum \cite{Schnabel2017}.

In practice, however, performing squeezed homodyne measurements along many quadrature axes within the same experiment is challenging, even though the local oscillator phase can be tuned. High-quality squeezing is typically generated by a sub-threshold optical parametric oscillator together with active phase control that locks the squeezed quadrature angle to a chosen reference field, and maintaining this lock over long times requires dedicated coherent control fields and feedback loops \cite{Vahlbruch2006, Chelkowski2007, Schnabel2017}. As a result, rapidly changing the effective squeezing angle between many values amounts to rapidly retuning or re-locking these control systems while preserving low phase noise, which is widely regarded as technically demanding in precision implementations of squeezed homodyne readout. For practical sensing experiments, homodyne measurements on a single mode can generally only be made along a small, fixed set of axes, rendering properties which are not captured well by this fixed set of marginals difficult to learn. This will be crucial to our discussion of practical quantum advantage in Section \ref{sec:transport_QA}.

\vspace{-1em}
\subsubsection{Heterodyne measurements} 
\vspace{-1em}

Heterodyne measurements simply involve projective measurement in the basis of coherent states. For a single mode, the heterodyne POVM consists of elements of the form
\begin{equation}
    \Pi_{\textnormal{het}}(\alpha) = \frac{1}{\pi}\ketbra{\alpha}{\alpha}, \quad \alpha\in \mathbb{C}
\end{equation}
From this definition, the vacuum noise cost of heterodyne tomography becomes clear.
\begin{fact} \label{fact:heterodyne_noise}
    Let $\rho$ be any continuous-variable quantum state, and let $\chi_\rho(\zeta)$ be its characteristic function. Then the outcome distribution of performing heterodyne measurement on $\rho$ has classical characteristic function $\chi_\rho(\zeta)e^{-|\zeta|^2/2}$. For a Gaussian state $\rho$ with (complex) mean vector $\bm{m}$ and covariance matrix $\bm{\Sigma}$, this simplifies to the outcome distribution
    \begin{equation}
        Q_{\rho}(\alpha) = \mathcal{N}_\mathbb{C}(\alpha \ ; \bm{m}, \bm{\Sigma} + \mathds{1}/2)
    \end{equation}
\end{fact}
\begin{proof}
    We work with single-mode states; the extension to $n$ modes follows immediately. Note here that we are using the complex Gaussian distribution for convenience, which is operationally equivalent to the real Gaussian distribution over $\xh, \ph$ coordinates up to a factor of two in the vacuum variance as discussed in Section \ref{subsubsec:Gaussian_states}. By definition, the output distribution is
    \begin{equation}
        Q_\rho(\alpha) = \frac{1}{\pi}\bra{\alpha}\!\rho\!\ket{\alpha} \ .
    \end{equation}
    It is well-known that for any two density matrices $\rho, \sigma$ with Wigner functions $W_\rho, W_\sigma$,
    \begin{equation}
        \tr(\rho\sigma) = \pi \int d^2\beta \ W_\rho(\beta)W_\sigma(\beta) \ .
    \end{equation}
    Using this identity and the Wigner function of a coherent state, we see
    \begin{equation}
        Q_\rho(\alpha) = \frac{2}{\pi}\int d^2\beta \ W_\rho(\beta)e^{-2|\alpha-\beta|^2}
    \end{equation}

    For general, non-Gaussian states $\rho$, we thus see  that the phase-space distribution of heterodyne measurements is obtained by convolving the state's Wigner function with a mean-zero Gaussian vacuum noise. The classical characteristic function of $Q_\rho$ is obtained by the product of the characteristic functions of $W_\rho(\beta)$ and the vacuum noise penalty, which yields  $\chi_\rho(\zeta)e^{-|\zeta|^2/2}$. Hence, in Fourier space, the heterodyne outcome distribution is simply the state's characteristic function, blurred by vacuum noise whose magnitude is independent of the state. This is particularly simple for Gaussian states; if $\rho$ is Gaussian with mean $\bm m$ and covariance $\Sigma$, its Wigner function is precisely 
    \begin{equation}
        W_\rho(\beta) = \mathcal{N}_\mathbb{C}(\beta \ ;  \bm m, \bm \Sigma) \ .
    \end{equation}
    Then $Q_\rho(\alpha)$ is the convolution of two Gaussians. The vacuum Gaussian is mean-0, with covariance matrix $\mathds{1}$, so the convolution yields another Gaussian distributed as $\mathcal{N}_\mathbb{C}(\bm m, \bm \Sigma + \mathds{1}/2)$.  
\end{proof}
Because the vacuum noise penalty is independent of the measured state, heterodyne measurements can never surpass the vacuum noise floor even if squeezed probe states are prepared. This implies, for instance, that heterodyne measurements cannot estimate a single unknown displacement with shot variance scaling like $e^{-2r}/2$, in the manner that homodyne measurements can. In exchange, heterodyne measurements are \textit{informationally complete}: a set of heterodyne outcomes can uniquely identify any quantum state. This is easy to see from Fact \ref{fact:heterodyne_noise}, because the outcome distribution $Q_\rho$ is directly proportional to the characteristic function $\chi_\rho$, which uniquely specifies $\rho$. 

Physically, heterodyne detection is a joint quadrature measurement implemented by mixing the signal with an auxiliary vacuum mode on a beamsplitter, and performing two homodyne measurements on the two output ports, with local oscillator phases differing by $\pi/2$. This setup is often described as double-homodyne or eight-port homodyne detection. In optical settings, heterodyne can also be realized by beating the signal against a frequency-shifted local oscillator and demodulating the resulting photocurrent to obtain in-phase and quadrature components simultaneously.

\vspace{-1em}
\subsubsection{Two-mode squeezed probes and Bell measurements} \label{sec:BM_prelim}
\vspace{-1em}

Beyond homodyne and heterodyne, a third measurement primitive that plays a central role in this work is the continuous-variable analogue of Bell (EPR) basis measurement, which we occasionally abbreviate as BM, performed on a two-mode squeezed vacuum (TMSV) probe. The underlying EPR correlations and their use for joint quadrature readout have been studied since the earliest days of continuous-variable entanglement, and became standard tools in quantum optics through the development of continuous-variable teleportation and related protocols \cite{Einstein1935,BraunsteinKimble1998,Furusawa1998,BraunsteinvanLoock2005}. More recently, the BM strategy was highlighted as a powerful way to estimate properties of bosonic displacement channels \cite{jiang2024bellfix}; our work leverages this primitive to go beyond proof-of-concept demonstrations of information-theoretic quantum advantage and establish BM as a practical tool for learning properties of physically relevant classical-field signals.

We define the two-mode squeezing unitary on modes $1$ and $2$ as
\begin{equation}
    S_2(r) = \exp\!\left(r\hat{a}_1\hat{a}_2 - r\hat{a}_1^\dagger \hat{a}_2^\dagger\right),
\end{equation}
and the corresponding TMSV state as $\ket{\mathrm{TMSV}_r}=S_2(r)\ket{0,0}$. A continuous-variable Bell measurement is a joint measurement of the commuting EPR quadratures
\begin{equation}
    \hat{x}_{\mathrm{BM}} = \frac{1}{\sqrt{2}}(\xh_1-\xh_2), \quad \hat{p}_{\mathrm{BM}} = \frac{1}{\sqrt{2}}(\ph_1+\ph_2),
\end{equation}
which can be realized experimentally by interfering two squeezed modes on a balanced beamsplitter (e.g. performing parameteric down-conversion) and then performing homodyne detection on the above conjugate quadratures at the two outputs \cite{BraunsteinKimble1998,BraunsteinvanLoock2005}. On the TMSV state, these commuting observables are simultaneously squeezed:
\begin{equation}
    \Var(\hat{x}_{\mathrm{BM}})=\Var(\hat{p}_{\mathrm{BM}})=\frac{1}{2}e^{-2r},
\end{equation}
and a complex phase-space displacement $\zeta$ can be reconstructed via $\zeta = \hat{x}_{\mathrm{BM}} + i\hat{p}_{\mathrm{BM}}$. Note that in the complex convention, $\Var(\zeta) = e^{-2r}$, dropping the factor of $1/2$.

The effective POVM for this measurement is given by the joint spectral measure of the commuting observables $\hat{x}_{\mathrm{BM}}$ and $\hat{p}_{\mathrm{BM}}$, with elements
\begin{equation}
    \Pi_{\mathrm{BM}}(x,p) = \ketbra{x,p}{x,p},
\end{equation}
where $\ket{x,p}$ denotes the (generalized) simultaneous eigenstate satisfying $\hat{x}_{\mathrm{BM}}\ket{x,p}=x\ket{x,p}$ and $\hat{p}_{\mathrm{BM}}\ket{x,p}=p\ket{x,p}$.

If an unknown displacement acts on mode $1$, then the BM outcomes provide simultaneous unbiased estimates of both  noncommuting quadratures with additive noise set by the EPR variances above, circumventing the Heisenberg uncertainty principle by sacrificing variance two of the four conjugate quadratures which we choose not to measure. In this sense, BM combines the key operational benefits of homodyne and heterodyne: it yields a two-dimensional outcome per mode (and is thus informationally complete for learning phase-space properties of displacement channels), while achieving sub-vacuum uncertainty in both quadratures at once.

This contrasts sharply with heterodyne measurement. As shown in Fact \ref{fact:heterodyne_noise}, heterodyne outcomes are unavoidably blurred by an additive vacuum contribution of covariance $\mathds{1}$, independent of the probe state. Consequently, in any regime where the intrinsic uncertainty scale relevant for the inference task is well below the vacuum noise floor, heterodyne behaves as a strictly weaker version of Bell sampling, namely the same simultaneous-quadrature readout but with a fixed, irreducible noise penalty. These intuitions are made precise in the following proposition and remark.

\begin{proposition}[Entanglement-enabled Bell sampling as low-noise two-dimensional readout]
Fix $r \ge 0$. Prepare a two-mode squeezed vacuum state with squeezing $r$ across modes $A$ and $B$, apply a displacement $D(\alpha)$ with $\alpha \sim p$ to mode $A$, and perform the standard CV Bell measurement of the commuting EPR quadratures
\begin{align}
\hat{x}_- = \frac{\hat{x}_A - \hat{x}_B}{\sqrt{2}}, \qquad \hat{p}_+ = \frac{\hat{p}_A + \hat{p}_B}{\sqrt{2}}.
\end{align}
Let the complex-valued measurement record be $\zeta = \hat{x}_- + i \hat{p}_+$. Then conditioned on $\alpha$, $\zeta$ is distributed as
\begin{align}
\zeta = \alpha + Z, \qquad Z \sim \mathcal N_{\mathbb C}(0, e^{-2 r}),
\label{eq:bell-additive-noise}
\end{align}
and hence the unconditional outcome distribution is the Gaussian convolution $p * \mathcal N_{\mathbb{C}}(0, e^{-2 r})$.
\end{proposition}

\begin{proof}
For a two-mode squeezed vacuum state with parameter $r$, one has $\mathrm{Var}(\hat{x}_-) = \mathrm{Var}(\hat{p}_+) = e^{-2 r} / 2$. Displacing mode $A$ by $\alpha$ shifts $\hat{x}_A$ by $\sqrt{2} \, \Re(\alpha)$ and $\hat{p}_A$ by $\sqrt{2} \, \Im(\alpha)$, hence shifts $( \hat{x}_- , \hat{p}_+ )$ by $( \Re(\alpha) , \Im(\alpha) )$. Therefore $\zeta = \hat{x}_- + i \hat{p}_+$ has mean $\alpha$ and independent Gaussian fluctuations with variance $e^{-2 r} / 2$ in each quadrature, which is exactly $Z \sim \mathcal N_{\mathbb C}(0, e^{-2 r})$.
\end{proof}

\begin{remark}
Standard heterodyne readout yields an outcome $\zeta_{\mathrm{het}} = \alpha + Z_{\mathrm{het}}$ with $Z_{\mathrm{het}} \sim \mathcal N_{\mathbb C}(0,1)$, reflecting the vacuum noise penalty. The Bell-sampling model reduces the noise variance from $1$ to $e^{-2 r}$.
\end{remark}

Meanwhile, we have seen that homodyne measurement can match the $e^{-2r}$ variance reduction when squeezed probes are precisely aligned to the measured axis. However, these measurements can only access features of the phase space distribution which are captured in the marginals along the chosen homodyne axis. This observation is the basis for our theory of practical quantum advantage developed in App. \ref{sec:transport_QA} and the following examples in App. \ref{sec:prac_Gaussian_QA}.

\vspace{-1em}
\subsection{Review of optimal transport theory}
\label{sec:OT_overview}
\vspace{-1em}

Optimal transport endows the space of probability measures with a geometry that quantifies how much ``mass'' must be rearranged to transform one distribution into another (see e.g.~the encyclopedic reference~\cite{villani2008optimal}; for an introduction geared towards physicists, see~\cite{cotler2023renormalization}). In this work, the relevant measures live on bosonic phase space, so we freely identify $\C^n \simeq \R^{2n}$ and equip $\R^{2n}$ with the Euclidean norm $\|\cdot\|_2$.

Let $\mathcal{P}(\R^d)$ denote Borel probability measures on $\R^d$, and for $k \geq 1$ let $\mathcal{P}_k(\R^d)$ be the subset of measures with finite $k$th moment, $\int p(dx)\,\|x\|_2^k < \infty$. For $p,q \in \mathcal{P}(\R^d)$, write $\Pi(p,q)$ for the set of \emph{couplings} (transport plans) between $p$ and $q$, i.e.\ probability measures $\pi$ on $\R^d\times\R^d$ whose first and second marginals are $p$ and $q$. Equivalently, $\pi \in \Pi(p,q)$ if and only if there exist random variables $(X,Y)$ on a common probability space with $X\sim p$, $Y\sim q$, and $\pi = \mathrm{Law}(X,Y)$.

\begin{definition}[Wasserstein distances]\label{def:wasserstein}
Fix $k\in[1,\infty)$. For $p,q\in\mathcal{P}_k(\R^d)$, the $k$-Wasserstein distance is
\begin{align}
W_k(p,q) := \Bigg( \inf_{\pi\in\Pi(p,q)} \int_{\R^d\times\R^d} \pi(dx,dy)\,\|x-y\|_2^{\,k} \Bigg)^{1/k}.
\label{eq:Wk_primal}
\end{align}
Any optimizer $\pi^\star\in\Pi(p,q)$ {\rm (}which exists for $p,q\in\mathcal{P}_k(\R^d)${\rm )} is called an \emph{optimal transport plan}.
\end{definition}

\begin{remark}
The space $(\mathcal{P}_k(\R^d),W_k)$ is a Polish metric space, and $W_k(p_n,p)\to 0$ if and only if $p_n \to p$ weakly and $\int p_n(dx)\,\|x\|_2^k \to \int p(dx)\,\|x\|_2^k$. When $p,q$ are discrete, $W_1$ coincides with the classical earthmover distance.
\end{remark}

For a measurable map $T: \R^d \to \R^m$ and $p \in \mathcal{P}(\R^d)$, the \emph{pushforward} $T_\#p\in\mathcal{P}(\R^m)$ is defined by $(T_\#p)(A):=p(T^{-1}(A))$. If $p,\eta \in \mathcal{P}(\R^d)$, their convolution $p*\eta$ is the law of $X+Z$ for independent $X \sim p$ and $Z \sim \eta$.

\begin{lemma}[Wasserstein contractions]\label{lem:ot_contractions}
Let $k \geq 1$.
\begin{enumerate}
\item If $T:\R^d\to\R^m$ is $L$-Lipschitz, then for all $p,q\in\mathcal{P}_k(\R^d)$,
\begin{align}
W_k(T_\#p,T_\#q)\leq L\,W_k(p,q).
\label{eq:pushforward_contraction}
\end{align}
\item For any noise law $\eta \in \mathcal{P}_k(\R^d)$ and all $p,q\in\mathcal{P}_k(\R^d)$,
\begin{align}
W_k(p*\eta,q*\eta) \leq W_k(p,q).
\label{eq:convolution_contraction}
\end{align}
\end{enumerate}
\end{lemma}

\begin{proof}
For \eqref{eq:pushforward_contraction}, fix $\pi\in\Pi(p,q)$ and push it forward by $(T,T)$ to obtain
$\tilde\pi\in\Pi(T_\#p,T_\#q)$. Since $\|T(x)-T(y)\|_2\leq L\|x-y\|_2$,
\begin{align}
\int \tilde\pi(du,dv)\,\|u-v\|_2^k \leq L^k\int \pi(dx,dy)\,\|x-y\|_2^k\,.
\end{align}
Infimizing over $\pi$ yields \eqref{eq:pushforward_contraction}. For \eqref{eq:convolution_contraction}, take $(X,Y)\sim \pi \in \Pi(p,q)$ and $Z\sim\eta$ independent; then $(X+Z,Y+Z)$ is a coupling of $(p*\eta,q*\eta)$ with cost $\E\| (X+Z)-(Y+Z)\|_2^k=\E\|X-Y\|_2^k$. Infimizing over $\pi$ gives \eqref{eq:convolution_contraction}.
\end{proof}
\noindent The case $k=1$ admits a particularly clean dual characterization.
\begin{theorem}[Kantorovich--Rubinstein duality]\label{thm:KR}
For $p,q\in\mathcal{P}_1(\R^d)$,
\begin{align}
W_1(p,q) = \sup_{f:\,\Lip(f)\leq 1}
\Big(\E_{X\sim p}[f(X)]-\E_{Y\sim q}[f(Y)]\Big),
\label{eq:KR}
\end{align}
where $\Lip(f)$ is the smallest $L$ such that $|f(x)-f(y)|\leq L\|x-y\|_2$ for all $x,y$.
\end{theorem}

\begin{remark}[Immediate corollaries]\label{rem:KR_cor}
For any $L$-Lipschitz $f$ and $p,q\in\mathcal{P}_1(\R^d)$,
\begin{align}
\big|\E_p[f]-\E_q[f]\big|\leq L\,W_1(p,q).
\end{align}
In particular, taking $f(x)=u^\top x$ with $\|u\|_2\leq 1$ gives $\|\E_p[X]-\E_q[X]\|_2\leq W_1(p,q)$.
\end{remark}

When $p$ and $q$ are measures on $\R$, the optimal coupling is monotone and $W_k$ admits closed forms in terms
of quantiles. Let $F_p,F_q$ be the CDFs and $F_p^{-1},F_q^{-1}$ their generalized inverses. Then
\begin{align}
W_k(p,q) &= \Big(\int_0^1 dt\,|F_p^{-1}(t)-F_q^{-1}(t)|^k\Big)^{1/k},
\label{eq:quantile_formula} \\
W_1(p,q) &= \int_{\R} dx\,|F_p(x)-F_q(x)|.
\label{eq:cdf_formula}
\end{align}
These expressions are especially convenient when only one-dimensional marginals of a high-dimensional distribution are observable. For completeness, we record the closed form for $W_2$ between Gaussians. If $p=\mathcal{N}(m_0,\Sigma_0)$ and $q=\mathcal{N}(m_1,\Sigma_1)$ on $\R^d$, then
\begin{align}
W_2^2(p,q) = \|m_0-m_1\|_2^2 + \tr(\Sigma_0)+\tr(\Sigma_1) - 2\,\tr\Big( (\Sigma_0^{1/2}\Sigma_1\Sigma_0^{1/2})^{1/2}\Big).
\label{eq:W2_gaussian}
\end{align}

Next we study a variant of the Wasserstein distance that will be useful for us. Let $u\in \mathbb{S}^{d-1}$, let $\pi_u : \R^d \to \R$ be the projection $\pi_u(x)=u^\top x$. For $k\geq 1$ and $p,q\in\mathcal{P}_k(\R^d)$, the \emph{sliced} Wasserstein distance is
\begin{align}
\mathrm{SW}_k(p,q) := \Bigg(\int_{\mathbb{S}^{d-1}} du\,W_k\big((\pi_u)_\#p,(\pi_u)_\#q\big)^{k} \Bigg)^{1/k},
\label{eq:SW_continuous}
\end{align}
where $du$ denotes the uniform measure on the sphere. For a finite set of directions
$\mathcal U=\{u_1,\dots,u_m\}\subset \mathbb{S}^{d-1}$, define the discrete analogue
\begin{align}
\mathrm{SW}_{k,\mathcal U}(p,q) :=
\Bigg(\frac{1}{m}\sum_{j=1}^m W_k\big((\pi_{u_j})_\#p,(\pi_{u_j})_\#q\big)^k \Bigg)^{1/k}.
\label{eq:SW_discrete}
\end{align}
Since each $\pi_u$ is $1$-Lipschitz, Lemma~\ref{lem:ot_contractions} implies $\mathrm{SW}_k(p,q)\leq W_k(p,q)$ and $\mathrm{SW}_{k,\mathcal U}(p,q)\leq W_k(p,q)$.

For our analysis in this paper, we will see that restricted-angle homodyne data naturally controls distances of the form $\mathrm{SW}_{k,\mathcal U}$ (after a known Gaussian convolution), while Bell-sensing measurements provide direct access to the full $d$-dimensional geometry captured by $W_k$.

\section{Quantum Signal Learning}
\label{sec:QSL_protocol}

In this section we formalize Quantum Signal Learning (QSL) and the Bell QSL protocol that underlies our main-text results. 
QSL is designed to capture a practical regime common in quantum-noise-limited platforms: an unknown classical field induces a distribution of phase-space responses, and one wishes to estimate many physically meaningful properties of this distribution
from a classical measurement record.

\vspace{-1em}
\subsection{The Quantum Signal Learning problem}
\label{sec:QSL_problem}
\vspace{-1em}
First, we collect definitions required for our estimator.

\begin{definition} [Bell smoothing operator]
    Let $\varphi_r$ denote the density $\mathcal{N}_{\mathbb{C}}(0, e^{-2r}\mathds{1}_n)$ over $\mathbb{C}^n$, and let $p * q$ denote the convolution between two densities over $\mathbb{C}^n$,
    \begin{equation}
        (p * q)(\xi) = \int d^{2n}\alpha \ p(\alpha) q(\xi - \alpha) \ .
    \end{equation}
    Then the Bell smoothing operator $T_r$ is the Gaussian convolution acting on functions $f:\mathbb{C}^n\rightarrow \mathbb{C}$ as 
    \begin{equation}
        (T_rf)(\alpha) = (f*\varphi_r)(\alpha) \ .
    \end{equation}
\end{definition}

Akin to how the Husimi Q function of a quantum state yielded its heterodyne outcome distribution in Fact \ref{fact:heterodyne_noise}, the Bell smoothing operator allows us to compute the outcome distribution over samples of a function on phase space drawn from Bell measurement.

Next, we see that any linear Hamiltonian can be described by a classical probability density over phase-space displacements, which will allow us to formalize a property-learning problem as learning characteristics of this density.

\begin{lemma}
\label{lemma:ham_to_rdc}
    Consider any $N$-mode Hamiltonian of the form
    \begin{equation}
        H(t) = \sum_{i=1}^N f_{x, i}(t) \hat{x}_i + f_{p, i}(t)\hat{p}_i \ ,
    \end{equation}
    i.e. consisting of only linear terms in single-mode quadrature operators with arbitrary time-dependent (or stochastic) coefficients. We call such Hamiltonians linear. Let $\mathcal{E}_H^{(T)}$ denote the quantum channel corresponding to time evolution under $H(t)$ up to a time $T$. Then the action of $\mathcal{E}_H$ on an $N$-mode density matrix $\rho$ can always be written in the form
    \begin{equation}
        \mathcal{E}_H^{(T)}(\rho) = \int d^{2N}\alpha \ P_H^{(T)}(\alpha)\disp(\alpha)\rho\disp(\alpha)^\dagger \ .
    \end{equation} 
\end{lemma}
\begin{proof}
We first treat a fixed, deterministic realization of the functions $\{f_{x,i}(t),f_{p,i}(t)\}_{i=1}^N$, and then average over randomness.

Write the quadrature vector as $\hat R = (\xh_1,\dots,\xh_N,\ph_1,\dots,\ph_N)^T \in \mathbb{R}^{2N}$, and define the classical coefficient vector
\begin{equation}
    F(t) = (f_{x,1}(t),\dots,f_{x,N}(t),f_{p,1}(t),\dots,f_{p,N}(t))^T \in \mathbb{R}^{2N}.
\end{equation}
Then the Hamiltonian can be written compactly as
\begin{equation}
    H(t) = F(t)^T \hat R \ .
\end{equation}
Let $U(T)$ denote the time-ordered unitary generated by $H(t)$,
\begin{equation}
    U(T) = \mathcal{T}\exp\!\left(-i\int_0^T dt \ H(t)\right) \ .
\end{equation}

The Heisenberg equation of motion for $\hat R$ is
\begin{equation}
    \frac{d}{dt}\hat R(t) = i[H(t),\hat R(t)] \ .
\end{equation}
Using $H(t)=\sum_{j=1}^{2N}F_j(t)\hat R_j$ and the canonical commutation relations $[\hat R_j,\hat R_k]=i\Omega_{jk}$, we obtain
\begin{equation}
    i[H(t),\hat R_k(t)] = i\sum_{j=1}^{2N}F_j(t)[\hat R_j(t),\hat R_k(t)]
    = i\sum_{j=1}^{2N}F_j(t)i\Omega_{jk}
    = -(\Omega^T F(t))_k \ .
\end{equation}
Equivalently, in vector form,
\begin{equation}
\label{eq:heis_linear_ham}
    \frac{d}{dt}\hat R(t) = \Omega F(t) \ ,
\end{equation}
where we used $\Omega^T=-\Omega$. Importantly, the right-hand side is a $c$-number vector, so we obtain:
\begin{equation}
\label{eq:R_shift}
    \hat R(T) = \hat R(0) + r(T), \quad r(T) = \int_0^T dt \ \Omega F(t) \in \mathbb{R}^{2N}.
\end{equation}

By the defining property of the Weyl displacement operator $\disp(r)=\exp(-ir^T\Omega \hat R)$, one has
\begin{equation}
\label{eq:disp_conj}
    \disp(r)^\dagger \hat R \disp(r) = \hat R + r \ .
\end{equation}
Comparing Eq.~\eqref{eq:R_shift} and Eq.~\eqref{eq:disp_conj}, we conclude that the unitary $U(T)$ implements a displacement by $r(T)$ in the Heisenberg picture. Hence, for any state $\rho$ the Schr\"odinger-picture evolution is simply
\begin{equation}
\label{eq:deterministic_channel}
    \mathcal{E}_H^{(T)}(\rho) = U(T)\rho U(T)^\dagger = \disp(r(T))\rho \disp(r(T))^\dagger \ .
\end{equation}
Decompose $r=(r_x,r_p)$ with $r_x,r_p\in\mathbb{R}^N$, and define the corresponding complex displacement amplitude $\alpha\in\mathbb{C}^N$ by
\begin{equation}
\label{eq:r_to_alpha}
    \alpha = \frac{1}{\sqrt{2}}(r_x + i r_p).
\end{equation}
Using the isomorphism between annihilation-operator coordinates and phase space, we see that the displacement operators are identical: 
\begin{equation}
    D(r) = \exp(-ir^T\Omega \hat R) = \exp(\hat A^\dagger\alpha- \alpha^\dagger \hat A) = D(\alpha),
\end{equation}
so Eq.~\eqref{eq:deterministic_channel} can be written as $\mathcal{E}_H^{(T)}(\rho)=D(\tilde{\alpha}(T))\rho D(\tilde{\alpha}(T))^\dagger$ with
\begin{equation}
\label{eq:ham_displacement}
    \tilde{\alpha}_i(T) = \frac{1}{\sqrt{2}}\left(\int_0^T dt \ f_{p,i}(t) - i\int_0^T dt \ f_{x,i}(t)\right). 
\end{equation}

Thus, for a deterministic signal the distribution $P_H^{(T)}$ is a delta function supported at the realized displacement:
\begin{equation}
\label{eq:delta_P}
    P_H^{(T)}(\alpha) = \delta(\alpha-\tilde{\alpha}(T)).
\end{equation}

Now suppose the coefficient functions $F(t)$ are drawn from an arbitrary distribution over trajectories, or more generally depend on an unobserved random variable $\omega$. For each realization $\omega$, the induced evolution is the unitary displacement $\disp(r_\omega(T))$ with
\begin{equation}
    r_\omega(T) = \int_0^T dt \ \Omega F_\omega(t).
\end{equation}
which maps to $\tilde{\alpha}_\omega(T)$ as before. Averaging over realizations of the stochastic signal, the resulting quantum channel is
\begin{equation}
    \mathcal{E}_H^{(T)}(\rho) = \mathbb{E}_\omega\!\left[\disp(\tilde{\alpha}_\omega(T))\rho \disp(\tilde{\alpha}_\omega(T))^\dagger\right]. \label{eq:ham_channel}
\end{equation}
This is a convex mixture of displacements and can be written as an integral over phase space by pushing forward the trajectory distribution to a distribution over endpoint displacements. Let $P_H^{(T)}$ denote the induced probability density on $\mathbb{C}^{N}$ of the random vector $\tilde{\alpha}_\omega(T)$. Then
\begin{equation}
    \mathcal{E}_H^{(T)}(\rho) = \int_{\mathbb{R}^{2N}} d^{2N}\alpha \ P_H^{(T)}(\alpha)\disp(\alpha)\rho\disp(\alpha)^\dagger \ ,
\end{equation}
which is the claimed form. In particular, Eq.~\eqref{eq:delta_P} is recovered when the field is deterministic, while more general distributions over the coefficients $\{f_{x,i}(t),f_{p,i}(t)\}$ induce a correspondingly broader $P_H^{(T)}$ that is fully determined by the distribution of the time-integrated drive in Eq.~\eqref{eq:R_shift}.
\end{proof}

\begin{remark}[Linear Hamiltonians correspond to classical fields.] Outside the strong-coupling limit, experiments in which quantum fields couple to bosonic probes can be modeled, to leading order, by coupling terms of the form $a_{\mathrm{Field}}^\dagger b_{\mathrm{Probe}}$, and other relevant pairwise combinations. Conceptually, taking the classical limit of the interaction corresponds to replacing the field operator $a_{\mathrm{Field}}$ with an expectation value determined by classical field amplitude and phase terms $\langle a_{\mathrm{Field}}\rangle$, which constitutes a possibly time-dependent c-number coupling in front of a linear probe operator. This is why linear Hamiltonians precisely correspond to usual classical signal interactions.
\end{remark}

Through Lemma \ref{lemma:ham_to_rdc}, we can associate bounded-time evolution under a linear Hamiltonian to a quantum channel generated by a distribution over phase-space displacements. Learning properties of the signal then equates to estimating a linear functional of this distribution, for any property recoverable from the Bell-smoothed outcome distribution.

\begin{definition}[Property of a classical-field signal] 
Fix an $N$-mode linear Hamiltonian $H$ generated by an unknown classical field and an evolution time $T$. Moreover, fix a model class $\mathcal{P}$ of distributions, such that any $P_H^{(T)}\in\mathcal{P}$. A property of the signal is the linear functional
\begin{equation}
    \Psi\!\left(P_H^{(T)}\right) = \int d^{2N}\alpha \  P_H^{(T)}(\alpha)\psi(\alpha) 
    \label{eq:property_defn}
\end{equation}
where we refer to the function $\psi(\alpha)$ as the property kernel. The property kernel must satisfy that
\begin{enumerate}
    \item \textnormal{($\psi$ can be recovered from the measurement outcomes)} There exists a measurable function $g_\psi$ such that $$T_rg_\psi = \psi$$
    \item \textnormal{(The recovery map has bounded variance for all relevant signals)} For every $P\in \mathcal{P}$, 
    \begin{equation}
    \mathbb{E}_{\xi\sim(P*\varphi_r)}\left[|g_\psi(\xi)|^2\right] < \infty \ .
    \end{equation}
\end{enumerate}
\end{definition}
With this, we define the core sensing task studied in this work. 
\begin{definition}[Quantum Signal Learning] \label{def:QSL_formal}
Fix an $N$-mode linear Hamiltonian $H$ and an evolution time $T$, and a property $\Psi$. The Quantum Signal Learning problem \textnormal{\textsc{QSL}}$(H, T, \Psi, \epsilon, \delta)$ is to estimate $\Psi(P_H^{(T)})$ to within absolute error $\epsilon$, with success probability at least $1-\delta$, given query access to $\mathcal{E}_{H}^{(T)}$. 
\end{definition}
QSL is a particularly natural formulation for quantum sensing. In typical quantum-noise-limited platforms, prior uncertainty, stochastic fields, and drift naturally appear as randomness in the linear coefficients $\{f_{x,i}(t),f_{p,i}(t)\}$, so that both the signal and ambient fluctuations are absorbed into the induced displacement law $P_H^{(T)}$ without requiring a separate noise model. This differs from Hamiltonian learning (which aims to reconstruct a detailed generator description) and from QFI-style metrology (which optimizes for a fixed parametric family): QSL instead targets physically meaningful properties of the response distribution, such as correlations, matched-filter scores, and low-order moments, which can be estimated post-hoc from the same measurement record.

Still, Definition \ref{def:QSL_formal} includes as a special case the canonical metrology problem of single-parameter estimation of a deterministic signal, assuming additional information about the structure of $H$ is also provided. In practice, however, QSL is intended for the setting in which a sensing experiment produces a \emph{dataset} and one subsequently asks many downstream questions about an uncertain environment, rather than committing in advance to a single prespecified parameterization assumed to be known in advance. A core benefit of the algorithm introduced in the following section is that even while the formal statement of QSL allows for measurement bases tailored to specific choices of properties $\Psi$, Bell QSL collects a classical record that supports post-hoc measurement while estimating individual properties with optimal quadrature uncertainty.

Moreover, additional implementation imperfections in the Bell-sensing primitive (e.g.\ loss, excess Gaussian noise, or small classical readout errors) can be modeled as an extra known channel composed with $\mathcal{E}_H^{(T)}$, and therefore enter as additional smoothing of the observed distribution, leaving the estimator framework unchanged and allowing for simultaneous classical estimation of noise and signal. We discuss several examples of QSL in this work, highlighting physical properties which cannot be practically learned to high accuracy in the absence of our entanglement-enhanced approach.

\vspace{-1em}
\subsection{General form of Bell QSL algorithm}
\label{sec:alg}
\vspace{-1em}

In this section we present an algorithm which leverages the BM primitive and classical postprocessing to solve the QSL task with a sample complexity which asymptotically achieves at minimum the performance of the optimal entanglement-free strategy, and achieves substantial quantum advantage in practical sensing regimes. The protocol only requires shallow quantum information processing in the form of two-mode entangled state preparation and passive beamsplitters, making measurements in a static homodyne basis. Our algorithm also enables post-hoc learning of a signal's properties, as the collected dataset is informationally complete with respect to the channel applied by the signal, and can thus be used to classically reconstruct any desired property given enough samples. In Figure~\ref{fig:unsupervised}, we show that this dataset is robust for unstructured learning even beyond estimating specific properties with our estimators, and can be used as inputs to supervised or unsupervised machine learning algorithms.

First, we present the Bell QSL algorithm in full generality, which encompasses the learning algorithm for all physically reconstructible properties. However, at this stage, the estimator construction is ambiguous, because it is unclear how to invert an abstract property kernel under the Bell smoothing operator, as described shortly. After giving the general algorithm, we will give simple closed-form expressions for the estimator in the relevant cases.

Fix an instance \textsc{QSL}($H, T, \psi, \epsilon, \delta$), where $H$ is a linear Hamiltonian on $n$ modes. The Bell QSL algorithm proceeds as follows.
\begin{enumerate}
    \item \label{alg:QSL_step_1} \textit{Single-shot data collection}. Prepare the state $\ket{\mathrm{TMSV_r}}_{AB}^{\otimes n}$, where the subscript $AB$ denotes entangled single-mode subsystems. There are $n$ squeezed modes in the partition $A$ which share no entanglement with one another, and likewise with $B$; all entanglement lies across the partition. Evolve subsystem $A$ by $\mathcal{E}_H^{(T)}$. Perform BM (as in Section \ref{sec:BM_prelim}) jointly on $(A,B)$ by performing 50:50 beamsplitters on paired modes and measuring $\xh_1, \ph_2$ on the beamsplitter outputs, and collect a sample $\hat{\zeta}^{(1)} \in \mathbb{C}^n$. 
    \item \textit{I.I.D.\ repeated measurements}.
    Repeat $N$ times to collect samples $\bm \zeta = \{\hat{\zeta}^{(1)}, \hat{\zeta}^{(2)}, ... , \hat{\zeta}^{(N)}\}$. 
    \item \textit{Compute estimator.} Return the property estimator 
    \begin{equation} \label{eq:general_estimator}
        \hat{\Psi} = \frac{1}{N}\sum_{i=1}^N g_\psi(\zeta^{(i)}) \ .
    \end{equation}
    where $g_\psi$ is the inverse of the property kernel $\psi$ under the Bell smoothing operator $T_r$. 
\end{enumerate}

\vspace{-1em}
\subsection{Bell QSL for Fourier-defined kernels}
\label{sec:Fourier_alg}
\vspace{-1em}

It is relatively simple to obtain closed-form expressions for $g_\psi$ given $\psi$ when we analyze two complementary classes of property kernels: those with a well-defined symplectic Fourier transform, and those without. We consider the former case first.

\begin{theorem} [The estimator for Fourier-defined kernels] \label{thm:Fourier_estimator}
    Consider a Bell measurement dataset $\bm \zeta$ of size $N$ as collected in the Bell QSL algorithm. Suppose the property kernel $\psi \in L^1(\mathbb{C}^n)$ and let 
    \begin{equation}
        \hat{\psi}(\beta) = \frac{1}{\pi^{2n}} \int d^{2n}\alpha \ \psi(\alpha)\exp\left(\beta^\dagger\alpha - \alpha^\dagger\beta\right)
    \end{equation}
    denote its Fourier transform. Moreover, let $\chi_P$ denote the characteristic function of $P_H^{(T)}$ and let $\hat{\chi}_P$ denote its estimator given by
    \begin{equation}
        \hat{\chi}_{P}(\beta) = \exp\left(e^{-2r}|\beta|^2\right)\frac{1}{N}\sum_{i=1}^N \exp\left(\zeta^{(i)\dagger}\beta - \beta^\dagger\zeta^{(i)}\right) \ .
    \end{equation}
    Then if 
    \begin{equation}
        \int d^{2n}\beta \ |\hat{\psi}(\beta)|\exp(e^{-2r}|\beta|^2) < \infty \ , \label{eq:Fourier_condition}
    \end{equation}
    the following estimator
    \begin{equation}\label{eq:Fourier_estimator}
        \hat{\Psi} = \frac{1}{\pi^{2n}}\int d^{2n}\beta \ \hat{\psi}(\beta)\hat{\chi}_P(\beta) 
    \end{equation}
    is an unbiased estimator for $\Psi$. Note, this implies that 
    \begin{equation}
        g_\psi(\zeta) = \int d^{2n}\beta \ \hat{\psi}(\beta)\exp(e^{-2r}|\beta|^2)\exp\left(\zeta^\dagger\beta - \beta^\dagger\zeta\right) 
    \end{equation}
    is the inverse of the Bell-smoothing operator.
\end{theorem}
\begin{proof}
The intuition behind this estimator is simple: when the conditions that $\psi\in L^1(\mathbb{C}^n)$ and equation \eqref{eq:Fourier_condition} are satisfied, the desired property has a well-defined, bounded Fourier transform that can be recovered from the measurement data. Hence, we first go into Fourier space by constructing an estimator for the distribution's characteristic function, then convolve that against this well-defined Fourier transform in classical postprocessing to obtain an unbiased estimator that is equivalent to equation \eqref{eq:general_estimator} for Fourier-defined properties. This equivalence is simple to see: $g_\psi$ is the expression to compute the inverse fourier transform of $\hat{\psi}(\beta)$, but smoothed by squeezed shot noise. That is, when the Fourier transform is well defined, convolving it against the noise kernel when inverting it is equivalent to inverting the noise kernel. We now make these intuitions precise.

To prove unbiasedness, we utilize the fact that BM outcomes can be used to estimate the characteristic function of a displacement channel, as proven in \cite{jiang2024bellfix}.

\begin{lemma}[Theorem 1 of \cite{jiang2024bellfix}]
\label{lem:bell_identity}
Let $P(\alpha)\coloneqq P_H^{(T)}(\alpha)$, and let $\chi_P(\beta)$ denote its characteristic function obtained via the standard symplectic Fourier transform. Suppose $\{\zeta^{(1)}, \zeta^{(2)}, ..., \zeta^{(N)}\}$ are samples drawn from Step 1 of the Bell QSL algorithm. Then the estimator
\begin{equation}
    \hat{\chi}_P(\beta) 
    = \frac{1}{N}\sum_{i=1}^N \exp\!\left(e^{-2r}|\beta|^2\right)\exp\!\left(\zeta^{(i)\dagger}\beta - \beta^\dagger\zeta^{(i)}\right)
    \label{eq:lambda_tilde_def}
\end{equation}
satisfies $\E[\hat{\chi}_P(\beta)]=\chi_P(\beta)$ for all $\beta \in \mathbb{C}^n$.
\end{lemma}
This lemma allows us to convolve an estimator for the characteristic function against the property kernel in Fourier space to estimate the desired property.

Now, by the Fourier-definedness assumption, the symplectic Fourier transform $\hat{\psi}$ of $\psi$ exists as an absolutely integrable function, and moreover the weighted transform
\begin{align}
\beta \mapsto e^{e^{-2r}|\beta|^2}\hat{\psi}(-\beta)
\end{align}
is integrable, allowing us all standard integral manipulations including exchanging order of integration. By Fourier inversion under our convention,
\begin{equation}
    \psi(\alpha) = \frac{1}{\pi^{2n}}\int d^{2n}\beta \ \hat{\psi}(\beta)\exp\!\left(\beta^\dagger\alpha - \alpha^\dagger\beta\right).
    \label{eq:psi_inversion}
\end{equation}
Substituting this into the definition of $\Psi$ from Equation \eqref{eq:property_defn} and exchanging integrals yields
\begin{equation}
    \Psi(P) = \frac{1}{\pi^{2n}}\int d^{2n}\beta \ \hat{\psi}(\beta)\chi_P(\beta).
\end{equation}
On the other hand, by linearity of expectation and Lemma~\ref{lem:bell_identity},
\begin{equation}
    \E[\hat{\Psi}]
    = \frac{1}{\pi^{2n}}\int d^{2n}\beta \ \hat{\psi}(\beta)\E\!\left[\hat{\chi}_P(\beta)\right]
    = \frac{1}{\pi^{2n}}\int d^{2n}\beta \ \hat{\psi}(\beta)\chi_P(\beta)
    = \Psi(P) \ ,
\end{equation}
completing the proof. This further implies that the sample mean $\hat{\Psi} \coloneqq \frac{1}{N}\sum_{i=1}^N g_\psi(\zeta^{(i)})$ is unbiased.
\end{proof}

Having established unbiasedness, we prove a sample complexity guarantee in the Fourier-defined case.

\begin{lemma}[Continuation of Theorem 1 in \cite{jiang2024bellfix}]
\label{lem:pointwise_conc}
Fix $\beta\in\mathbb{C}^n$ and $\epsilon,\delta\in(0,1)$. Let $\hat{\chi}_P(\beta)$ be defined as in Eq.~\eqref{eq:lambda_tilde_def}. Then
\begin{equation}
    \Pr\!\left(| \hat{\chi}_P(\beta) - \chi_P(\beta) |\le \epsilon\right) \ge 1-\delta
\end{equation}
provided
\begin{equation}
    N \ \ge\ 8 \exp\!\left(2e^{-2r}|\beta|^2\right)\epsilon^{-2}\log\!\left(\frac{4}{\delta}\right).
    \label{eq:pointwise_N}
\end{equation}
\end{lemma}

\begin{corollary}[Sample complexity for estimating $\Psi(P)$]
\label{cor:property_sample_complexity}
Define the weighted Fourier norm
\begin{equation}
    \mathcal{K}_r(\psi) := \frac{1}{\pi^{2n}}\int d^{2n}\beta \ |\hat{\psi}(\beta)|\exp\!\left(e^{-2r}|\beta|^2\right),
    \label{eq:Kr_def}
\end{equation}
and assume $\mathcal{K}_r(\psi)<\infty$. Then
\begin{equation}
    \Pr\!\left(|\hat{\Psi}-\Psi(P)|\le \epsilon\right)\ge 1-\delta
\end{equation}
provided
\begin{equation}
    N \ \ge\ 8\,\mathcal{K}_r(\psi)^2\,\epsilon^{-2}\log\!\left(\frac{4}{\delta}\right).
    \label{eq:Fourier_sample_complexity_gen}
\end{equation}
\end{corollary}

\begin{proof}
Rewrite \eqref{eq:Fourier_estimator} by pulling the sample average outside the $\beta$-integral:
\begin{equation}
    \hat{\Psi} = \frac{1}{N}\sum_{i=1}^N Y(\zeta^{(i)}), \qquad
    Y(\zeta) := \frac{1}{\pi^{2n}}\int d^{2n}\beta \ \hat{\psi}(\beta)\exp\!\left(e^{-2r}|\beta|^2\right)\exp\!\left(\zeta^\dagger\beta-\beta^\dagger\zeta\right).
\end{equation}
Using $|\exp(\zeta^\dagger\beta-\beta^\dagger\zeta)|=1$, we have the uniform bound $|Y(\zeta)|\le \mathcal{K}_r(\psi)$ for all $\zeta$. Applying Hoeffding's inequality to $\Re(Y)$ and $\Im(Y)$ and then a union bound yields
\begin{equation}
    \Pr\!\left(|\hat{\Psi}-\E[\hat{\Psi}]|\ge \epsilon\right)\le 4\exp\!\left(-\frac{N\epsilon^2}{8\mathcal{K}_r(\psi)^2}\right),
\end{equation}
and substituting $\E[\hat{\Psi}]=\Psi(P)$ from Theorem~\ref{thm:Fourier_estimator} gives \eqref{eq:Fourier_sample_complexity_gen}.
\end{proof}

To illustrate the performance of Bell sensing for a simple and practical Fourier-defined kernel, we consider the example of estimating the characteristic function of the state itself. 

\begin{example}[Fourier-based estimator for the characteristic function]
    \textnormal{
    For some fixed $H,T$, let the property $\Psi = \chi_P(\beta_0)$ for a fixed $\beta\in \mathbb{C}^n$, where $\chi_P(\beta_0)$ is the characteristic function of $P_H^T$ evaluated at $\beta_0$. Then, choosing the property kernel
    \begin{equation}
        \psi_{\beta_0}(\alpha) = \exp\!\left(\alpha^\dagger \beta_0 - \beta_0^\dagger \alpha\right)
    \end{equation}
    immediately gives us $\Psi(P_H^T) = \chi_P(\beta_0)$ from equation \eqref{eq:Fourier_estimator}. Then
    \begin{equation}
        \hat{\psi}_{\beta_0}(\beta) = \frac{1}{\pi^{2n}}\int d^{2n}\alpha \exp\!\left(\alpha^\dagger \beta_0 - \beta_0^\dagger \alpha\right)\exp\!\left(\beta^\dagger \alpha - \alpha^\dagger \beta\right) = \pi^{2n}\delta(\beta-\beta_0) \ .
    \end{equation}
    Substituting this into the definition of $\mathcal{K}_r$, we see
    \begin{equation}
        \mathcal{K}_r(\psi_{\beta_0}) = \int d^{2n}\beta \ \delta(\beta-\beta_0) \exp\left(e^{-2r}|\beta|^2\right) = \exp(e^{-2r}|\beta_0|^2) \ .
    \end{equation}
    From \eqref{eq:Fourier_sample_complexity_gen}, we obtain a sample complexity of
    \begin{equation}
        N = O\left(\frac{\exp\!\left(2e^{-2r}|\beta_0|^2\right)}{\epsilon^2}\log\left(\frac{1}{\delta}\right)\right)
    \end{equation}
    to solve QSL. In this special case, we recover the result from \cite{jiang2024bellfix}. In physical settings, we expect each mode to have a frequency cutoff (as physical channels will not generally act with arbitrarily high energy on individual modes), and in this case we can take the bound $|\beta_0|^2 \leq n|b|^2$, where $b = O(1)$ is the per-mode frequency cutoff. This gives us  
    \begin{equation}
        N = O\left(\frac{\exp\!\left(2e^{-2r}n\right)}{\epsilon^2}\log\left(\frac{1}{\delta}\right)\right) \ , \label{eq:charfun_sample_complexity}
    \end{equation}
    illustrating that the exponential dependence on mode number is damped by squeezing. \cite{jiang2024bellfix} illustrates that this phenomenon cannot be matched by any entanglement-free strategy in the worst case. In section \ref{sec:worst_case_QA}, we show that this example leads to a worst-case quantum advantage for matched filtering, a prominent error-mitigation and signal-processing strategy in shot-noise-limited quantum experiments, including gravitational wave detection \cite{Allen2005ChiSquare, Allen2012FINDCHIRP}. 
    }
\end{example}

\vspace{-1em}
\subsection{Bell QSL for polynomial kernels}
\label{sec:distribution_alg}
\vspace{-1em}

When the Fourier-definedness conditions of section \ref{sec:Fourier_alg} are not satisfied but our properties remain physical, we fall into property kernels from the class of \textit{tempered distributions}. In this work, the main class of such kernels we will need are \emph{polynomial} kernels in phase-space coordinates (equivalently, derivatives of Dirac deltas in Fourier space). These kernels do not admit a well-behaved Fourier inversion of the form in equation \eqref{eq:Fourier_condition}, but they can still be inverted in closed form under the Bell smoothing operator because Gaussian convolution acts algebraically on polynomials. 

More simply, a natural and physically ubiquitous example is the covariance kernel $\hat{x}\hat{p}$, which when convolved against a (mean-0) single-mode phase-space distribution yields its mode covariance. While this example is trivial to estimate in practice by simply taking the sample covariance of BM outcomes, in this section we generalize to more complex tempered distribution such as higher moments. However, for intuition, we first expand on this representative example. 

\begin{example}[Estimating mode covariance with Bell QSL] \label{ex:gaussian_covar}
\normalfont
Consider a single-mode signal which induces a displacement distribution $P$ and write the displacement coordinate as \(\alpha = \alpha_x + i \alpha_p\), where \(\alpha_x = \Re(\alpha)\) and \(\alpha_p = \Im(\alpha)\). As usual, a Bell measurement dataset $\bm \zeta = \{\zeta^{(i)}\}_{i=1}^N$ obtained at squeezing parameter $r$ obeys the additive
noise model
\begin{equation}
\zeta = \alpha + Z, \qquad
Z = Z_x + i Z_p, \qquad
Z_x, Z_p \ \text{i.i.d.} \sim \mathcal{N}(0,\nu_r), \qquad
\nu_r := \frac{1}{2}e^{-2r} \ . \label{eq:covariance_additive_noise}
\end{equation}
where decomposing the complex $\alpha$ into real coordinates reintroduces the factor of $1/2$ in $\nu_r$. Assume for simplicity that $P$ is mean-zero so that the covariance equals the raw second moment
$\mathrm{Cov}(\alpha_x,\alpha_p) = \mathbb{E}[\alpha_x \alpha_p]$ (when this is not true, we need only estimate and subtract the product of the means of $\alpha_x$, $\alpha_p$). This corresponds to the property kernel
\begin{equation}
\psi_{\mathrm{cov}}(\alpha) := \alpha_x \alpha_p,
\qquad
\Psi(P) = \int d^2\alpha \, P(\alpha)\,\psi_{\mathrm{cov}}(\alpha) = \mathbb{E}_{\alpha \sim P}[\alpha_x \alpha_p].
\end{equation}
In this case there is an immediate recovery map: define \(g_{\mathrm{cov}}(\zeta) := \zeta_x \zeta_p\), where
\(\zeta_x := \Re(\zeta)\) and \(\zeta_p := \Im(\zeta)\). Then, conditioning on \(\alpha\) and using independence and
zero-mean of \((Z_x,Z_p)\),
\begin{equation}
\mathbb{E}\!\left[g_{\mathrm{cov}}(\zeta)\mid \alpha\right]
=
\mathbb{E}\!\left[(\alpha_x+Z_x)(\alpha_p+Z_p)\mid \alpha\right]
=
\alpha_x \alpha_p,
\end{equation}
so \(T_r g_{\mathrm{cov}} = \psi_{\mathrm{cov}}\) and the estimator
\begin{equation}
\widehat{\Psi}_{\mathrm{cov}}
=
\frac{1}{N}\sum_{i=1}^N \zeta_x^{(i)} \zeta_p^{(i)}
\end{equation}
is unbiased for \(\Psi(P)\) (equivalently, it is exactly \eqref{eq:general_estimator} with \(g_\psi = g_{\mathrm{cov}}\)).

Computing the sample complexity is simple. Unlike the Fourier-defined case, polynomial recovery maps are typically unbounded, so the uniform boundedness assumption used for Hoeffding-type bounds in Sec.~C.3 does not apply. Instead, we use standard concentration. Let \(Y := g_{\mathrm{cov}}(\zeta) = \zeta_x \zeta_p\). A direct expansion using \eqref{eq:covariance_additive_noise} gives
\begin{align}
\mathrm{Var}(Y) &=
\mathbb{E}\!\left[\zeta_x^2 \zeta_p^2\right] - \Big(\mathbb{E}[\alpha_x\alpha_p]\Big)^2 \nonumber =
\mathbb{E}\!\left[\alpha_x^2 \alpha_p^2\right]
+ \nu_r \mathbb{E}[\alpha_x^2]
+ \nu_r \mathbb{E}[\alpha_p^2]
+ \nu_r^2
- \Big(\mathbb{E}[\alpha_x\alpha_p]\Big)^2,
\end{align}
which is finite whenever \(P\) has finite fourth moment. In the special case where \(P\) is Gaussian with
\(\mathrm{Var}(\alpha_x)=\sigma_x^2\), \(\mathrm{Var}(\alpha_p)=\sigma_p^2\), and \(\mathrm{Cov}(\alpha_x,\alpha_p)=c\), one obtains
the closed form \(\mathrm{Var}(Y) = (\sigma_x^2+\nu_r)(\sigma_p^2+\nu_r)+c^2\), exhibiting the explicit improvement
as \(\nu_r \to 0\). In particular, for shot-noise-limited quantum experiments, we expect the parameters $\sigma_x, \sigma_p, c \lesssim \nu_r$, often comparable to or smaller than the squeezed-noise floor.

To obtain failure probability \(\delta\), we may simply use the sample mean estimator
\(\widehat{\Psi}_{\mathrm{cov}}=\frac{1}{N}\sum_{i=1}^N Y_i\) with \(Y_i=\zeta_x^{(i)}\zeta_p^{(i)}\). In the Gaussian case,
\(Y\) has sub-exponential tails, so a standard Bernstein concentration bound yields that for a universal constant
\(c>0\),
\begin{equation}
\Pr\!\left(\left|\widehat{\Psi}_{\mathrm{cov}}-\Psi(P)\right|>\epsilon\right) \le O\!\left(\exp\!\left(-cN\frac{\epsilon^2}{\mathrm{Var}(Y)}\right)\right),
\end{equation}
In the shot-noise-limited regime, where $\nu_r$ is the relevant variance scale, we thus obtain a sample complexity
\begin{equation}
N \sim O\!\left(\frac{e^{-4r}}{\epsilon^2}\log\!\left(\frac{1}{\delta}\right)\right)
\end{equation}
to estimate \(\Psi(P)\) to accuracy \(\epsilon\) with success probability at least \(1-\delta\). In the following appendices, we will see that even for this simple, physically ubiquitous instance of estimating a Gaussian covariance, conventional measurement schemes can incur sample complexity exponentially worse in $\epsilon$. 
\end{example}

We now generalize the preceding covariance example to arbitrary polynomial kernels. For clarity, we pass to real
coordinates. Given \(\alpha \in \mathbb{C}^n\), define the real vector
\begin{equation}
u(\alpha) := (\Re(\alpha_1),\dots,\Re(\alpha_n), \Im(\alpha_1),\dots,\Im(\alpha_n)) \in \mathbb{R}^{2n},
\end{equation}
Under Bell sensing, \(u(\zeta)=u(\alpha)+G\) is the sample quadrature vector, where \(G\sim \mathcal{N}(0,\nu_r I_{2n})\).
Therefore, for any measurable \(f:\mathbb{R}^{2n}\to \mathbb{C}\),
\begin{equation}
(T_r f)(u) = \mathbb{E}_{G\sim \mathcal{N}(0,\nu_r I_{2n})}\!\left[f(u+G)\right],
\end{equation}

The inverse \(T_r^{-1}\) is ill-posed on general function classes, but it is perfectly well-defined
on the finite-dimensional space of polynomials. The relevant basis is given by Hermite polynomials.
For \(k\in \mathbb{N}\), define
\begin{equation}
\mathrm{He}_k(t) := (-1)^k e^{t^2/2}\frac{d^k}{dt^k}e^{-t^2/2},
\end{equation}
and for a variance parameter \(\nu>0\), define the scaled Hermite polynomial
\begin{equation}
H_k^{(\nu)}(t) := \nu^{k/2}\,\mathrm{He}_k\!\left(\frac{t}{\sqrt{\nu}}\right).
\end{equation}
These satisfy the generating function identity
\begin{equation}
\exp\!\left(st - \frac{\nu s^2}{2}\right)
=
\sum_{k=0}^\infty \frac{s^k}{k!} H_k^{(\nu)}(t).
\label{eq:hermite_mgf}
\end{equation}

\begin{lemma}[Hermite inversion of Gaussian smoothing on monomials]
\label{lem:hermite_inversion}
Let \(X\in \mathbb{R}\) be arbitrary and let \(Z\sim \mathcal{N}(0,\nu)\) be independent. Then for every \(k\in \mathbb{N}\),
\begin{equation}
\mathbb{E}\!\left[H_k^{(\nu)}(X+Z)\mid X\right] = X^k.
\end{equation}
\end{lemma}

\noindent
\begin{proof}
Taking the conditional expectation of \eqref{eq:hermite_mgf} with \(t=X+Z\) and using the Gaussian identity \(\mathbb{E}[e^{sZ}]=e^{\nu s^2/2}\) gives
\begin{equation}
\mathbb{E}\!\left[\exp\!\left(s(X+Z)-\frac{\nu s^2}{2}\right)\middle | X\right]
= \exp\!\left(sX-\frac{\nu s^2}{2}\right)\mathbb{E}[e^{sZ}] = 
e^{sX}.
\tag{C26}
\end{equation}
Expanding both sides as a power series in $s$ gives 
\begin{equation}
    \sum_{k=0}^\infty \frac{s^k}{k!} \mathbb{E}[H_k^{(\nu)}(X+Z)|X] = \sum_{k=0}^\infty \frac{s^k}{k!}X^k
\end{equation}
and matching coefficients completes the proof.
\end{proof}

This lets us generalize to polynomials easily. Let \(\kappa=(\kappa_1,\dots,\kappa_{2n})\in \mathbb{N}^{2n}\) be a multi-index, and define
\begin{equation}
H_{\kappa}^{(\nu)}(u) := \prod_{j=1}^{2n} H_{\kappa_j}^{(\nu)}(u_j),
\qquad
u^\kappa := \prod_{j=1}^{2n} u_j^{\kappa_j}.
\end{equation}
These relations allow us to pick out a polynomial of a desired degree at each mode index, because upon taking expectations, Lemma \ref{lem:hermite_inversion} will precipitate the desired polynomial kernel at each mode. To be precise, because the Bell noise is i.i.d.\ across the $2n$ real coordinates, Lemma~\ref{lem:hermite_inversion} implies
\begin{equation}
\mathbb{E}\!\left[H_{\kappa}^{(\nu_r)}(u(\alpha)+G)\middle| u(\alpha)\right] = u(\alpha)^\kappa.
\label{eq:monomial_exp2}
\end{equation}
which allows us to generate any desired polynomial of the $n$ modes. We are now prepared to state the generalized polynomial estimator and prove its unbiasedness and sample complexity within the Bell QSL framework.

\begin{theorem}[The estimator for polynomial kernels]
\label{thm:tempered_kernels}
Let \(\psi\) be a polynomial property kernel in \(2n\) real variables, i.e.\ \(\psi(u)=\sum_{\kappa} c_\kappa u^\kappa\) for a
finite set of multi-indices \(\kappa\). Define the recovery map \(g_\psi\) by
\begin{equation}
g_\psi(\zeta)
:=
\sum_{\kappa} c_\kappa\, H_{\kappa}^{(\nu_r)}(u(\zeta)),
\qquad
\nu_r=\frac{1}{2}e^{-2r}.
\label{eq:polynomial_estimator}
\end{equation}
Then \(T_r g_\psi = \psi\). Consequently, the Bell estimator
\begin{equation}
\hat{\Psi}
=
\frac{1}{N}\sum_{i=1}^N g_\psi(\zeta^{(i)})
\end{equation}
is unbiased for \(\Psi(P)=\int d^{2n}\alpha\, P(\alpha)\psi(\alpha)\).
\end{theorem}

\noindent
\begin{proof}
For a monomial \(\psi(u)=u^\kappa\), the choice \(g_\psi(\zeta)=H_\kappa^{(\nu_r)}(u(\zeta))\) immediately
gives \(T_r g_\psi = \psi\) by our generating function identity resulting from Lemma \ref{lem:hermite_inversion}. Summing over \(\kappa\), we obtain \eqref{eq:polynomial_estimator}. Unbiasedness is then immediate from \eqref{eq:general_estimator}.
\end{proof}

Sample complexity follows from basic Chebyshev-type concentration. We simply assume bounds on the moments of $P$, which are set by the problem as in the covariance example, which immediately yields the desired sample bound. 

\begin{lemma}[Sample complexity for polynomial kernels with the mean estimator]
\label{lem:poly_mean_sc}
Let \(\psi\) be a polynomial kernel and let \(g_\psi\) be the corresponding recovery map defined in
Theorem~\ref{thm:tempered_kernels}. Let \(\zeta^{(1)},\dots,\zeta^{(N)}\) be i.i.d.\ Bell outcomes with squeezing $r$, and define the mean estimator
\begin{equation}
\widehat{\Psi}
=
\frac{1}{N}\sum_{i=1}^N g_\psi(\zeta^{(i)}).
\end{equation}
Assume there exists a variance bound \(V_r(\psi)\) such that
\begin{equation}
\sup_{P\in \mathcal{P}} \mathrm{Var}\left(g_\psi(\zeta)\right) \le V_r(\psi)
\end{equation}
where $\mathcal{P}$ is the model class for valid distributions. 
Then for every \(\epsilon>0\) and \(\delta\in(0,1)\),
\begin{equation}
\Pr\!\left(\left|\widehat{\Psi}-\Psi(P)\right|>\epsilon\right)
\le
\frac{V_r(\psi)}{N\epsilon^2}.
\end{equation}
which follows immediately from Chebyshev's inequality. In particular, it suffices to take \(N \ge V_r(\psi)/(\epsilon^2\delta)\) to ensure
\(\Pr(|\widehat{\Psi}-\Psi(P)|>\epsilon)\le \delta\).
\end{lemma}

\begin{remark}[Robust aggregation for heavy-tailed regimes]
Note that Lemma~\ref{lem:poly_mean_sc} yields a $1/\delta$ dependence on the failure probability, as it uses only a second-moment bound. For the sample mean estimator, this is fundamental because polynomial kernels can yield $g_\psi$ which are heavy-tailed, so despite finite variance we cannot hope for a universal sub-exponential bound. While the general statement of Bell QSL uses this sample-mean estimator, in such heavy-tailed settings it is preferable to use a robust aggregator such as the median-of-means estimator, which under the same finite-variance assumption
recovers the sharper high-probability scaling \(N = O\!\left(\frac{V_r(\psi)}{\epsilon^2}\log\!\left(\frac{1}{\delta}\right)\right)\). Thus, in practice, the standard $\epsilon^{-2}\log(1/\delta)$ scaling is always achievable.
\end{remark}

Finally, notice that the choice of
\begin{equation}
g_{u_j u_k}(\zeta) = u_j(\zeta)u_k(\zeta)\quad (j\neq k),
\end{equation}
is exactly the special case of covariance, so our polynomial-kernel estimator \eqref{eq:polynomial_estimator} indeed generalizes this instance.

\section{Quantum Advantage in Learning Noisy Classical Signals}
\label{sec:qa_theory}
In this section we establish complementary quantum advantage statements for QSL: first, an information-theoretic \emph{worst-case} advantage of Bell QSL for
post-hoc matched filtering/beamforming, in which scoring a template bank reduces to characteristic-function queries of the induced
displacement law; second, a theory of \emph{practical} advantage against experimentally restricted baselines, formalizing how limited
measurement access can make learning ill-conditioned even when the underlying phase-space structure is large. 
The latter is quantified via an optimal-transport stability modulus, which makes explicit the tradeoff between the size of an advantage regime and the strength of the achievable separation.

\vspace{-1em}
\subsection{Worst-case matched filtering and beamforming}
\label{sec:worst_case_QA}
\vspace{-1em}

Matched filtering and beamforming are ubiquitous primitives in classical sensing. One acquires a data record and compares it against a bank of candidate templates, assigning each template a scalar score that quantifies agreement with the data. In the simplest Gaussian model this score reduces, up to known normalizations, to a whitened inner product, and optimality follows from standard likelihood-ratio arguments. For our purposes, the key feature is operational rather than algorithmic. Data are collected once, while the template bank is explored later, often at large scale, so the sensing stage must produce a dataset that supports many post-hoc template queries.

In the bosonic setting studied here, a classical field induces a random displacement vector $\alpha\in\mathbb{C}^n$ across a mode basis, and one channel use produces a noisy record $\zeta$. A template $w\in\mathbb{C}^n$ plays the role of a beamforming vector and defines the matched-filter output $Y_w=\Re(\alpha^\dagger w)$. Rather than restricting to the linear statistic $Y_w$ itself, we view matched filtering as post-hoc scoring of the \emph{law} of $Y_w$ using a single dataset. In particular, we consider bounded Fourier features $\varphi_{Y_w}(t)=\mathbb{E}[e^{itY_w}]$, which are natural in likelihood-based pipelines and which admit clean post-hoc evaluation. The crucial point for us is that these Fourier scores reduce exactly to characteristic-function queries of the underlying displacement distribution, which is what allows worst-case separations to be imported from distribution-learning lower bounds.

A simple application of matched filtering is to estimate amplitudes of overlapping sinusoidal signals. For instance, consider a waveform 
\begin{equation}
    F(t) = \sum_{k=1}^n \sqrt{\omega_k} (A_k\cos(\omega_k t) + B_k\sin(\omega_kt)) \ .
    \label{eq:matched_filtering_f}
\end{equation}
As we saw in e.g. \eqref{eq:ham_displacement}, a Hamiltonian which realizes this drive acts on an $n$-mode bosonic probe as a sum over displacements $D(\alpha_k)$, with
\begin{equation}
    \alpha_k = -\frac{i}{\sqrt{2\omega_k}} \int_0^{2\pi} dt \ F(t)\exp(i\omega_kt) = \frac{\pi}{\sqrt{2}}(B_k - iA_k) \ .
    \label{eq:matched_filtering_iso}
\end{equation}
In matched filtering, it is generally presumed that parameters such as $A_k$ are static realizations of a sample from an underlying distribution specified by theoretical priors and uncertainties. This viewpoint implies that the displacement vector of $\alpha_k$'s, $\alpha = (\alpha_1, \alpha_2, ..., \alpha_n)$, is sampled from a distribution which we denote by $P(\alpha)$, in line with our previous notation. 

To learn these parameters $A_k, B_k$ via matched filtering, we first collect multichannel observations $\zeta$, given by 
\begin{equation}
    \zeta = \alpha + \eta, \qquad \alpha\sim P, \qquad \eta\sim \mathcal{N}_\mathbb{C}^n(0, \nu)
\end{equation}
where $\nu$ denotes the shot-noise variance. Classically, we keep a template bank of hypotheses for possible tuples $(A_1, B_1, ..., A_n, B_n)$, which can equivalently be stored as vectors $w\in\mathbb{C}^n$. For each template $w$, the underlying matched-filter output is given by the random variable $Y_w = \Re(\alpha^\dagger w)$ with $\alpha\sim P$.

If one restricts attention to linear scores based only on $Y_w$ (for example, estimating $\mathbb{E}[Y_w]$), then separable probes and measurements already achieve the standard $\Theta(\epsilon^{-2})$ sample complexity for additive Gaussian noise.  Indeed, given observations $\zeta=\alpha+\eta$ with $\eta$ Gaussian, the quantity $\Re(\zeta^\dagger w)$ is an unbiased noisy sample of $Y_w$, so $\mathbb{E}[Y_w]$ can be estimated by simple averaging with $O(\epsilon^{-2})$ samples.  Consequently, any superpolynomial worst-case separation in the dimension $n$ must arise from post-hoc nonlinear template scores that probe finer structure in the law
of $Y_w$, rather than only its first moment.

The classical post-processing then assigns to each template $w$ a score that is a chosen functional of the empirical distribution of the matched-filter outputs $\{\Re(\zeta^{(j)\dagger}w)\}_{j=1}^N$, and then selects the maximizer over the template bank. In the classical signal-processing literature, ``matched filtering'' is often used to mean the linear correlation statistic that is optimal under a Gaussian model after whitening.  Here we adopt a broader viewpoint that is common in likelihood-based pipelines, in which template scores are allowed to be nonlinear functionals of the matched-filter output distribution.  The Fourier scores we use below, $\varphi_{Y_w}(t)=\mathbb{E}[e^{itY_w}]$, form a canonical bounded family of such post-hoc queries, and they reduce exactly to characteristic-function evaluations. It is crucial that the collected dataset supports post-hoc evaluation to compute a score for every template in the bank, a property not arbitrarily satisfied by e.g. homodyne-based tomography.

Given that different frequencies in our coefficient-determination problem are independent in Fourier space, it is natural to consider a Fourier-based score. In many sensing pipelines, template scores are ultimately derived from likelihood ratios or Bayes factors rather than from a single linear statistic.  Even in the simplest Gaussian model, log-likelihoods take the form $\langle d, s_w\rangle - \frac{1}{2}\langle s_w,s_w\rangle$, and the resulting likelihood ratio involves an \emph{exponential} of a matched-filter inner product.  This observation motivates studying bounded exponential features of the matched-filter output as primitive post-hoc queries.  In particular, the complex Fourier feature $e^{itY_w}$ is a numerically stable, bounded surrogate that still captures the high-frequency structure that becomes
hard to access under vacuum-limited readout.

With the above in mind, let
\begin{equation}
    \varphi_{Y_w}(t) = \mathbb{E}[e^{itY_w}]
\end{equation}
be the one-dimensional characteristic function of $Y_w$. Then an unbiased estimator for $\varphi_{Y_w}(t)$ is
\begin{equation}
    \hat\varphi_{Y_w}(t) = \exp\left(\frac{\nu t^2\|w\|_2^2}{4}\right)N^{-1}\sum_{j=1}^N e^{it\Re(\zeta^{(j)\dagger}w)}
\end{equation}
where the Gaussian prefactor is required to remove the biased caused by shot noise, and any score that admits a Fourier representation in $Y_w$ can be expressed in terms of $\varphi_{Y_w}(t)$. Notice, however, that 
\begin{equation}
    \mathbb{E}[e^{itY_w}] = \int d^{2n}\alpha \ P(\alpha)\exp(\alpha^\dagger (itw/2) - (itw/2)^\dagger\alpha) = \chi_P\left(\frac{itw}{2}\right) \ ,
\end{equation}
indicating that a matched-filter query $(w, t)$ is equivalent to evaluating the characteristic function of $P$ at a point $\beta = itw/2$. This is the key identity which bridges post-hoc matched filtering with post-hoc learning of $\chi_P(\beta)$. We make this precise in the following. 

\begin{fact} \label{fact:filtering_bank} Let $F(t)$ be a multimodal waveform and $P$ be its induced displacement distribution. Fix a template bank $W$ of size $M$. Then for any $t\in (0, 2\pi]$, the problem of scoring all templates $w\in W$ reduces to estimating the characteristic function $\chi_P$ at a bank of points $\mathcal{Q} = \{\beta_1,...,\beta_M\}$ with $\beta_i = itw_i/2$.
\end{fact}

Leveraging this observation, we can view the matched filtering problem backwards, by discussing the equivalent characteristic-function learning problem and fixing a template bank of points. With this, we are ready to demonstrate that there is a worst-case instance of distributional priors for the physically ubiquitous family of resonant waveforms that results in a superpolynomial advantage of Bell measurement over any unentangled strategy for data collection in matched filtering.

Having reframed the problem in terms of characteristic-function learning, our proof relies on results from \cite{jiang2024bellfix} which define a family of displacement distributions that are hard to learn without entanglement.

\begin{definition}[Hard family of distributions] Fix $\sigma > 0$ and $\epsilon_0 \in [0, 1/4]$. For each $\gamma \in \mathbb{C}^n$, define
\begin{equation}
    P_\gamma(\alpha) = \left(\frac{2\sigma^2}{\pi}\right)^n e^{-2\sigma^2|\alpha|^2} \left(1 + 4\epsilon_0\sin(2\Im(\gamma^\dagger \alpha))\right)
\end{equation}
This is a valid probability density since the sine term is bounded and the Gaussian factor normalizes the density, and physically corresponds to a Gaussian density modulated by oscillations of frequency controlled by $\gamma$. \cite{jiang2024bellfix} demonstrates that
\begin{equation}
    \chi_{P_\gamma}(\beta) = \exp\left(-\frac{|\beta|^2}{2\sigma^2}\right) -2i\epsilon_0 \exp\left(-\frac{|\beta-\gamma|^2}{2\sigma^2}\right) + 2i\epsilon_0\exp\left(-\frac{|\beta+\gamma|^2}{2\sigma^2}\right) \,
\end{equation}
a function with three Gaussian bumps centered at $\beta=0$ and $\beta=\pm \gamma$.
\end{definition}
\noindent As one simple diagnostic, differentiating the closed form for $\chi_{P_\gamma}$ at $\beta=0$ yields a mean displacement whose magnitude scales as
\begin{align}
\big\|\mathbb{E}_{\alpha\sim P_\gamma}[\alpha]\big\|_2 = O\!\left(\frac{\epsilon_0}{\sigma^2}\,\|\gamma\|_2\,e^{-|\gamma|^2/(2\sigma^2)}\right).
\end{align}
In the regime $|\gamma|^2=\Theta(n)$, this dependence is exponentially small in $n$, so any mean-based linear matched-filter statistic would require exponentially many samples to resolve and cannot witness the advantage captured by characteristic-function queries.

The two off-center bumps at $\beta=\pm\gamma$ have amplitude $\Theta(\epsilon_0)$, and when $\gamma$ is unknown until after data collection they are hard to locate under vacuum-limited smoothing. In displacement coordinates, larger $\gamma$ results in finer oscillations, making features difficult to characterize when smoothed with vacuum noise. We will formalize these intuitions shortly.

Returning to the task of learning the coefficients of $F(t)$, we observe that \eqref{eq:matched_filtering_iso} indicates that the map $(A,B)\rightarrow\alpha$ is a linear isomorphism, so any choice of distribution on $\alpha$ can be pulled back to a distribution on Fourier coefficients $(A, B)$. Then one can realize the distribution $P_\gamma(\alpha)$ in the displacement distribution of the standard sinusoid by sampling each tuple $(A_1, B_1, ..., A_n, B_n)$ from the pullback of $P_\gamma(\alpha)$ under \eqref{eq:matched_filtering_iso}. In this picture, $\gamma$ is a spectral template across the frequency bank: the parameter $\gamma$ plays the role of a hidden template that specifies which collective mode carries the non-Gaussian feature.

\begin{theorem}
    Fix an unknown $\gamma \in \mathbb{C}^n$ with $|\gamma|^2 \leq 1.03n$. Let $F(t)$ be a resonant waveform given by \eqref{eq:matched_filtering_f}, with coefficients $(A_1, B_1, ..., A_n, B_n)$ sampled from the pullback of $P_\gamma(\alpha)$ under \eqref{eq:matched_filtering_iso}. Let $\mathcal{Q} = \{\beta_1, ...,\beta_M\}$ be a template bank satisfying $|\beta_j|^2 \leq  1.03n$ for all $j$, noting this may be fixed by an underlying template bank as in Fact \ref{fact:filtering_bank}. Then, for any $\epsilon \leq 0.24, \delta \leq 1/3$, the following are true.
    \begin{enumerate}
        \item Any entanglement-free protocol (possibly with squeezing) for collecting the dataset $\{\zeta^{(1)}, ... \zeta^{(N)}\}$ that estimates any single template score $\chi_P(\beta_i)$ to within accuracy $\epsilon$ with success probability $2/3$ requires at least
        \begin{equation}
            N \geq \Omega\left(\frac{3^n}{\epsilon^2}\right)
        \end{equation}
        samples.
        \item If the dataset $\{\zeta^{(1)}, ... \zeta^{(N)}\}$ is collected via Bell QSL, the estimates $\chi_{P_\gamma}(\beta_j)$ obtained via Bell QSL satisfy
        \begin{equation}
            \Pr\left[\max_{1\leq j\leq M} \left|\hat{\chi}_{P_\gamma}(\beta_j) - \chi_{P_\gamma} (\beta_j)\right| \leq \epsilon\right]\geq 1-\delta
        \end{equation}
        provided
        \begin{equation}
            N \geq O\left(\frac{\exp\!\left(2.06e^{-2r}n\right)}{\epsilon^2}\log\left(\frac{M}{\delta}\right)\right)
        \end{equation}
    \end{enumerate}
\end{theorem}
\begin{proof}
    We derived that the sample complexity of characteristic-function estimation via Bell QSL in \eqref{eq:charfun_sample_complexity}; replacing $n$ with $1.03 n$ and $\delta$ with $\delta/M$ yields the stated upper bound. The lower bound is a direct application of \cite{jiang2024bellfix}, Theorem 2, with $\kappa = 1.03$.
\end{proof}

When $r \geq \frac{1}{2}\log n$, the Bell QSL approach to matched filtering requires $\sim O(\epsilon^{-2}\log(M/\delta))$. Hence, we see that Bell QSL always enables matched filtering of $n$ overlapping signals with $M$ templates with logarithmic dependence on $M$ and independent of $n$. Meanwhile, scoring even a single template using an entanglement-free data collection requires exponentially many shots in the number of signals. 

This discussion recovers the argument of \cite{jiang2024bellfix} in the special case where Bell QSL is tailored to estimate the characteristic function, reframing their quantum advantage in the sensing setting, and establishes an information-theoretic advantage of Bell QSL over any entanglement-free strategy for worst-case matched filtering.

\vspace{-1em}
\subsection{Theory of practical advantage against restricted learners}
\label{sec:transport_QA}
\vspace{-1em}

In Section~\ref{sec:worst_case_QA}, we established that the Bell-sensing protocol can achieve an information-theoretic quantum advantage against \emph{arbitrary} entanglement-free measurement strategies. While this provides the strongest possible baseline comparison, such worst-case guarantees are typically witnessed by highly structured instances and measurement models that are not representative of the constraints faced in precision sensing experiments.

In this subsection, we therefore turn to \emph{practical} quantum advantage, in the sense of benchmarking against experimentally realistic restricted learners. We focus on a canonical laboratory constraint in which an entanglement-free baseline is limited to homodyne detection at a fixed finite set of quadrature angles, even allowing optimal phase-sensitive squeezing aligned to each available measurement axis. This restriction captures the common situation in which rapid, high-fidelity scanning over many squeezed homodyne angles is technically costly, due to the challenge of repeatedly re-locking a quadrature squeezing axes during the sensing window, so the classical learner effectively observes only finitely many noisy one-dimensional projections of the underlying phase-space displacement distribution. In the precision shot-noise-limited regime, where the vacuum penalty of heterodyne measurements is fatal, this homodyne benchmark is the practically meaningful entanglement-free comparison for Bell QSL.

We formalize this restricted baseline and show how it can be provably ill-conditioned. There exist broad families of displacement distributions that are indistinguishable, or only weakly distinguishable, from finite-angle homodyne data, yet are efficiently learnable from Bell QSL data. The resulting separations are robust in the sense that they persist under small perturbations away from exactly non-identifiable instances, yielding quantitative sample-complexity gaps that depend on a natural ``distance-to-hard-instances'' parameter.

Throughout this subsection, we write $P$ for the displacement distribution induced by the signal (Lemma~\ref{lemma:ham_to_rdc}). Bell QSL yields samples obeying the additive-noise model in Eq.~\eqref{eq:bell-additive-noise}, while heterodyne corresponds to the same observation model with vacuum noise (Fact~\ref{fact:heterodyne_noise}). We refer to Lemma~\ref{lem:bell_identity} and Corollary~\ref{cor:property_sample_complexity} for the corresponding estimation guarantees from Bell data. Our goal here is instead to quantify what can be learned when the entanglement-free baseline is restricted to finite-angle homodyne, and to identify explicit instance families where this restriction is provably ill-conditioned. The framework developed here will then be instantiated in specific, physically-motivated examples where Bell QSL realizes quantum advantage, in Appendix \ref{sec:prac_Gaussian_QA}.

\vspace{-1em}
\subsubsection{Definitions}
\vspace{-1em}

In contrast to Bell sensing, which yields a two-dimensional record per mode, finite-angle homodyne provides only noisy scalar projections of each displacement draw along the accessible quadrature axes. Fix a homodyne angle $\theta$ and write $\hat{x}_\theta = \hat{x}\cos\theta + \hat{p}\sin\theta$. If a displacement $D(\alpha)$ acts on the probe mode, then $\hat{x}_\theta$ is shifted by $\sqrt{2}\,\Re(e^{- i\theta}\alpha)$. If the probe is squeezed along $\hat{x}_\theta$ with variance
\begin{align}
\nu_r = \frac{1}{2} e^{-2 r},
\end{align}
then conditioned on a displacement draw $\alpha \sim P$, the homodyne outcome admits the additive-noise model
\begin{align}
Y_\theta = T_\theta(\alpha) + G, \qquad T_\theta(\alpha) := \sqrt{2}\,\Re\!\left(e^{- i \theta}\alpha\right), \qquad G \sim \mathcal{N}(0,\nu_r).
\label{eq:prac_hom_model}
\end{align}
Equivalently, the one-shot outcome distribution is the projected-and-blurred marginal
\begin{align}
\mu_\theta(P) = (T_\theta)_{\#}P \ast \mathcal{N}(0,\nu_r).
\label{eq:prac_hom_push}
\end{align}
We summarize this restricted baseline as follows.

\begin{definition}[Restricted-angle homodyne experiment]
Let $\Theta = \{\theta_1,\dots,\theta_m\}$ be a fixed finite set of accessible angles. A single \emph{round} of $\Theta$-restricted homodyne data consists of independent samples $Y_{\theta_j} \sim \mu_{\theta_j}(P)$ for each $\theta_j \in \Theta$ (each sample consumes an independent channel use).
The induced one-round outcome distribution is
\begin{align}
\mathcal{M}_\Theta(P) = \bigotimes_{j=1}^{m} \mu_{\theta_j}(P).
\label{eq:prac_M_theta}
\end{align}
Observing $N$ rounds corresponds to i.i.d.~samples from $\mathcal{M}_\Theta(P)^{\otimes N}$.
\end{definition}

\begin{remark}[Fourier restriction under homodyne]
Let $\chi_{P}$ denote the characteristic function of $P$ under the same convention used in Lemma~\ref{lem:bell_identity}. For $t \in \R$, the noiseless projected characteristic function satisfies
\begin{align}
\E\!\left[e^{i t\,T_\theta(\alpha)}\right] = \chi_{P}\!\left(i\,\frac{t}{\sqrt{2}}\,e^{i \theta}\right),
\end{align}
and with additive noise $G \sim \mathcal{N}(0,\nu_r)$ the observed characteristic function is multiplied by
\begin{align}
\exp\!\left(-\frac{\nu_r t^2}{2}\right) = \exp\!\left(-\frac{e^{-2 r}}{4}\,t^2\right).
\end{align}
Thus a finite $\Theta$ reveals $\chi_{P}$ only on finitely many rays, with additional Gaussian radial damping.
\end{remark}

We cast restricted homodyne as an inverse problem on a hypothesis class of displacement distributions. The next definitions package its conditioning into a single quantity that yields minimax lower bounds for learning Lipschitz properties.
\begin{definition}[Measurement-induced pseudometric and OT ambiguity modulus] 
Fix an experiment $\mathcal{M}$ mapping a displacement law $P$ to a one-round outcome distribution $\mathcal{M}(P)$. Let $\mathrm{TV}$ denote total variation distance:
\begin{align}
\mathrm{TV}(Q_1,Q_2) := \sup_{A} |Q_1(A)-Q_2(A)|.
\end{align}
Define the measurement-induced pseudometric
\begin{align}
d_{\mathcal{M}}(P,Q) := \mathrm{TV}\!\left(\mathcal{M}(P),\mathcal{M}(Q)\right).
\end{align}
Given a hypothesis class $\mathcal{C}$ of displacement laws on $\R^2$, define the OT ambiguity modulus
\begin{align}
\omega_{\mathcal{M},\mathcal{C}}(\varepsilon) := \sup\Big\{ W_1(P,Q)\,:\,P,Q \in \mathcal{C},\,d_{\mathcal{M}}(P,Q) \leq \varepsilon\Big\}.
\label{eq:prac_ambiguity}
\end{align}
\end{definition}
\noindent Intuitively, $\omega_{\mathcal{M},\mathcal{C}}(\varepsilon)$ quantifies how much Wasserstein separation can be hidden inside an $\varepsilon$-ball of measurement indistinguishability.

\vspace{-1em}
\subsubsection{Optimal transport and minimax lower bounds}
\vspace{-1em}

To convert such ``indistinguishable-but-far'' pairs into statistical lower bounds, we use a standard two-point reduction based on hypothesis testing.

\begin{lemma}[Two-point Le Cam bound for absolute-error estimation]
\label{lem:prac_lecam}
Let $Q_0$ and $Q_1$ be distributions on a data space and let $\tau(Q_0),\tau(Q_1) \in \R$. For any estimator $\widehat{\tau}$ based on one draw,
\begin{align}
\max\Big\{ \E_{X \sim Q_0}\!\big[|\widehat{\tau}(X)-\tau(Q_0)|\big], \, \E_{X \sim Q_1}\!\big[|\widehat{\tau}(X)-\tau(Q_1)|\big] \Big\} \geq \frac{|\tau(Q_0)-\tau(Q_1)|}{4}\Big(1-\mathrm{TV}(Q_0,Q_1)\Big).
\label{eq:prac_lecam_bound}
\end{align}
The same inequality holds with $Q_0,Q_1$ replaced by product distributions $Q_0^{\otimes N},Q_1^{\otimes N}$.
\end{lemma}

\begin{proof}
Without loss of generality assume $\tau(Q_0)\le \tau(Q_1)$, and set $m := (\tau(Q_0)+\tau(Q_1))/2$. Define the measurable set $A := \{x:\widehat{\tau}(x)\ge m\}$. If $x\in A$ then $\widehat{\tau}(x)-\tau(Q_0)\geq m-\tau(Q_0)=|\tau(Q_1)-\tau(Q_0)|/2$, so $|\widehat{\tau}(x)-\tau(Q_0)|\ge |\tau(Q_1)-\tau(Q_0)|/2$. If $x\notin A$ then $\tau(Q_1)-\widehat{\tau}(x)\ge \tau(Q_1)-m=|\tau(Q_1)-\tau(Q_0)|/2$, so $|\widehat{\tau}(x)-\tau(Q_1)|\ge |\tau(Q_1)-\tau(Q_0)|/2$. Therefore
\begin{align}
\E_{Q_0}\!\big[|\widehat{\tau}-\tau(Q_0)|\big] + \E_{Q_1}\!\big[|\widehat{\tau}-\tau(Q_1)|\big] \geq \frac{|\tau(Q_1)-\tau(Q_0)|}{2}\Big(Q_0(A)+Q_1(A^{c})\Big).
\end{align}
Minimizing $Q_0(A)+Q_1(A^{c})$ over measurable $A$ yields $1-\mathrm{TV}(Q_0,Q_1)$. Finally, $\max\{u,v\}\ge (u+v)/2$ gives \eqref{eq:prac_lecam_bound}. The product case is identical.
\end{proof}

Lemma~\ref{lem:prac_lecam} reduces minimax lower bounds to exhibiting two hypotheses whose induced data distributions have small total variation distance but whose target values differ by a constant amount. In our setting, the ambiguity modulus supplies such pairs in Wasserstein distance, while Kantorovich-Rubinstein duality (Theorem~\ref{thm:KR}) converts a large $W_1(P,Q)$ gap into a $1$-Lipschitz observable whose mean differs by the same order. Combining these ingredients with the product bound $\mathrm{TV}(Q_0^{\otimes N},Q_1^{\otimes N}) \leq N\,\mathrm{TV}(Q_0,Q_1)$ yields the following minimax statement.

\begin{theorem}[OT ambiguity implies minimax lower bounds for restricted experiments]
\label{thm:prac_ot_minimax}
Fix a class $\mathcal{C}$ of displacement laws on $\R^2$ with finite first moment and fix an experiment $\mathcal{M}$.
Let $d_{\mathcal{M}}(P,Q) = \mathrm{TV}(\mathcal{M}(P),\mathcal{M}(Q))$.
For $N \geq 1$, define
\begin{align}
\varepsilon_N := \frac{1}{4 N}, \qquad \Delta_N := \omega_{\mathcal{M},\mathcal{C}}(\varepsilon_N).
\end{align}
Then there exists a $1$-Lipschitz function $f : \R^2 \to \R$ such that any estimator $\widehat{\mu}$ of the mean functional $\mu_f(P) = \E_{X \sim P}[f(X)]$, based on $N$ i.i.d.~samples from $\mathcal{M}(P)$,
satisfies
\begin{align}
\inf_{\widehat{\mu}}\ \sup_{P \in \mathcal{C}} \E\!\left[|\widehat{\mu}-\mu_f(P)|\right] \geq \frac{3}{16}\,\Delta_N.
\label{eq:prac_ot_minimax_bound}
\end{align}
In particular, if $\omega_{\mathcal{M},\mathcal{C}}(0) > 0$, then the minimax risk is bounded below by a positive constant for all $N$.
\end{theorem}

\begin{proof}
Fix $\eta > 0$. By definition of $\Delta_N = \omega_{\mathcal{M},\mathcal{C}}(\varepsilon_N)$ with $\varepsilon_N = \frac{1}{4 N}$, there exist $P,Q \in \mathcal{C}$ such that
\begin{align}
d_{\mathcal{M}}(P,Q) \leq \varepsilon_N, \qquad W_1(P,Q) \geq \Delta_N - \eta.
\end{align}
By Kantorovich-Rubinstein duality (Theorem~\ref{thm:KR}), there exists a $1$-Lipschitz function $f:\R^2\to\R$ such that
\begin{align}
\left|\E_{P}[f]-\E_{Q}[f]\right| \geq W_1(P,Q) - \eta \geq \Delta_N - 2\eta.
\end{align}
Let $Q_P := \mathcal{M}(P)$ and $Q_Q := \mathcal{M}(Q)$. Since $d_{\mathcal{M}}(P,Q) = \mathrm{TV}(Q_P,Q_Q) \leq \frac{1}{4 N}$, we have for products
\begin{align}
\mathrm{TV}\!\left(Q_P^{\otimes N},Q_Q^{\otimes N}\right) \leq N\,\mathrm{TV}(Q_P,Q_Q) \leq \frac{1}{4}.
\end{align}
Apply Lemma~\ref{lem:prac_lecam} to $Q_P^{\otimes N}$ and $Q_Q^{\otimes N}$ with $\tau(Q_P^{\otimes N}) = \E_{P}[f]$ and $\tau(Q_Q^{\otimes N})=\E_{Q}[f]$ to obtain
\begin{align}
\max\Big\{\E_{Q_P^{\otimes N}}\!\big[|\widehat{\mu} - \E_{P}[f]|\big],\,\E_{Q_Q^{\otimes N}}\!\big[|\widehat{\mu} - \E_{Q}[f]|\big]\Big\} \geq \frac{|\E_{P}[f] - \E_{Q}[f]|}{4}\Big(1-\mathrm{TV}\!\left(Q_P^{\otimes N},Q_Q^{\otimes N}\right)\Big) \geq \frac{\Delta_N - 2\eta}{4}\cdot\frac{3}{4} = \frac{3}{16}\,(\Delta_N - 2\eta).
\end{align}
Since $\eta > 0$ is arbitrary, this implies \eqref{eq:prac_ot_minimax_bound}.
\end{proof}

\vspace{-1em}
\subsubsection{Examples of restricted measurement models}
\vspace{-1em}

We begin by exhibiting families of displacement laws that are exactly non-identifiable for the most common restricted baseline, $\Theta = \{0,\pi/2\}$, even though they are separated in Wasserstein distance.

\begin{proposition}[Axis-marginal dependence of $\Theta = \{0,\pi/2\}$ homodyne]
\label{prop:prac_axis_only}
Let $\Theta = \{0,\pi/2\}$. The restricted homodyne experiment $\mathcal{M}_{\Theta}(P)$ depends only on the (noisy) $\Re(\alpha)$-marginal and $\Im(\alpha)$-marginal of $P$ as measures on $\R$. In particular, if $P$ and $Q$ have identical $\Re(\alpha)$- and $\Im(\alpha)$-marginals, then
\begin{align}
\mathcal{M}_{\Theta}(P) = \mathcal{M}_{\Theta}(Q),
\end{align}
even after convolution with the same known Gaussian measurement noise in \eqref{eq:prac_hom_push}.
\end{proposition}

\begin{proof}
For $\theta = 0$, \eqref{eq:prac_hom_push} shows $\mu_{0}(P)$ depends only on the pushforward of $P$ under $\alpha \mapsto \sqrt{2}\,\Re(\alpha)$, convolved with $\mathcal{N}(0,\nu_r)$. For $\theta = \pi/2$ it depends only on $\alpha \mapsto \sqrt{2}\,\Im(\alpha)$ convolved with the same noise. The one-round distribution is their product.
\end{proof}

Proposition~\ref{prop:prac_axis_only} shows that $\Theta = \{0,\pi/2\}$ homodyne is blind to any change in the \emph{dependence structure} between $\Re(\alpha)$ and $\Im(\alpha)$ that preserves the two axis marginals. The next two propositions give concrete instances of this phenomenon.

\begin{proposition}[A discrete XOR pair]
\label{prop:prac_xor_pair}
Fix $a > 0$ and define distributions on $\R^2$,
\begin{align}
P = \frac{1}{2}\,\delta_{(a,a)} + \frac{1}{2}\,\delta_{(-a,-a)}, \qquad Q = \frac{1}{2}\,\delta_{(a,-a)} + \frac{1}{2}\,\delta_{(-a,a)}.
\end{align}
Then $P$ and $Q$ have identical $\Re(\alpha)$- and $\Im(\alpha)$-marginals and are indistinguishable by $\Theta = \{0,\pi/2\}$ homodyne. Moreover,
\begin{align}
W_1(P,Q) = 2 a, \qquad W_2(P,Q) = 2 a.
\end{align}
\end{proposition}

\begin{proof}
Both $P$ and $Q$ place $\Re(\alpha)$-mass equally on $\{+a,-a\}$ and similarly for $\Im(\alpha)$. Any coupling must move $(a,a)$ to either $(a,-a)$ or $(-a,a)$, each at Euclidean distance $2 a$. The coupling $(a,a) \mapsto (a,-a)$ and $(-a,-a) \mapsto (-a,a)$ achieves average transport cost $2 a$, proving $W_1 = 2 a$, and since all transported mass moves exactly distance $2 a$, also $W_2 = 2 a$.
\end{proof}

The discrete XOR pair changes only how the two coordinates are \emph{paired}. A smooth analogue is obtained by flipping the sign of an off-diagonal covariance, wherein the axis marginals (diagonal entries) remain fixed but the correlation structure changes. We expand on this smooth analogue in Appendix \ref{sec:ex_sum_of_deltas}.

\begin{proposition}[A Gaussian correlation-hiding pair]
\label{prop:prac_gauss_hide}
Fix $c \in (0,1)$ and define
\begin{align}
\Sigma_{\pm} = \begin{pmatrix} \sigma^2 & \pm c \\ \pm c & \sigma^2 \end{pmatrix}, \qquad P = \mathcal{N}(0,\Sigma_{+}), \qquad Q = \mathcal{N}(0,\Sigma_{-}).
\end{align}
Then $P$ and $Q$ have identical axis marginals and are indistinguishable by $\Theta = \{0,\pi/2\}$ homodyne. Their Wasserstein-2 distance satisfies
\begin{align}
W_2^2(P,Q) = 4\Big(\sigma^2-\sqrt{\sigma^4-c^2}\Big),
\label{eq:prac_W2_gauss_hide}
\end{align}
and hence $W_2(P,Q) \asymp \sqrt{2}\,|c|$ when $|c| \ll 1$.
\end{proposition}

\begin{proof}
The axis marginals of a centered Gaussian are determined by the diagonal entries of its covariance, which coincide. For $W_2$ between centered Gaussians, use the closed form \eqref{eq:W2_gaussian}. A direct computation gives $\Sigma_{+}\Sigma_{-} = (\sigma^4-c^2) \mathds{1}_2$, hence $\big(\Sigma_{+}^{1/2}\Sigma_{-}\Sigma_{+}^{1/2}\big)^{1/2} = \sqrt{\sigma^4-c^2}\,\mathds{1}_2$ and \eqref{eq:prac_W2_gauss_hide} follows.
\end{proof}

\noindent These exactly non-identifiable pairs certify that the OT ambiguity modulus can be strictly positive at $\varepsilon = 0$ for broad hypothesis classes, and therefore force a non-vanishing minimax risk for estimating some Lipschitz functional from $\Theta$-restricted homodyne data.

\begin{corollary}[Non-identifiability implies $\omega_{\mathcal{M}_\Theta,\mathcal{C}}(0) > 0$ for broad classes]
\label{cor:prac_nonid_omega}
Let $\mathcal{C}$ contain the Gaussian family $\{\mathcal{N}(0,\Sigma_{\pm}) : c \in [c_0,1)\}$ for some $c_0 > 0$. For restricted homodyne with $\Theta = \{0,\pi/2\}$, one has $\omega_{\mathcal{M}_\Theta,\mathcal{C}}(0) > 0$. Consequently, there exists a $1$-Lipschitz observable $f$ whose mean $\E_{P}[f]$ cannot be consistently estimated from $\Theta$-restricted homodyne data, regardless of the number of samples.
\end{corollary}

\begin{proof}
Take $P = \mathcal{N}(0,\Sigma_{+})$ and $Q = \mathcal{N}(0,\Sigma_{-})$ for any fixed $c \in [c_0,1)$. Then $\mathcal{M}_\Theta(P) = \mathcal{M}_\Theta(Q)$ so $d_{\mathcal{M}_\Theta}(P,Q)=0$, but $W_2(P,Q) > 0$ and hence $W_1(P,Q) > 0$ under finite-moment control. Therefore $\omega_{\mathcal{M}_\Theta,\mathcal{C}}(0) \geq W_1(P,Q) > 0$. Theorem~\ref{thm:prac_ot_minimax} implies the existence of a $1$-Lipschitz $f$ with non-vanishing minimax risk.

Moreover, since $\mathcal{M}_\Theta(P)=\mathcal{M}_\Theta(Q)$, we have $\mathcal{M}_\Theta(P)^{\otimes N}=\mathcal{M}_\Theta(Q)^{\otimes N}$ for every $N$, hence the total variation distance is identically zero at all sample sizes. By Kantorovich-Rubinstein duality, for any $\eta>0$ there exists a $1$-Lipschitz $f$ with $|\E_P[f]-\E_Q[f]|\ge W_1(P,Q)-\eta$; taking $\eta=W_1(P,Q)/2$ yields a fixed $f$ with a positive mean gap. Applying Lemma~\ref{lem:prac_lecam} with $\mathrm{TV}=0$ then gives a constant minimax lower bound uniformly in $N$, so $\mu_f(P)$ cannot be consistently estimated from
$\Theta$-restricted homodyne data.
\end{proof}

In contrast, Bell sensing yields a two-dimensional record (up to additive Gaussian noise), so it can directly estimate such hidden correlations.

\begin{proposition}[Bell sampling learns the hidden correlation efficiently]
\label{prop:prac_bell_cov}
Consider the Gaussian family in Proposition~\ref{prop:prac_gauss_hide}. Under Bell sampling with squeezing $r$, the outcome $\zeta$ satisfies Eq.~\eqref{eq:bell-additive-noise} under $\C\simeq\R^2$, so one has
\begin{align}
(\Re(\zeta),\Im(\zeta)) \sim \mathcal{N}\!\left(0,\Sigma_{\pm} + \nu_r \mathds{1}_2\right), \qquad \nu_r = \frac{1}{2} e^{-2 r}.
\end{align}
Let $\widehat{c} := \frac{1}{N}\sum_{i=1}^{N}\Re(\zeta^{(i)})\,\Im(\zeta^{(i)})$. Then for centered Gaussians
\begin{align}
\Var\!\left(\Re(\zeta)\,\Im(\zeta)\right) = (\sigma^2+\nu_r)(\sigma^2+\nu_r) + c^2,
\end{align}
and achieving mean-squared error $\E[(\widehat{c}-c)^2] \leq \epsilon^2$ requires
\begin{align}
N = O\!\left(\frac{(\sigma^2+\nu_r)^2 + c^2}{\epsilon^2}\right).
\label{eq:prac_bell_cov_sc}
\end{align}
This scales as $O\left(e^{-4r}/{\epsilon^2}\right)$ in the shot-noise-limited regime.
\end{proposition}

\begin{proof}
For centered jointly Gaussian $(X,Y)$, one has $\Var(XY) = \Var(X)\Var(Y) + \text{Cov}(X,Y)^2$, giving the stated variance. Chebyshev or standard sub-exponential concentration for products of Gaussians yields the $1/\epsilon^2$ scaling. In the shot-noise-limited regime where $\sigma, c \ll 1$ and comparable to the squeezing, the $e^{-4r}$ dependence follows.
\end{proof}

The preceding constructions are explicit, but exact non-identifiability is not special to $\Theta = \{0,\pi/2\}$. In fact, for any finite angle set, there are smooth ``ghost'' directions that vanish under all the corresponding projections. One convenient way to exhibit this kernel is via Fourier support.

\begin{lemma}[Finite-angle invisibility via Fourier support]
\label{lem:prac_fourier_invisible}
Fix a finite angle set $\Theta \subset [0,2\pi)$. There exists a nonzero smooth, compactly supported signed density $g:\R^2 \to \R$ with
\begin{align}
\int_{\R^2} g(x)\,dx = 0
\end{align}
such that for every $\theta \in \Theta$, one has
\begin{align}
(\pi_{u_\theta})_{\#} g \equiv 0
\end{align}
as a signed measure on $\R$, where $u_\theta = (\cos\theta,\sin\theta)$ and $\pi_{u_\theta}(x)=u_\theta^{\top}x$. Consequently, for any strictly positive base density $p_0$ and sufficiently small $\varepsilon > 0$, the two densities
\begin{align}
p_{+} = p_0 + \varepsilon g, \qquad p_{-} = p_0 - \varepsilon g
\end{align}
are distinct valid probability densities and satisfy $\mathcal{M}_\Theta(p_{+}) = \mathcal{M}_\Theta(p_{-})$ for noiseless homodyne projections along $\Theta$, and also after convolution with any fixed known Gaussian noise.
\end{lemma}

\begin{proof}
Fix any nonzero $\varphi\in C_c^\infty(\R^2)$. For each $\theta\in\Theta$, let $u_\theta=(\cos\theta,\sin\theta)$ and $u_\theta^\perp:=(-\sin\theta,\cos\theta)$. Define the constant-coefficient differential operator
\begin{align}
L := \prod_{\theta\in\Theta}\big(u_\theta^\perp\cdot\nabla\big)^2,
\end{align}
and set $g := L\varphi$. Then $g\in C_c^\infty(\R^2)$, and since $L$ annihilates constants we have $\int_{\R^2} g(x)\,dx = 0$. Moreover, for a generic choice of $\varphi$ one has $g\not\equiv 0$ (and if a particular $\varphi$ happens to lie in $\ker L$, choose another nonzero bump).

Let $\widehat{g}(k)=\int_{\R^2} g(x)e^{-ik\cdot x}\,dx$ be the Fourier transform. Because differentiation corresponds to multiplication in Fourier space,
\begin{align}
\widehat{g}(k) = \Big(\prod_{\theta\in\Theta} (u_\theta^\perp\cdot k)^2\Big)\,\widehat{\varphi}(k).
\end{align}
Hence for every $\theta\in\Theta$ and every $t\in\R$,
\begin{align}
\widehat{g}(t u_\theta)=0,
\end{align}
since $u_\theta^\perp\cdot u_\theta=0$. For each $\theta$, the characteristic function of the one-dimensional pushforward $(\pi_{u_\theta})_{\#}g$ is $t \mapsto \widehat{g}(t u_\theta)$, so it vanishes identically and therefore $(\pi_{u_\theta})_{\#} g \equiv 0$ as a signed measure on $\R$.

Now take any strictly positive base density $p_0$. Since $g$ is compactly supported and bounded, choosing $\varepsilon>0$ sufficiently small ensures $p_\pm=p_0\pm \varepsilon g\ge 0$ pointwise, and $\int p_\pm = 1$ because $\int g=0$. Equality of noiseless projections along $\Theta$ implies $\mathcal{M}_\Theta(p_+)=\mathcal{M}_\Theta(p_-)$, and convolving each projected marginal with the same known Gaussian noise preserves equality.
\end{proof}

Exact invisibility is a knife-edge phenomenon. To capture \emph{conditioning} rather than identifiability, we perturb the exactly invisible instances so that $\Theta$-restricted homodyne becomes informative only through an $\varepsilon$-scale leakage, which forces a $1/\varepsilon^2$ sample complexity blowup. By relaxing around exactly invisible instances, we are able to quantify the wedge of the problem space associated a family of distributions which admits a gap between Bell and homodyne learners. In Appendix \ref{sec:prac_Gaussian_QA}, we identify the problem wedges for the Gaussian and discrete XOR cases, and associate them with problems related to learning classical electromagnetic fields and noise sources. The problem wedge concept is depicted in Figure \ref{fig:deltas_prob_wedge}.

\indent Finally, we give an OT perspective on why such hard (but not impossible) instances are easy to find in practice. As reviewed in Section~\ref{sec:OT_overview}, both projections and common-noise convolution are Wasserstein contractions (Lemma~\ref{lem:ot_contractions}). The restricted experiment $\mathcal{M}_{\Theta}$ therefore controls at best a finite-angle projected-and-blurred geometry, whereas Bell sensing provides a two-dimensional record that constrains the full joint law up to the Gaussian blur in Eq.~\eqref{eq:bell-additive-noise}. Geometrically, each projection $\pi_{u_\theta}$ collapses mass along affine fibers $\{x : u_\theta^{\top}x=t\}$, so rearrangements of mass within these fibers leave the corresponding marginal unchanged. With finitely many angles, there remains substantial freedom to modify the dependence structure of $(\Re(\alpha),\Im(\alpha))$ while leaving all observed marginals nearly invariant. This is exactly the mechanism exploited by the invisible-correlation constructions and their $\varepsilon$-deformations in the following sections: the per-sample information available to $\Theta$-restricted homodyne scales quadratically in the leakage parameter (e.g.~$D_{\mathrm{KL}} \asymp \varepsilon^2$), forcing a $1/\varepsilon^2$ sample-complexity blowup, while Bell sensing remains governed by the intrinsic two-dimensional separation. This behavior is summarized abstractly by the OT ambiguity modulus \eqref{eq:prac_ambiguity} and converted into minimax lower bounds by Theorem~\ref{thm:prac_ot_minimax}.

\section{Examples of practical quantum advantage}
\label{sec:prac_Gaussian_QA}

Having built a comprehensive theory of practical quantum advantage in Section \ref{sec:transport_QA}, we turn to simple, yet physically crucial sensing tasks which fall under the QSL framework, and for which Bell QSL provides meaningful gains over current conventional strategies.

\vspace{-1em}
\subsection{A classical mixture of peaked displacements}
\label{sec:ex_sum_of_deltas}
\vspace{-1em}

To build intuition for how practical quantum advantages arise from relaxing exact marginal blindness, we first return to the illustrative example from Proposition \ref{prop:prac_xor_pair}. 

\begin{figure}
    \centering
    \includegraphics[width=0.99\linewidth]{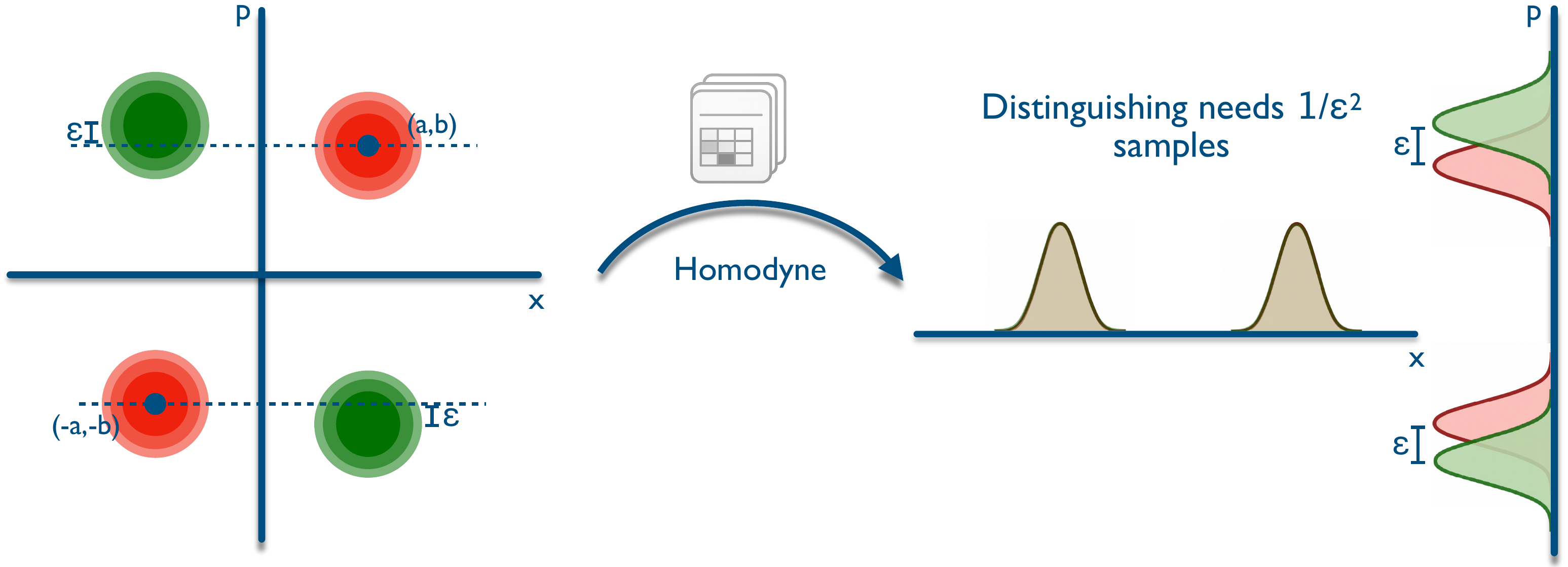}
    \caption{Measurement distributions corresponding to peaked displacement channels in \eqref{eq:prac_eps_discrete_hyps}. $P_{\mathrm{diag}}$ is depicted in red and $Q_\varepsilon$ in green. The homodyne marginals upon measuring along axes $\{0, \pi/2\}$ are shown to the right. Their only distinguishable feature is a slight displacement in means along the $p$ axis.} 
    \label{fig:deltas_pic}
\end{figure}

Fix $a,b > 0$ and $0 < \varepsilon \leq b/2$ and define
\begin{align}
P_{\mathrm{diag}} = \frac{\delta_{(a,b)}+\delta_{(-a,-b)}}{2}, \ \ 
Q_{\varepsilon} = \frac{\delta_{(a,-b-\varepsilon)}+\delta_{(-a, b+\varepsilon)}}{2}\,.
\label{eq:prac_eps_discrete_hyps}
\end{align}
Write $\nu_r = \frac{1}{2} e^{-2 r}$ for the homodyne noise variance in \eqref{eq:prac_hom_model}. Bell sensing remains robust to this perturbation, in the sense that the quadrant-sign statistic still separates the hypotheses at fixed SNR.

\begin{proposition}[Bell sampling distinguishes \eqref{eq:prac_eps_discrete_hyps} with $O(1)$ samples]
\label{prop:prac_eps_discrete_bell}
Consider the test
\begin{align}
\phi(\zeta) = \indicator\!\left\{\Re(\zeta)\,\Im(\zeta) \geq 0\right\},
\end{align}
which decides $P_{\mathrm{diag}}$ when $\phi(\zeta)=1$ and decides $Q_{\varepsilon}$ when $\phi(\zeta)=0$. Then under Bell sampling with squeezing $r$,
\begin{align}
\max\Big\{\Pr_{P_{\mathrm{diag}}}\!\big[\phi(\zeta)=0\big],\,\Pr_{Q_{\varepsilon}}\!\big[\phi(\zeta)=1\big]\Big\} \leq 2\exp\!\left(-\frac{\min\{a, b\}^2}{2 \nu_r}\right).
\label{eq:prac_eps_discrete_bell_err}
\end{align}
\end{proposition}

\begin{proof}
Under $P_{\mathrm{diag}}$, $(\Re(\alpha),\Im(\alpha))$ lies in quadrants I or III, so $\Re(\alpha)$ and $\Im(\alpha)$ have the same sign. An error occurs only if exactly one of $\Re(\zeta)$ or $\Im(\zeta)$ flips sign due to additive noise $\mathcal{N}(0,\nu_r)$ in each coordinate. A union bound gives the stated estimate. The same argument applies under $Q_{\varepsilon}$, where the points lie in quadrants II or IV with magnitudes at least $\min\{a, b\}$.
\end{proof}

For $\Theta = \{0,\pi/2\}$ homodyne, the $\theta=0$ outcomes are identically distributed under $P_{\mathrm{diag}}$ and $Q_{\varepsilon}$, so testing reduces to the informative $\theta=\pi/2$ outcomes. The resulting one-shot $\theta=\pi/2$ laws are the symmetric Gaussian mixtures
\begin{align}
P_0 &= \frac{1}{2}\,\mathcal{N}\!\left(\sqrt{2}\,b,\nu_r\right) + \frac{1}{2}\,\mathcal{N}\!\left(-\sqrt{2}\,b,\nu_r\right), \\
P_{\varepsilon} &= \frac{1}{2}\,\mathcal{N}\!\left(\sqrt{2}\,(b+\varepsilon),\nu_r\right) + \frac{1}{2}\,\mathcal{N}\!\left(-\sqrt{2}\,(b+\varepsilon),\nu_r\right).
\label{eq:prac_eps_discrete_xmarg}
\end{align}

The next lemma bounds the \emph{per-sample} information available to this baseline by controlling the one-shot KL divergence.

\begin{lemma}[A KL upper bound for \eqref{eq:prac_eps_discrete_xmarg}]
\label{lem:prac_eps_discrete_kl}
One has
\begin{align}
D_{\mathrm{KL}}\!\left(P_0 \,\|\, P_{\varepsilon}\right) \leq \frac{\varepsilon^2}{\nu_r}.
\label{eq:prac_eps_discrete_kl}
\end{align}
\end{lemma}

\begin{proof}
Introduce a latent sign $S \in \{\pm 1\}$ uniform. Under $P_0$, let $Y = S \sqrt{2}\,b + G$; under $P_{\varepsilon}$, let $Y = S \sqrt{2}\,(b+\varepsilon) + G$, where $G \sim \mathcal{N}(0,\nu_r)$ is independent. Conditioned on $S$, the two laws are Gaussians with identical variance and means differing by $\pm\sqrt{2}\,\varepsilon$, hence
\begin{align}
D_{\mathrm{KL}}\!\left(\mathcal{N}\!\left(S\sqrt{2}\,b,\nu_r\right)\,\Big\|\,\mathcal{N}\!\left(S\sqrt{2}\,(b+\varepsilon),\nu_r\right)\right) = \frac{(\sqrt{2}\,\varepsilon)^2}{2\nu_r} = \frac{\varepsilon^2}{\nu_r}.
\end{align}
Therefore $D_{\mathrm{KL}}\!\left((S,Y)_0 \,\|\, (S,Y)_{\varepsilon}\right) = \varepsilon^2/\nu_r$. Marginalizing out $S$ is a stochastic map, so KL cannot increase, giving \eqref{eq:prac_eps_discrete_kl}.
\end{proof}

For these explicit two-hypothesis deformations, we use the standard total-variation characterization of testing risk together with Pinsker's inequality. This is complementary to Lemma~\ref{lem:prac_lecam}; above we used Le Cam in an absolute-error form to obtain a general minimax statement for some Lipschitz functional, whereas here we specialize to concrete binary tests and bound their distinguishability directly via $\mathrm{TV}$.

\begin{lemma}[Testing bounds from total variation and Pinsker]
\label{lem:prac_tv_pinsker}
For any test $\psi$ between distributions $Q_0,Q_1$,
\begin{align}
\max\Big\{ Q_0[\psi=1],\,Q_1[\psi=0]\Big\} \geq \frac{1}{2}\Big(1-\mathrm{TV}(Q_0,Q_1)\Big).
\label{eq:prac_testing_tv}
\end{align}
Moreover,
\begin{align}
\mathrm{TV}(Q_0,Q_1)^2 \leq \frac{1}{2}\,D_{\mathrm{KL}}(Q_0\,\|\,Q_1).
\label{eq:prac_pinsker}
\end{align}
\end{lemma}

Combining Lemma~\ref{lem:prac_eps_discrete_kl} with Lemma~\ref{lem:prac_tv_pinsker} yields the following sample-complexity lower bound for the restricted baseline.

\begin{theorem}[$\Theta=\{0,\pi/2\}$ homodyne needs $\Omega(\nu_r/\varepsilon^2)$ samples for \eqref{eq:prac_eps_discrete_hyps}]
\label{thm:prac_eps_discrete_lower}
Let $N \geq 1$ and consider testing $P_{\mathrm{diag}}$ versus $Q_{\varepsilon}$ using $\Theta=\{0,\pi/2\}$ homodyne rounds. If
\begin{align}
N \leq \frac{\nu_r}{8 \varepsilon^2},
\end{align}
then every $\Theta$-restricted homodyne test has error probability at least $3/8$.
\end{theorem}

\begin{proof}
The $\theta=\pi/2$ outcomes are identical under both hypotheses, so testing reduces to distinguishing $P_0^{\otimes N}$ from $P_{\varepsilon}^{\otimes N}$ where $P_0,P_{\varepsilon}$ are \eqref{eq:prac_eps_discrete_xmarg}. By \eqref{eq:prac_pinsker} and additivity of KL under products,
\begin{align}
\mathrm{TV}\!\left(P_0^{\otimes N},P_{\varepsilon}^{\otimes N}\right)^2 &\leq \frac{1}{2}\,D_{\mathrm{KL}}\!\left(P_0^{\otimes N}\,\|\,P_{\varepsilon}^{\otimes N}\right) = \frac{N}{2}D_{\mathrm{KL}}\!\left(P_0\,\|\,P_{\varepsilon}\right) \leq \frac{N}{2}\cdot\frac{\varepsilon^2}{\nu_r},
\end{align}
using Lemma~\ref{lem:prac_eps_discrete_kl}. If $N \leq \nu_r/(8\varepsilon^2)$ then $\mathrm{TV}\leq 1/4$. Lemma~\ref{lem:prac_tv_pinsker} then gives minimax error at least $(1-1/4)/2 = 3/8$.
\end{proof}

\noindent The lower bound is essentially tight in the high-SNR regime.  In particular, after resolving the latent sign of the mixture component, the problem reduces to distinguishing two nearby Gaussian means.

\begin{proposition}[A matching $\Theta=\{0,\pi/2\}$ homodyne upper bound at high SNR]
\label{prop:prac_eps_discrete_upper}
Fix $\delta \in (0,1/2)$ and suppose
\begin{align}
N \exp\!\left(-\frac{b^2}{\nu_r}\right) \leq \frac{\delta}{2}.
\label{eq:prac_eps_discrete_snr}
\end{align}
Then there exists an explicit $\Theta=\{0,\pi/2\}$ homodyne test using only $\theta=\pi/2$ outcomes with the following guarantee: if
\begin{align}
N \geq \frac{4 \nu_r}{\varepsilon^2}\log\!\left(\frac{4}{\delta}\right),
\label{eq:prac_eps_discrete_upperN}
\end{align}
then the error probability is at most $\delta$.
\end{proposition}

\begin{proof}
Draw $N$ independent $\theta=\pi/2$ outcomes $Y_1,\dots,Y_N$. Under either hypothesis we can represent
\begin{align}
Y_i = S_i \sqrt{2}\,B + G_i,
\end{align}
where $S_i \in \{\pm 1\}$ are i.i.d.~uniform signs, $G_i \sim \mathcal{N}(0,\nu_r)$ are i.i.d.~and $B=b$ under $P_{\mathrm{diag}}$ while $B=b+\varepsilon$ under $Q_{\varepsilon}$. Define $\widehat{S}_i := \mathrm{sign}(Y_i)$ and the ideal sign-corrected variables $\widetilde{W}_i := S_i Y_i$. Then unconditionally,
\begin{align}
\widetilde{W}_i \sim \mathcal{N}(\sqrt{2}\,b,\nu_r) \quad \text{under } P_{\mathrm{diag}}, \qquad \widetilde{W}_i \sim \mathcal{N}(\sqrt{2}\,(b+\varepsilon),\nu_r) \quad \text{under } Q_{\varepsilon}.
\end{align}
Let $E$ be the event that no sign errors occur, i.e. $E := \{\widehat{S}_i = S_i \text{ for all } i\}$. If $E$ occurs then $\widehat{S}_i Y_i = S_i Y_i$ for all $i$, hence $W_i := \widehat{S}_i Y_i$ satisfies $W_i = \widetilde{W}_i$ for all $i$.

We first bound $\Pr(E^c)$. For each $i$, a sign error occurs only if $G_i \geq \sqrt{2}\,b$ or $G_i \leq -\sqrt{2}\,b$, hence
\begin{align}
\Pr(\widehat{S}_i \neq S_i) \leq 2 \Pr\!\left(G_i \geq \sqrt{2}\,b\right) \leq 2 \exp\!\left(-\frac{b^2}{\nu_r}\right),
\end{align}
and by a union bound,
\begin{align}
\Pr(E^c) \leq 2 N \exp\!\left(-\frac{b^2}{\nu_r}\right).
\end{align}
Under the high-SNR condition \eqref{eq:prac_eps_discrete_snr}, this gives $\Pr(E^c) \leq \delta$ up to an absolute-factor slack (absorbed by the constants in \eqref{eq:prac_eps_discrete_snr}).

Now consider the threshold test based on the empirical mean of $W_i$:
\begin{align}
\text{decide } Q_{\varepsilon} \,\, \text{if and only if} \,\, \frac{1}{N}\sum_{i=1}^{N} W_i \geq \sqrt{2}\Big(b+\frac{\varepsilon}{2}\Big).
\end{align}
Since $W_i=\widetilde{W}_i$ on $E$, the conditional error on $E$ is upper bounded by the error of the same threshold test applied to the Gaussian sample mean of i.i.d.~$\widetilde{W}_i$. A standard Gaussian tail bound yields
\begin{align}
\Pr\!\left(\left.\textnormal{error}\ \right|\ E\right) \leq 2 \exp\!\left(-\frac{N \varepsilon^2}{4 \nu_r}\right).
\end{align}
Therefore
\begin{align}
\Pr(\textnormal{error}) \leq \Pr(E^c) + \Pr\!\left(\left.\textnormal{error}\ \right|\ E\right) \leq 2 N \exp\!\left(-\frac{b^2}{\nu_r}\right) + 2 \exp\!\left(-\frac{N \varepsilon^2}{4 \nu_r}\right).
\end{align}
Combining \eqref{eq:prac_eps_discrete_snr} with the sample size condition \eqref{eq:prac_eps_discrete_upperN} implies the right-hand side is at most $\delta$ (up to absolute-factor slack absorbed by the constants), completing the proof.
\end{proof}

\noindent An analogous deformation exists for the Gaussian correlation-hiding family, since we perturb what $\Theta=\{0,\pi/2\}$ homodyne can see (a diagonal variance) while keeping the correlation sign fixed.

\vspace{-1em}
\subsection{Learning Gaussian correlations}
\label{sec:ex_gaussian_covar}
\vspace{-1em}

Fix $c \lesssim e^{-2r}$ and small $\varepsilon > 0$ and define
\begin{align}
\Sigma_{+} = \begin{pmatrix} \sigma^2 & c \\ c & \sigma^2 \end{pmatrix}, \qquad \Sigma_{-,\varepsilon} = \begin{pmatrix} \sigma^2+\varepsilon & -c \\ -c & \sigma^2 \end{pmatrix},
\end{align}
with hypotheses $P_{+} = \mathcal{N}(0,\Sigma_{+})$ and $Q_{\varepsilon} = \mathcal{N}(0,\Sigma_{-,\varepsilon})$, and $\sigma^2 = O(e^{-2r})$. This is precisely the shot-noise-limited regime in the Gaussian setting discussed earlier in Example \ref{ex:gaussian_covar}, in which quantum-enhanced precision measurements are capable of providing advantage over unentangled ones. Under $\theta=0$ homodyne, $Y_0 = \sqrt{2}\,X + G$ with $X \sim \mathcal{N}(0,\Sigma_{11})$ and $G \sim \mathcal{N}(0,\nu_r)$ independent, hence the informative one-shot laws are
\begin{align}
P_0 = \mathcal{N}(0,v_0), \qquad v_0 = 2\sigma^2 + \nu_r, \qquad P_{\varepsilon} = \mathcal{N}(0,v_1), \qquad v_1 = 2(\sigma^2+\varepsilon) + \nu_r = v_0 + 2\varepsilon.
\label{eq:prac_eps_gauss_xmarg}
\end{align}
The $\theta=\pi/2$ homodyne outcomes agree under $P_{+}$ and $Q_{\varepsilon}$. The next lemma makes explicit that the one-shot information scales as $\varepsilon^2$.

\begin{lemma}[KL for the variance-shifted marginals \eqref{eq:prac_eps_gauss_xmarg}]
\label{lem:prac_eps_gauss_kl}
One has
\begin{align}
D_{\mathrm{KL}}\!\left(\mathcal{N}(0,v_0)\,\|\,\mathcal{N}(0,v_1)\right) = \frac{1}{2}\left(\frac{v_0}{v_1}-1-\log\!\left(\frac{v_0}{v_1}\right)\right).
\label{eq:prac_eps_gauss_kl_exact}
\end{align}
Moreover, if $0 < 2\varepsilon \leq v_0$, then
\begin{align}
\frac{\varepsilon^2}{(v_0+2\varepsilon)^2} \leq D_{\mathrm{KL}}\!\left(\mathcal{N}(0,v_0)\,\|\,\mathcal{N}(0,v_1)\right) \leq \frac{2\varepsilon^2}{v_0^2}.
\label{eq:prac_eps_gauss_kl_bounds}
\end{align}
\end{lemma}

\begin{proof}
The identity \eqref{eq:prac_eps_gauss_kl_exact} is standard for Gaussians. Write $v_1=v_0(1+x)$ with $x=2\varepsilon/v_0\in(0,1]$, so
\begin{align}
D_{\mathrm{KL}} =\frac{1}{2}\left(\frac{1}{1+x}-1+\log(1+x)\right) =\frac{1}{2}\left(\log(1+x)-\frac{x}{1+x}\right).
\end{align}
For the upper bound, use $\log(1+x)\le x$ and $\frac{x}{1+x}\ge x-x^2$ (valid for $x\in(0,1]$), giving $D_{\mathrm{KL}}\le \frac{1}{2}(x-(x-x^2))=x^2/2 = 2\varepsilon^2/v_0^2$.

For the lower bound, define
\begin{align}
h(x):=\log(1+x)-\frac{x}{1+x}-\frac{x^2}{2(1+x)^2}.
\end{align}
Then $h(0)=0$ and
\begin{align}
h'(x)=\frac{1}{1+x}-\frac{1}{(1+x)^2}-\frac{x}{(1+x)^3}=\frac{x^2}{(1+x)^3}\geq 0,
\end{align}
so $h(x)\geq 0$ for $x\geq 0$. Thus
\begin{align}
D_{\mathrm{KL}}
\geq \frac{1}{2}\cdot \frac{x^2}{2(1+x)^2} = \frac{x^2}{4(1+x)^2} = \frac{\varepsilon^2}{v_1^2} = \frac{\varepsilon^2}{(v_0+2\varepsilon)^2}.
\end{align}
\end{proof}

Plugging Lemma~\ref{lem:prac_eps_gauss_kl} into Pinsker's inequality and the testing bound from Lemma~\ref{lem:prac_tv_pinsker} yields the corresponding sample-complexity lower bound, with a matching upper bound given by an empirical second-moment test.

\begin{theorem}[$\Theta=\{0,\pi/2\}$ homodyne needs $\Omega(v_0^2/\varepsilon^2)$ samples for the deformed Gaussian pair]
\label{thm:prac_eps_gauss_lower}
Consider testing $P_{+}$ versus $Q_{\varepsilon}$ using $\Theta=\{0,\pi/2\}$ homodyne rounds. If
\begin{align}
N \leq \frac{v_0^2}{64 \varepsilon^2},
\end{align}
then every $\Theta$-restricted homodyne test has error probability at least $3/8$. Conversely, a $\theta=0$ variance-threshold test achieves error at most $\delta$ using
\begin{align}
N = O\!\left(\frac{v_0^2}{\varepsilon^2}\log\!\left(\frac{1}{\delta}\right)\right).
\end{align}
\end{theorem}

\begin{proof}
The $\theta=\pi/2$ outcomes are identical, so testing reduces to distinguishing $\mathcal{N}(0,v_0)^{\otimes N}$ from $\mathcal{N}(0,v_1)^{\otimes N}$. By Pinsker,
\begin{align}
\mathrm{TV}^2 \leq \frac{1}{2}D_{\mathrm{KL}}^{(N)} = \frac{N}{2}D_{\mathrm{KL}}\!\left(\mathcal{N}(0,v_0)\,\|\,\mathcal{N}(0,v_1)\right) \leq \frac{N}{2}\cdot\frac{2\varepsilon^2}{v_0^2},
\end{align}
using \eqref{eq:prac_eps_gauss_kl_bounds}. If $N \leq v_0^2/(64\varepsilon^2)$ then $\mathrm{TV}\leq 1/4$ and Lemma~\ref{lem:prac_tv_pinsker} yields error at least $3/8$. For the upper bound, take $T = \frac{1}{N}\sum_{i=1}^{N}Y_i^2$ for $\theta=0$ outcomes. Under $\mathcal{N}(0,v_j)$, we have $\E[T]=v_j$ and $\Var(T)=2v_j^2/N$. A threshold at $(v_0+v_1)/2$ and standard concentration for $Y_i^2$ gives the stated scaling.
\end{proof}

Bell sensing remains sensitive to the correlation sign $\pm c$ through the joint samples, so its sample complexity is governed by $c$ rather than by the leakage scale $\varepsilon$. The following Corollary is obtained directly from Proposition \ref{prop:prac_bell_cov}, setting the precision level to $c/2$ to solve the distinguishing problem.

\begin{corollary}[Bell sampling distinguishes the deformed Gaussian pair with $O(1/c^2)$ samples]
\label{cor:prac_eps_gauss_bell}
Fix $\delta \in (0,1/2)$ and suppose $\varepsilon \leq 1$. In the shot-noise-limited regime $\sigma^2 \lesssim e^{-2r} \ll 1$, there is a Bell-sampling test with error probability at most $\delta$ using
\begin{align}
N = O\!\left(\frac{(\sigma^2+\nu_r)(\sigma^2+\nu_r + \varepsilon)}{c^2}\log\!\left(\frac{1}{\delta}\right)\right)\sim O\!\left(\frac{e^{-4r}}{c^2}\log\!\left(\frac{1}{\delta}\right)\right).
\end{align}
In particular, for fixed $c$, the Bell sample complexity remains bounded as $\varepsilon \to 0$, while $\Theta=\{0,\pi/2\}$ homodyne requires $\Omega(v_0^2/\varepsilon^2)$ samples for any fixed constant advantage in the distinguishing task.
\end{corollary}

\vspace{-1em}
\subsection{Sensing electromagnetic fields and cavity phase noise}
\label{sec:prac_Hams}
\vspace{-1em}

We have established that conventional sensing protocols which are restricted to fixed-angle homodyne tomography can be ill-conditioned for a sizable fraction of problem instances, even within simple Gaussian families. For these examples, conventional sensing below the vacuum heterodyne limit is challenging, and Bell QSL yields meaningful quantum advantages. In this section, we exhibit two very common real-world sensing tasks that generate such learning problems.

\vspace{-1em}
\subsubsection{Sensing electromagnetic fields}
\label{sec:D4-EM-transducer}
\vspace{-1em}

Consider a single bosonic mode realized as an $LC$ resonator (or any narrowband cavity mode) used as an
electromagnetic sensor. Let $\hat{\Phi}$ and $\hat{Q}$ denote the canonical flux and charge operators,
$[\hat{\Phi},\hat{Q}]=i\hbar$, so that the free Hamiltonian is
\begin{equation}
H_0=\frac{\hat Q^2}{2C}+\frac{\hat\Phi^2}{2L},\qquad \omega_0=\frac{1}{\sqrt{LC}}.
\end{equation}
An incident classical electromagnetic field induces an effective electromotive force and flux bias,
determined by the sensor geometry (antenna/loop/cavity coupling),
\begin{equation}
V_{\rm eff}(t)\,\equiv\,\oint_{\mathcal C}\mathbf E(\mathbf r,t)\!\cdot\! d\boldsymbol\ell,
\qquad
\Phi_{\rm ext}(t)\,\equiv\,\int_{\mathcal S}\mathbf B(\mathbf r,t)\!\cdot\! d\mathbf A,
\end{equation}
for appropriate pickup curve/surface $(\mathcal C,\mathcal S)$. In lumped-element form, the driven Hamiltonian
can be written (up to an irrelevant $c$-number term) as
\begin{align}
H(t)
&= \frac{\hat Q^2}{2C}+\frac{(\hat\Phi-\Phi_{\rm ext}(t))^2}{2L}-V_{\rm eff}(t)\hat Q \\
&= H_0 -V_{\rm eff}(t)\hat Q - I_{\rm eff}(t)\hat\Phi \ ,
\label{eq:D4-EM-H}
\end{align}
where $I_{\rm eff}(t):=\Phi_{\rm ext}(t)/{L}$. Equation~\eqref{eq:D4-EM-H} illustrates that the field enters through \emph{two} linear
couplings, one ``electric-like'' ($V_{\rm eff}\hat Q$) and one ``magnetic-like'' ($I_{\rm eff}\hat\Phi$). In practice,
$V_{\rm eff}(t)$ and $I_{\rm eff}(t)$ may be treated as random processes, to reflect stochastic signals, clutter, or uncertainty in priors.

We can obtain a familiar Hamiltonian by defining the dimensionless quadratures
\begin{equation}
\hat x := \frac{\hat\Phi}{\sqrt{2}\,\Phi_{\rm zpf}},
\qquad
\hat p := \frac{\hat Q}{\sqrt{2}\,Q_{\rm zpf}},
\end{equation}
which satisfy the canonical commutation relations, and where $\Phi_{\rm zpf}$ and $Q_{\rm zpf}$ are the zero-point fluctuations of the resonator. In the interaction
picture, the signal appears as a linear Hamiltonian of the
form
\begin{equation}
H_{\rm sig}(t)= f_x(t)\,\hat x + f_p(t)\,\hat p,
\qquad
f_x(t)=-\sqrt{2}\,\Phi_{\rm zpf}\,I_{\rm eff}(t),
\qquad
f_p(t)=-\sqrt{2}\,Q_{\rm zpf}\,V_{\rm eff}(t).
\label{eq:D4-EM-linear}
\end{equation}
Therefore, the channel induced by evolution for time $T$ is a pure displacement channel,
\begin{equation}
\mathcal{E}_H^{(T)}(\rho)=\int d^2\alpha\,P_H^{(T)}(\alpha)\,\hat D(\alpha)\rho \hat D(\alpha)^\dagger,
\end{equation}
with complex displacement amplitude given by
\begin{equation}
\tilde\alpha(T)
=\frac{1}{\sqrt{2}}\left(\int_0^T\!dt\,f_p(t)-i\int_0^T\!dt\,f_x(t)\right)
= -Q_{\rm zpf}\!\int_0^T\!dt\,V_{\rm eff}(t)\,+\,i\,\Phi_{\rm zpf}\!\int_0^T\!dt\,I_{\rm eff}(t).
\label{eq:D4-alpha}
\end{equation}
If $(V_{\rm eff},I_{\rm eff})$ are stochastic, then $\tilde\alpha(T)$ is a random variable and
$P_H^{(T)}$ is the pushforward of the classical field distribution under~\eqref{eq:D4-alpha}.

Many physically relevant electromagnetic environments, such as electromagnetic radiation, produce \emph{correlated} electric- and magnetic-like
drives. In the simplest narrowband model, the endpoint displacement $\alpha$ is approximately Gaussian on phase
space with a covariance matrix (in $(x,p)$ coordinates)
\begin{equation}
\Sigma(c)=
\begin{pmatrix}
\sigma_x^2 & c\\
c & \sigma_p^2
\end{pmatrix},
\qquad c=\mathrm{Cov}(x,p),
\label{eq:D4-Sigma}
\end{equation}
where $c$ encodes the \emph{joint} (electric--magnetic) quadrature correlation of the incident field as seen by the
transducer. Hence, we see that the example of learning a Gaussian covariance discussed in Sec. \ref{sec:ex_gaussian_covar} can be meaningfully instantiated as a precision measurement of an external electromagnetic field. As derived in Proposition \ref{prop:prac_bell_cov}, Bell QSL  achieves mean-squared error $\mathbb E[(\widehat c-c)^2]\le \epsilon^2$ with a number of samples
\begin{equation}
\label{eq:D4-Nbell-cov}
N_{\rm Bell}
=O\!\left(\frac{(\sigma_x^2+\nu_r)(\sigma_p^2+\nu_r)+c^2}{\epsilon^2}\right) \sim O(e^{-4r}/\epsilon^2)
\end{equation}
in the shot-noise-limited regime. By contrast, heterodyne incurs the irreducible vacuum penalty $\nu_r\mapsto \nu_{\rm het}:=\tfrac12$, yielding
\begin{equation}
N_{\rm het}
=O\!\left(\frac{(\sigma_x^2+\tfrac12)(\sigma_p^2+\tfrac12)+c^2}{\epsilon^2}\right) \sim O\!\left(1/\epsilon^2\right).
\end{equation}
while homodyne measurements remain non-identifiable. If we once again introduce a symmetry-breaking parameter, denoted here by $\Delta = \sigma_x^2 - \sigma^2_p$, encoding the variance anisotropy of the field, homodyne strategies have their sample complexity limited not by precision-dependence, but by $\Delta^{-2}$ in the near-isotropic regime where $\Delta < \epsilon$. As $\Delta\rightarrow 0$, approaching the homodyne-invisible regime, the resulting restricted-homodyne sample complexity can become parametrically (and even exponentially) larger than the Bell-sensing scaling~\eqref{eq:D4-Nbell-cov}, reproducing the qualitative behavior sketched in Fig.~\ref{fig:deltas_prob_wedge} and the empirical data plotted in Fig.~\ref{fig:em_field}.

Equation~\eqref{eq:D4-EM-H} is a standard effective model for how unknown classical electromagnetic fields enter bosonic transducers. The salient point is the identification of a physically interpretable feature of the classical EM field, namely its $x$--$p$ covariance, that is directly accessible to Bell QSL but ill-conditioned for practical conventional benchmarks.

\vspace{-1em}
\subsubsection{Learning phase-rotation noise in interferometric cavities}
\label{sec:D4-interferometer-noise}
\vspace{-1em}
A common source of error in interferometric cavities is the population of undesirable quadrature modes due to an unknown quadrature drift. Detecting and correcting this drift in real time is often challenging, as the effect can be once again hidden in quadrature correlations and thus difficult to capture over the short timescales of an experiment. Here we argue that Bell QSL provides a speedup that can enable this form of real-time error mitigation. 

We model quadrature noise as a rotation,
\begin{equation}
\label{eq:D4-rot-noise}
\begin{pmatrix}\hat x\\ \hat p\end{pmatrix}
\longmapsto
R(\theta)\begin{pmatrix}\hat x\\ \hat p\end{pmatrix},
\qquad
R(\theta)=
\begin{pmatrix}
\cos\theta & -\sin\theta\\
\sin\theta & \cos\theta
\end{pmatrix}.
\end{equation}
To perform feedback-control, one wants to estimate $\theta$, including its sign, in few shots. However, we note that probe states centered at the origin of phase space will struggle to estimate $\theta$, as the effect of rotating around the origin is most muted for states themselves lying at the origin. Hence, we make a simple, experimentally tractable modification: upon preparing a two-mode squeezed probe with squeezing parameter $r$, imprint a known displacement $\beta$ on the injected arm of the entangled sideband, retain the other arm as an idler phase reference, and perform Bell measurement on the returned signal and the idler. Upon displacing the probe state from the origin of phase space, rotations now change the phase-space structure of our state.

Denoting the commuting Bell quadratures by
\begin{equation}
\label{eq:D4-BM-quads}
\hat x^{\rm BM}=\frac{\hat x_S-\hat x_I}{\sqrt2},
\qquad \hat p^{\rm BM}=\frac{\hat p_S+\hat p_I}{\sqrt2},
\end{equation}
the complex Bell record $\zeta^{(i)}:=x_{\rm BM}^{(i)}+ip_{\rm BM}^{(i)}$ is i.i.d.\ Gaussian with mean $\E[\zeta]=e^{i\theta}\beta$ and (per-quadrature) variance
\begin{equation}
\label{eq:D4-BM-var}
\Var(x_{\rm BM})=\Var(p_{\rm BM})=:v(\theta)=\frac12\!\left(\cosh 2r-\sinh 2r\,\cos\theta\right),
\qquad
\mathrm{Cov}(x_{\rm BM},p_{\rm BM})=0,
\end{equation}
so realistic, small $\theta$, one has $v(\theta)\approx \nu_r:=\tfrac12 e^{-2r}$.

The sample mean is an efficient estimator of $e^{i\theta}\beta$. To isolate $\theta$ from the phase of $\beta$, we write $\beta=|\beta|e^{i\phi_0}$, and may estimate the signed rotation angle directly by the argument of the phase-aligned sample mean,
\begin{equation}
\label{eq:D4-theta-est}
\hat\theta:=\arg\!\left(e^{-i\phi_0}\bar\zeta\right),
\end{equation}
Equivalently, $\hat\theta=\operatorname{atan2}\!\big(\Im(e^{-i\phi_0}\bar\zeta),\Re(e^{-i\phi_0}\bar\zeta)\big)$. As in the
covariance-learning estimators of App.~\ref{sec:ex_gaussian_covar}, this estimator can be constructed from a single classical record resulting from Bell QSL.

Our estimator can be written as $\bar\zeta=e^{i\theta}\beta+w$, with $w$ a circular Gaussian of per-quadrature
variance $v(\theta)/N$. For any $0<\epsilon\le \pi/2$,
\begin{equation}
\label{eq:D4-theta-tail}
\Pr\big(|\hat\theta-\theta|>\epsilon\big) \le \Pr\big(|w|>|\beta|\sin\epsilon\big) =\exp\!\left(-\frac{N|\beta|^2\sin^2\epsilon}{2v(\theta)}\right),
\end{equation}
Therefore, to guarantee $|\hat\theta-\theta|\le \epsilon$ with probability at least $1-\delta$, it suffices to take
\begin{equation}
\label{eq:D4-N-theta}
N \ge \frac{2v(\theta)}{|\beta|^2\sin^2\epsilon}\log\left(\frac{1}{\delta}\right)
= O\left(\frac{v(\theta)}{\epsilon^2|\beta|^2}\log\left(\frac{1}{\delta}\right)\right) \sim \frac{e^{-2r}}{\epsilon^2|\beta|^2}\log\left(\frac{1}{\delta}\right),
\end{equation}

One may also view \eqref{eq:D4-theta-est} as
simultaneously estimating $\sin\theta$ and $\cos\theta$ from one classical record, which can be done with only small constant overhead due to Bell QSL supporting post-hoc estimation. Meanwhile, conventional heterodyne processing incurs sample complexity larger by a factor $\nu_{\rm het}/\nu_r\sim e^{2r}$ and cannot utilize the quantum enhancement provided by squeezing. Homodyne access, even with a pilot displacement, requires multiple, adaptively-chosen angles. In fast detection, such as capturing an incident gravitational wave, re-locking squeezing axes to perform real-time feedback control is not feasible. In this case, a single fixed homodyne quadrature cannot recover the sign of $\theta$, relegating error-mitigation strategies to the vacuum-noise limit.

Operationally, two-mode-squeezed light may be injected into optical sidebands at each analysis frequency, while the idler can be kept outside the main readout band. This enables noninvasive monitoring while still tracking the phase/control fluctuations that rotate the signal quadratures \cite{McCuller2020FDSAdvLIGO,Yap2020EPRFDS,Gould2021EntangledWitnessQNC}.

\subsection{Superpolynomial quantum advantage without entanglement}
\label{sec:entanglement_req}

In this section, we explore a different operational definition of quantum advantage in sensing of classical fields. Our previous results take entanglement to be the core quantum-enhancing resource, treating displacement and squeezing on equal footing. This is an experimentally-motivated convention due to the maturity of Gaussian operations in optical sensing. 
However, we now consider treating squeezing itself as a quantum-enhancing resource, to distinguish the power of purely classical sensing protocols from conventional quantum sensors which use squeezed light. In this setting, we show that it is possible to engineer entanglement-free, squeezing-enhanced strategies that perform exponentially better than any squeezing-free strategy relative to instance size. 

Through this discussion, we make clear that specific notion of quantum advantage must be specified with respect to a concrete class of weaker protocols. Conventional \textit{quantum} sensing uses unentangled Gaussian probe states, and Bell QSL is suited to providing an advantage against this class. However, if one considers the weaker set of entirely \textit{classical} protocols, much simpler sensing is already sufficient to obtain large advantages. Importantly, we note that a learning-theoretic reformulation of sensing opens up the possibility for superpolynomial speedups using standard Gaussian sensing protocols, which is not possible when one only works with the quantum Fisher information perspective.

\vspace{-1em}
\subsubsection{Sensing a fluctuating random background}
\vspace{-1em}

Consider a single bosonic mode and a linear Hamiltonian $H(t)=f(t)\ph$. By Lemma~\ref{lemma:ham_to_rdc}, evolution for a fixed time interval induces a displacement channel that shifts $\xh$ by a random classical displacement $X \in \mathbb{R}$ determined by $f(t)$, i.e.~$\rho\mapsto\disp_q(X)\rho\disp_q(X)^\dagger$ with $\disp_q(X)=\exp(-iX\ph)$.

Now consider the setting in which a sensor is placed in a stochastic background, and experiences many momentum kicks during its active time. This could be instantiated, say, by thermal noise which acts randomly on the sensor several times during the shortest resolvable sensing time. To model this, let us fix an integer $n\ge 1$ and partition the evolution time into $n$ slices, such that in slice $j$, a ``kick'' of magnitude $B>0$ with a random sign $z_j\in\{+1,-1\}$ is applied. The resulting net displacement is
\begin{equation}
    X = B \sum_{j=1}^n z_j .
    \label{eq:false_sep_X_def}
\end{equation}
For this toy example, we impose one of two promises on the kick string $z=(z_1,\dots,z_n)$:
\begin{enumerate}
    \item \textbf{Uniform:} $z_1,\dots,z_n$ are i.i.d.\ uniform on $\{\pm 1\}$.
    \item \textbf{Parity-constrained:} $z$ is uniform over all strings in $\{\pm 1\}^n$ satisfying $\prod_{j=1}^n z_j = +1$.
\end{enumerate}
The task is to decide which promise holds, assuming equal prior probabilities, from $N$ independent uses of the induced displacement channel.

The distinguishing feature of these two distributions is a single Fourier coefficient of the displacement distribution. Defining the usual characteristic function
\begin{equation}
    \chi(k) = \E\left[\exp(i k X)\right], \qquad k\in\mathbb{R},
    \label{eq:false_sep_lambda_def}
\end{equation}
we see that under a uniform ground truth (hypothesis $H_0$), $\chi_0(k)=\cos^n(kB)$. Under the parity-constrained promise (hypothesis $H_1$), a direct computation using the identity $\indicator\{\prod_j z_j=+1\}=(1+\prod_j z_j)/2$ yields
\begin{equation}
    \chi_1(k)=\cos^n(kB) + i^n\sin^n(kB) .
    \label{eq:false_sep_lambda_plus}
\end{equation}
In particular, at the ``Nyquist'' frequency $k_\star=\pi/(2B)$, these expressions are at their most distinguishable:
\begin{equation}
    \chi_0(k_\star)=0,
    \qquad
    \chi_1(k_\star)=i^n .
    \label{eq:false_sep_gap}
\end{equation}
Thus the characteristic-function witnesses differ by a constant-sized Fourier component at a frequency that scales as $1/B$. We choose
\begin{equation}
    B = \frac{1}{\sqrt{n}}
    \quad \Longrightarrow \quad
    k_\star = \frac{\pi}{2}\sqrt{n}.
    \label{eq:false_sep_delta_choice}
\end{equation}

\vspace{-1em}
\subsubsection{Coherent-probe protocols}
\vspace{-1em}

We first define the classical sensing strategies against which we prove a lower bound.
\begin{definition}[Coherent-probe protocol]
\label{def:false_sep_coherent_protocol}
A coherent-probe protocol is an adaptive sensing protocol in which, conditioned on the preceding classical transcript, the joint state immediately before each use of the unknown channel has the form $\sigma_R\otimes\ketbra{\beta}{\beta}$, where $R$ contains arbitrary retained quantum systems and $\beta\in\mathbb{C}$ may depend on the transcript. The protocol may otherwise use arbitrary ancillary systems, quantum operations, quantum storage, and POVMs, including a collective final measurement on all channel outputs. No restriction is placed on $|\beta|$ or on the energy used by the coherent probes, ancillary systems, or measurements.
\end{definition}
Randomized mixtures of coherent probes are included by convexity. This class in fact allows unrestricted quantum processing after interaction with the signal, including potentially extremely high-energy squeezed measurements or non-Gaussian measurements; the only restriction is that the signal must be probed with (arbitrarily high-energy) classical light. This of course includes all practical classical sensing protocols.

\begin{proposition}[Squeezing-enabled exponential speedup]
\label{prop:false_sep_scaling}
Let $B$ be as in \eqref{eq:false_sep_delta_choice}. Consider the parity-discrimination task above.
\begin{enumerate}
    \item Any coherent-probe protocol that distinguishes $H_0$ and $H_1$ with success probability at least $2/3$ requires a number $N$ of channel uses satisfying
    \begin{equation}
        N = \Omega\left(\frac{(e/2)^{n/2}}{\sqrt{n}}\right).
        \label{eq:false_sep_N_unsqueezed}
    \end{equation}
    This lower bound holds regardless of the coherent-probe energy.
    \item If the probe is an $\xh$-squeezed vacuum with squeezing parameter $r=\frac{1}{2}\log n$ there is an entanglement-free protocol which distinguishes $H_0$ and $H_1$ with success probability at least $1-\delta$ using a number $N$ of channel uses satisfying
    \begin{equation}
        N = O\left(\log\left(\frac{1}{\delta}\right)\right).
        \label{eq:false_sep_N_squeezed}
    \end{equation}
\end{enumerate}
\end{proposition}

\begin{proof}
We first prove the coherent-probe lower bound. Let $\mathcal{E}_0$ and $\mathcal{E}_1$ denote the displacement channels under the uniform and parity-constrained hypotheses, respectively. For a coherent probe $\ket{\beta}$, define the corresponding one-query output states by
\begin{equation}
    \rho_b^{(\beta)} = \mathcal{E}_b\left(\ketbra{\beta}{\beta}\right),
    \qquad b\in\{0,1\}.
    \label{eq:false_sep_output_states}
\end{equation}
Let $\Delta_\beta=\rho_1^{(\beta)}-\rho_0^{(\beta)}$. Because Weyl displacements commute up to a phase, $\rho_b^{(\beta)}=\disp(\beta)\rho_b^{(0)}\disp(\beta)^\dagger$, so the trace distance between the two output hypotheses is independent of $\beta$. The two distributions over kick strings can be written as $P_0(z)=2^{-n}$ and $P_1(z)=2^{-n}(1+\prod_{j=1}^n z_j)$. Consequently,
\begin{equation}
    \Delta_\beta = 2^{-n}\sum_{z\in\{\pm1\}^n}\left(\prod_{j=1}^n z_j\right)\disp_q\left(B\sum_{j=1}^n z_j\right)\ketbra{\beta}{\beta}\disp_q\left(B\sum_{j=1}^n z_j\right)^\dagger .
    \label{eq:false_sep_output_difference}
\end{equation}

Let $\psi_\beta(p)=\langle p|\beta\rangle$ be the momentum-space wavefunction of the coherent state. Since $\disp_q(X)=e^{-iX\ph}$, equation \eqref{eq:false_sep_output_difference} gives
\begin{align}
    \langle p|\Delta_\beta|p'\rangle
    &= \psi_\beta(p)\overline{\psi_\beta(p')}2^{-n}\sum_{z\in\{\pm1\}^n}\left(\prod_{j=1}^n z_j\right)\exp\left(iB(p'-p)\sum_{j=1}^n z_j\right) \\
    &= \psi_\beta(p)\overline{\psi_\beta(p')}\left[i\sin\left(B(p'-p)\right)\right]^n .
    \label{eq:false_sep_difference_kernel}
\end{align}
It follows that
\begin{equation}
    \|\Delta_\beta\|_2^2 = \int_{\mathbb{R}^2} dp \ dp'\ |\psi_\beta(p)|^2|\psi_\beta(p')|^2\left|\sin\left(B(p'-p)\right)\right|^{2n}.
    \label{eq:false_sep_hs_norm}
\end{equation}
For every coherent state, $|\psi_\beta(p)|^2$ is Gaussian with variance $1/2$. Thus, if $U$ and $U'$ are independent draws from $|\psi_\beta|^2$, then $G=U'-U$ is distributed as $\mathcal{N}(0,1)$, independently of $\beta$. Using $B=1/\sqrt{n}$ and $|\sin x|\leq |x|$, we obtain
\begin{align}
    \|\Delta_\beta\|_2^2
    &= \E_{G\sim\mathcal{N}(0,1)}\left[\sin^{2n}\left(\frac{G}{\sqrt{n}}\right)\right] \\
    &\leq \frac{1}{n^n}\E_{G\sim\mathcal{N}(0,1)}\left[G^{2n}\right] \\
    &= \frac{(2n-1)!!}{n^n} \\
    &= \frac{(2n)!}{2^n n! n^n}.
    \label{eq:false_sep_hs_bound}
\end{align}

We next convert this Hilbert--Schmidt bound into a trace-norm bound. The net displacement $B\sum_j z_j$ assumes only the $n+1$ values $B(n-2m)$ for $m=0,\dots,n$. Grouping equation \eqref{eq:false_sep_output_difference} by the number $m$ of negative signs gives
\begin{equation}
    \Delta_\beta = 2^{-n}\sum_{m=0}^n(-1)^m\binom{n}{m}\disp_q\left(B(n-2m)\right)\ketbra{\beta}{\beta}\disp_q\left(B(n-2m)\right)^\dagger .
\end{equation}
Hence the range of $\Delta_\beta$ is contained in the span of at most $n+1$ coherent states, and therefore $\operatorname{rank}(\Delta_\beta)\leq n+1$. The Schatten-norm inequality $\|A\|_1\leq\sqrt{\operatorname{rank}(A)}\|A\|_2$ then implies
\begin{equation}
    \|\Delta_\beta\|_1 \leq \sqrt{n+1}\left(\frac{(2n)!}{2^n n! n^n}\right)^{1/2}.
    \label{eq:false_sep_trace_bound}
\end{equation}
This bound is uniform over $\beta$. In particular, increasing the coherent amplitude, and hence the coherent-probe energy, does not improve the one-query distinguishability.

Define the optimal one-query trace distance over coherent probes by
\begin{equation}
    d_n = \sup_{\beta\in\mathbb{C}}\frac{1}{2}\left\|\rho_1^{(\beta)}-\rho_0^{(\beta)}\right\|_1.
\end{equation}
Equation \eqref{eq:false_sep_trace_bound} gives
\begin{equation}
    d_n \leq \frac{\sqrt{n+1}}{2}\left(\frac{(2n)!}{2^n n! n^n}\right)^{1/2}.
    \label{eq:false_sep_single_query_distance}
\end{equation}
By the standard hybrid argument \cite{chen2021exponentialseparationslearningquantum, cotler2026noisyquantumlearningtheory}, this quantity controls arbitrary adaptive coherent-probe protocols. Let $\Omega_b^{(N)}$ be the final classical transcript resulting from an arbitrary adaptive $N$-query protocol under hypothesis $H_b$. Then,
\begin{equation}
    d_{TV}\left(\Omega_1^{(N)}-\Omega_0^{(N)}\right) \leq N d_n.
    \label{eq:false_sep_adaptive_hybrid}
\end{equation}
By Le Cam's two-point Lemma \ref{lem:prac_lecam}, the success probability of any final POVM with equal priors satisfies
\begin{equation}
    p_{\mathrm{succ}} \leq \frac{1}{2}+\frac{1}{2}d_{TV}\left(\Omega_1^{(N)}-\Omega_0^{(N)}\right) \leq \frac{1}{2}+\frac{N d_n}{2}.
\end{equation}
Therefore success probability at least $2/3$ requires
\begin{equation}
    N \geq \frac{1}{3d_n} \geq \frac{2}{3\sqrt{n+1}}\left(\frac{2^n n! n^n}{(2n)!}\right)^{1/2}.
    \label{eq:false_sep_exact_query_lower}
\end{equation}
Finally, Stirling's approximation gives $\frac{(2n)!}{2^n n! n^n}=\Theta((2/e)^n)$, which proves equation \eqref{eq:false_sep_N_unsqueezed}.

We now prove the squeezed-probe upper bound. Prepare a single-mode probe state $\rho$, apply one random displacement $X$, and perform homodyne measurement of $\xh$. If the probe has $\xh$-quadrature variance $V_x=\Var_\rho(\xh)$, then the homodyne outcome $\zeta$ admits the additive-noise representation
\begin{equation}
    \zeta = X + \nu,
    \qquad
    \nu\sim\mathcal{N}(0,V_x),
    \label{eq:false_sep_additive_noise}
\end{equation}
where $\nu$ is independent of $X$. For the bounded random variable $Y=\exp(i k_\star\zeta)$, which satisfies $|Y|=1$, we have
\begin{equation}
    \E[Y] = \E[\exp(i k_\star X)]\E[\exp(i k_\star\nu)] = \chi(k_\star)\exp\left(-\frac{V_x k_\star^2}{2}\right).
    \label{eq:false_sep_blur_factor}
\end{equation}
By \eqref{eq:false_sep_gap}, under the uniform promise $\E[Y]=0$, while under the parity-constrained promise
\begin{equation}
    \E[Y] = i^n\exp\left(-\frac{V_x k_\star^2}{2}\right).
    \label{eq:false_sep_mean_under_H1}
\end{equation}
Hence the distinguishing bias is
\begin{equation}
    \mu = \left|\E[Y]\right| = \exp\left(-\frac{V_x k_\star^2}{2}\right).
    \label{eq:false_sep_mu_def}
\end{equation}
Collect i.i.d.\ samples $Y_1,\dots,Y_N$ and compute the empirical mean $\overline{Y}=N^{-1}\sum_{i=1}^N Y_i$, deciding ``parity-constrained'' if $|\overline{Y}|>\mu/2$ and ``uniform'' otherwise. Hoeffding's inequality applied to $\Re(Y)$ and $\Im(Y)$ and a union bound imply that this test succeeds with probability at least $1-\delta$ whenever
\begin{equation}
    N \geq O\left(\frac{1}{\mu^2}\log\left(\frac{1}{\delta}\right)\right).
    \label{eq:false_sep_N_hoeffding}
\end{equation}
For an $\xh$-squeezed vacuum, $V_x=\frac{1}{2}e^{-2r}$. Choosing $r=\frac{1}{2}\log n$ gives $V_x=1/(2n)$ and $\mu=\exp(-\pi^2/16)$, which is a constant independent of $n$. Substituting into \eqref{eq:false_sep_N_hoeffding} proves equation \eqref{eq:false_sep_N_squeezed}.
\end{proof}

Proposition~\ref{prop:false_sep_scaling} establishes an exponential separation between a single squeezed probe and the full class of coherent-probe protocols. The coherent-probe lower bound is independent of the amplitude of the chosen coherent probe state and permits arbitrary adaptive quantum processing and measurements, so no increase in coherent-state energy can reproduce the squeezed-probe performance on this instance. As such, squeezing is a quantum resource that enables exponential sensing advantages even in the absence of entanglement.

\vspace{-1em}
\subsubsection{Superpolynomial advantages in sensing}
\vspace{-1em}

The canonical resource normalization in bosonic sensing fixes the probe energy, typically quantified by the mean photon number $\bar{n}=\E[\hat{a}^\dagger \hat{a}]$ (or equivalently total energy across modes), and compares strategies at equal $\bar{n}$ \cite{Caves1981,Giovannetti2011}. For a single-mode squeezed vacuum, the mean photon number is
\begin{equation}
    \bar{n}(r) = \sinh^2(r),
    \label{eq:false_sep_photon_number}
\end{equation}
so the choice $r=\frac{1}{2}\log n$ used in Proposition~\ref{prop:false_sep_scaling} corresponds to energy scaling as $\bar{n}=\Theta(n)$ \cite{Weedbrook2012}. The full class of coherent-probe protocols, on the other hand, can access arbitrarily large energies. 

As optical sensing is usually limited to Gaussian operations which can always be decomposed into displacements, squeezing, and passive linear optics, Proposition \ref{prop:false_sep_scaling} provides an important operational refinement that is not usually transparent when one works in the quantum Fisher information picture of bosonic sensing. The QFI picture suggests that an energy budget should be balanced between between displacement and squeezing resources to obtain an at-most polynomial gain; however, our result shows that by evaluating sensing efficiency using sample complexity, \textit{non-classical} energy budget emerges as a unique resource that can enable superpolynomial advantages not apparent from the Fisher information perspective. In this setting, no amount of energy allocated to classical resources can match the performance of a quantum probe.

\end{document}